\documentclass[envcountsect,orivec]{llncs} 
\usepackage{etex} 
\usepackage[]{graphicx}
\newif\ifignore 
\ignorefalse
\newcommand{\auxproof}[1]{
  \ifignore\mbox{}\newline
  \textbf{BEGIN: AUX-PROOF} \dotfill\newline
  {#1}\mbox{}\newline
  \textbf{END: AUX-PROOF}\dotfill\newline
  \fi}

\renewcommand{\marginpar}[1]{}    

\usepackage{times}

\usepackage{theorem}
\usepackage{fancybox,amssymb,amstext,amsmath,stmaryrd,wasysym,cite,proof,mathtools,multirow,stackrel}
\SetSymbolFont{stmry}{bold}{U}{stmry}{m}{n}
\SetSymbolFont{wasy}{bold}{U}{wasy}{m}{n}
\usepackage[pdftex,all]{xy}
\usepackage{xspace}
\allowdisplaybreaks[1] 

\xyoption{v2}
\xyoption{curve}
\xyoption{2cell}
\SelectTips{cm}{}  
\UseAllTwocells
\SilentMatrices

\newdir{ >}{{}*!/-8pt/@{>}}  
\newdir{|>}{%
  !/1.6pt/@{|}*:(1,-.2)@^{>}*:(1,+.2)@_{>}}
\newdir{pb}{:(1,-1)@^{|-}}
\def\pb#1{\save[]+<20 pt,0 pt>:a(#1)\ar@{pb{}}[]\restore}

\usepackage{algorithm}
\usepackage[noend]{algpseudocode} 
\usepackage{etoolbox}

\makeatletter
\newcommand*{\algrule}[1][\algorithmicindent]{%
   \makebox[#1][l]{%
       \hspace*{.2em}
       \vrule height .75\baselineskip depth .25\baselineskip
   }
}

\newcount\ALG@printindent@tempcnta
\def\ALG@printindent{%
    \ifnum \theALG@nested>0
    \ifx\ALG@text\ALG@x@notext
    \else
    \unskip
    \ALG@printindent@tempcnta=1
    \loop
    \algrule[\csname ALG@ind@\the\ALG@printindent@tempcnta\endcsname]%
    \advance \ALG@printindent@tempcnta 1
    \ifnum \ALG@printindent@tempcnta<\numexpr\theALG@nested+1\relax
    \repeat
    \fi
    \fi
}
\patchcmd{\ALG@doentity}{\noindent\hskip\ALG@tlm}{\ALG@printindent}{}{\errmessage{failed to patch}}
\patchcmd{\ALG@doentity}{\item[]\nointerlineskip}{}{}{} 
\makeatother

\usepackage{wrapfig}
\theorembodyfont{\itshape}
\newtheorem{mytheorem}{Theorem}[section]
\newtheorem{mylemma}[mytheorem]{Lemma}
\newtheorem{myproposition}[mytheorem]{Proposition}

\theorembodyfont{\rmfamily}

\newtheorem{myremark}[mytheorem]{Remark}

\newtheorem{mydefinition}[mytheorem]{Definition}
\newtheorem{myassumption}[mytheorem]{Assumption}

\spnewtheorem*{myproof}{Proof}{\itshape}{\rmfamily}

\def\myqed{\qed}

\usepackage{xcolor}
\definecolor{dgreen}{rgb}{0, .6, 0}

\usepackage{pifont}
\newcommand{\cmark}{\ding{51}}%
\newcommand{\xmark}{\ding{55}}%

\newcommand{\ttrue}{\mathrm{t{\kern-1.5pt}t}}
\newcommand{\ffalse}{\mathrm{f{\kern-1.5pt}f}}

\newif\iftikzgnuplot
\tikzgnuplottrue 
\usepackage[hyphens]{url}
\iftikzgnuplot
\usepackage{gnuplot-lua-tikz}
\fi
\usepackage{pgfplotstable} 
\pgfplotsset{compat=1.12}
\usepackage{booktabs}

\usepackage{enumitem}
\setlistdepth{20}
\renewlist{itemize}{itemize}{20}
\setlist[itemize]{label=\textbullet}

\newcommand{\integer}{\mathrm{int}}
\newcommand{\fractional}{\mathrm{frac}}
\newcommand{\Rev}{\mathit{Rev}}
\newcommand{\Opt}{\mathit{Opt}}
\newcommand{\Run}{\mathit{Runs}}
\newcommand{\pref}{\mathrm{pref}}
\newcommand{\str}{\mathit{str}}
\newcommand{\pat}{\mathit{pat}}
\newcommand{\reset}{\mathrm{reset}}
\newcommand{\eval}{\mathrm{eval}}
\newcommand{\Init}{\mathit{Immd}}
\newcommand{\solConstr}{\mathrm{solConstr}}
\newcommand{\CurrConf}{\mathit{CurrConf}}
\newcommand{\PrevConf}{\mathit{PrevConf}}
\newcommand{\Conf}{\mathit{Conf}}

\newcommand{\rhoEmpty}{\rho_{\emptyset}}
\newcommand{\noredundancy}[1]{#1^{\mathrm{r}}}

\newcommand{\regionstate}[4][]{
 \node[state,region] (#2) [#1] {
 \begin{tabular}{c}
  #3\\
  #4\\
 \end{tabular}
 }}

\usepackage{tikz}
\usetikzlibrary{automata,positioning,matrix,shapes.callouts}
\tikzset{
region/.style={
rectangle,
rounded corners,
draw=black,very thick
},
accepting/.style={double distance=2pt}
}
\newcommand{\drawexectable}[1]{
 \pgfplotstabletypeset[
 multicolumn names, 
 display columns/0/.style={
 column name=$|w|$, 
 fixed,fixed zerofill,precision=0,
 },
 display columns/1/.style={
 sci,
 sci zerofill,
 column name=Naive Time\,(ms.)},
 display columns/2/.style={
 sci,
 sci zerofill,
 column name=BM Time\,(ms.)},
 display columns/3/.style={
 sci,
 sci zerofill,
 column name=BM + Preproc. Time\,(ms.)},
 display columns/4/.style={
 fixed,fixed zerofill,precision=0,
 column name=$|Z|$},
 every head row/.style={
 before row={\toprule}, 
 after row={\midrule} 
 },
 every last row/.style={after row=\bottomrule}, 
 ]{#1} 
 }

 \title{A Boyer-Moore Type Algorithm\\ for Timed Pattern Matching
}
 \author{
Masaki Waga
\inst{1}
  \and
Takumi Akazaki
\inst{1,2}
  \and
Ichiro Hasuo
\inst{1}
 }
 \institute{
     University of Tokyo, Tokyo, Japan
     \and
   JSPS Research Fellow, Tokyo, Japan
}

\begin{document}

\maketitle

\begin{abstract}
The \emph{timed pattern matching} problem is formulated 
by Ulus et al.\ and has been actively studied since, with its evident application in monitoring real-time
 systems. The problem takes as input a \emph{timed word/signal} and a
 \emph{timed pattern} (specified either by a \emph{timed regular
 expression} or by a \emph{timed automaton}); and it returns the set of
 those intervals for which the given timed word, when restricted to the
 interval, matches the given pattern.  We contribute a
 \emph{Boyer-Moore} type optimization in timed pattern matching, relying
 on the classic Boyer-Moore string matching algorithm and its extension
 to (untimed) pattern matching by Watson and Watson. We assess its
 effect through experiments; for some problem instances our Boyer-Moore type
 optimization achieves speed-up by two times, indicating its potential
 in real-world monitoring tasks where data sets tend to be massive.
\end{abstract}

\section{Introduction}
 Importance of systems' \emph{real-time} properties is ever growing,
 with rapidly diversifying applications of computer
 systems---cyber-physical systems, health-care systems, automated
 trading, etc.---being increasingly pervasive in every human activity.
 For real-time properties, besides classic problems in theoretical
 computer science such as \emph{verification} and \emph{synthesis}, the
 problem of \emph{monitoring} already turns out to be
 challenging. Monitoring asks, given an execution log and a
 specification, whether the log satisfies the specification; sometimes
 we are furthermore interested in \emph{which segment} of the log
 satisfies/violates the specification. In practical deployment scenarios
 where we would deal with a number of very long logs, finding matching 
 segments in a computationally tractable manner is therefore a pressing
 yet challenging matter.

 In this context, inspired by the problems of \emph{string} and
 \emph{pattern matching} of long research histories, Ulus et al.\
 recently formulated the problem of \emph{timed pattern matching}~\cite{Ulus2014}. 
 In their formalization, the problem takes as input a \emph{timed signal} 
 $w$ (values that change over the continuous notion of time) and a
 \emph{timed regular expression (TRE)} 
 $\mathcal{R}$
(a real-time extension of regular expressions); and it returns the
 \emph{match set} $\mathcal{M}(w,\mathcal{R})=\{(t,t')\mid t<t',
 w|_{(t,t')}\in L(\mathcal{R})\}$, where $w|_{(t,t')}$ is the
 restriction of $w$ to the time interval $(t,t')$ and $L(\mathcal{R})$
 is the set of signals that match $\mathcal{R}$.

 Since its formulation timed pattern matching has been actively
 studied. The first offline algorithm is introduced in~\cite{Ulus2014};
 its application in conditional performance evaluation is pursued
 in~\cite{Ferrere2015}; and in~\cite{Ulus2016} an online algorithm is
 introduced based on \emph{Brzozowski derivatives}. Underlying these
 developments is the fundamental observation~\cite{Ulus2014} that 
 the match set $\mathcal{M}(w,\mathcal{R})$---an uncountable subset of
 $\mathbb{R}_{\ge 0}^{2}$---allows a finitary symbolic representation by 
 inequalities.

\noindent\textbf{Contributions}\quad In this paper we are concerned with
\emph{efficiency} in timed pattern matching, motivated by our
collaboration with the automotive industry on various light-weight
verification techniques. Towards that goal we introduce optimization
that extends the classic \emph{Boyer-Moore} algorithm for string
matching (finding a pattern string $\pat$ in a given word
$w$). Specifically we rely on the extension of the latter to
\emph{pattern matching} (finding subwords of $w$ that is accepted by an
NFA $\mathcal{A}$) by Watson \& Watson~\cite{Watson2003}, and introduce
its \emph{timed} extension.  

We evaluate its efficiency through a series of experiments; in some
cases (including an automotive example) our Boyer-Moore type algorithm
outperforms a naive algorithm (without the optimization) by twice.  This
constant speed-up may be uninteresting from the complexity theory point
of view. However, given that in real-world monitoring scenarios the
input set of words $w$ can be literally \emph{big data},\footnote{For
example, in~\cite{DBLP:conf/rv/ColomboP12}, a payment transaction record
of 300K users over almost a year is monitored---against various
properties, some of them timed and others not---and they report the task
took hundreds of hours.}  halving the processing time is a substantial
benefit, we believe.

Our technical contributions are concretely as follows: 1) a (naive)
algorithm for timed pattern matching (\S{}\ref{sec:naiveAlgo}); 2) its
online variant (\S{}\ref{sec:naiveAlgo}); 3) a proof that the match set
allows a finitary presentation
(Thm.~\ref{thm:terminationAndCorrectnessOfNaiveAlg}), much like
in~\cite{Ulus2014}; and 4) an algorithm with Boyer-Moore type
optimization (\S\ref{201300_31Jan16}).
Throughout the paper we let (timed) patterns expressed as \emph{timed
automata (TA)}, unlike  timed regular expressions (TRE)
in~\cite{Ulus2014,Ferrere2015,Ulus2016}. Besides TA is known to be
strictly more expressive than TRE (see~\cite{Herrmann1999} and also
Case~2 of~\S{}\ref{222310_24Jan16}), our principal reason for choosing TA is so that the Boyer-Moore type pattern matching algorithm
in~\cite{Watson2003} smoothly extends.

\noindent \textbf{Related and Future Work}\quad 
The context of the current work
 is \emph{run-time verification} and \emph{monitoring} of cyber-physical
 systems, a field of growing research activities (see e.g.\
 recent~\cite{DBLP:conf/rv/GeistRS14, DBLP:conf/rv/KaneCDK15}).
One promising application  is in
\emph{conditional quantitative analysis}~\cite{Ferrere2015}, e.g.\ of 
fuel consumption of a car during acceleration, from a large data set of
driving record. Here our results can be used to
efficiently isolate the acceleration
phases.


Aside from timed automata and TREs, \emph{metric} and \emph{signal temporal logics (MTL/STL)}
are commonly used for
specifying  continuous-time signals. Monitoring against
these formalisms has been actively studied,
too~\cite{DBLP:conf/rv/HoOW14, DBLP:conf/rv/DeshmukhDGJJS15,
DBLP:conf/rv/DokhanchiHF14,DBLP:conf/cav/DonzeFM13
}. 
It is known that an MTL formula can be translated to
a timed alternating automaton~\cite{DBLP:journals/corr/abs-cs-0702120}. 
MTL/STL tend to be used against ``smooth'' signals whose changes are
continuous, however, and it is not clear how our current results (on
timed-stamped finite words) would apply to such a situation. One
possible
practical approach would be to quantize continuous-time signals.

Being \emph{online}---to process a long timed word $w$ one can already
 start with its prefix---is obviously a big advantage in monitoring
algorithms. In~\cite{Ulus2016} an online   timed pattern
matching algorithm (where a specification is a TRE) is given, relying on the timed
 extension of  \emph{Brzozowski derivative}.  We
shall aim at an online version of our Boyer-Moore type algorithm
(our online algorithm in~\S{}\ref{sec:naiveAlgo} is without the
Boyer-Moore type optimization), although it seems hard already for the
prototype problem of string matching.

It was suggested by multiple reviewers that use of \emph{zone automata}
can further enhance our Boyer-Moore type algorithm for timed pattern matching. See
Rem.~\ref{rem:zoneAutom}.


\noindent \textbf{Organization of the Paper}\quad 
We introduce necessary backgrounds in~\S{}\ref{sec:prelim}, on: the
basic theory of timed automata, and  the previous Boyer-Moore
algorithms (for string matching, and the one
in~\cite{Watson2003} for (untimed) pattern matching). The latter will
pave the way to our main contribution of the timed Boyer-Moore algorithm. We formulate the
timed pattern matching problem in~\S{}\ref{103806_9Apr16}; and a (naive)
algorithm is presented in~\S{}\ref{sec:naiveAlgo} together with its
online variant. In~\S{}\ref{201300_31Jan16} a Boyer-Moore algorithm for
timed pattern matching is described, drawing intuitions from the untimed
one and emphasizing  where are the
differences. In~\S{}\ref{222310_24Jan16} we present the experiment
results; they indicate the potential of the proposed algorithm in
real-world monitoring applications.

Most proofs are deferred to the appendix due to lack of space.

\auxproof{
\subsection{Timed Automata and Timed Regular Expressions}
Two representations of running information of real-time systems are
introduced in~\cite{Asarin2002}: they are \emph{timed words} and
\emph{timed signals}.
A timed word is a sequence of events with timestamps.
For a finite set of events $\Sigma$, a timed word $(\overline{a},\overline{\tau})$ is an
element of $(\Sigma \times \mathbb{R}_{> 0})^*$.
Fig.~\ref{003016_31Jan16} is an example of a timed word.
The domain of the values of timestamps is the set of positive reals.
In this sense, timed words represent time as a dense set.
A timed signal is a function from a point of time to a value.
The domain of timed signal is the set of positive reals.
Timed signals describe states of each point of time.
For a finite set of states $\Sigma$, a timed signal $p$ is a function
$p : \mathbb{R}_{\geq 0} \to \Sigma$.
The state of each point of time can change arbitrarily.
Fig.~\ref{003025_31Jan16} is an example of a timed signal.
Both representations are dense-time models, in which time is a dense
set. 

\begin{figure}[h]
 \begin{minipage}[b]{.5\linewidth}
 \centering
 \begin{tikzpicture} 
  \draw [thick, -stealth](-0.5,0)--(5,0) node [anchor=north]{$t$};
  \draw (0,0.1) -- (0,-0.1) node [anchor=north]{$0$};

  \draw (1,0.1) node [anchor=south]{$a$} -- (1,-0.1) node [anchor=north]{$1.0$};
  \draw (1.9,0.1) node[anchor=south]{$b$} -- (1.9,-0.1) node[anchor=north]{$1.9$};
  \draw (3.4,0.1) node[anchor=south]{$a$} -- (3.4,-0.1) node[anchor=north]{$3.4$};
 \end{tikzpicture}
 \label{003016_31Jan16}
 \end{minipage}
 \begin{minipage}[b]{.5\linewidth}
  \centering
  \begin{tikzpicture} 
   \draw [thick, -stealth](-0.5,0)--(5,0) node [anchor=north]{$t$};
   \draw [thick, -stealth](0,-0.5)--(0,2) node [anchor=east]{$\Sigma$};
   \draw [fill=pink](0,0) rectangle (1,1);
   \draw [fill=cyan](1,0) rectangle (1.9,1);
   \draw [fill=pink](1.9,0) rectangle (4.9,1);
   \node at (0.5,0.5) {$a$};
   \node at (1.45,0.5) {$b$};
   \node at (3.4,0.5) {$a$};

   \draw (1,0) -- (1,-0.1) node [anchor=north]{$1.0$};
   \draw (1.9,0) -- (1.9,-0.1) node[anchor=north]{$1.9$};
  \end{tikzpicture}
  \label{003025_31Jan16}
 \end{minipage}
 \caption{A timed word and a timed signal}
\end{figure}

\begin{wrapfigure}[8]{r}{0pt}
 \begin{tikzpicture}[shorten >=1pt,node distance=2.3cm,on grid,auto] 
  \node[state,initial] (s_0)   {$s_0$}; 
  \node[state] (s_1) [right=of s_0] {$s_1$}; 
  \node[state,accepting] (s_2) [right=of s_1] {$s_2$}; 
  \path[->] 
  (s_0) edge  [above] node {
  \begin{tabular}{c}
   $a,t=1$\\
   $/t:=0$\\
  \end{tabular}
  } (s_1)
  (s_1) edge [loop above] node  {
  \begin{tabular}{c}
   $b,0<t<1$\\
   $/t:=0$
  \end{tabular}
  } (s_1)
  (s_1) edge  [above] node {$a,t>1$} (s_2);
 \end{tikzpicture}
 \caption{A timed automaton.}
 \label{235528_30Jan16}
\end{wrapfigure}
\emph{Timed automata}~\cite{Alur1994} can represent real-time systems.
Timed automata have the concept of time as clock variables.
Fig.~\ref{235528_30Jan16} is an example of a timed automaton.
Each edge of timed automata has a clock constraint $\delta$ and a set
$\lambda$ of reset clock variables.
We can move through the edge when the clock variables satisfy the clock
constraint.
When we move through an edge, the values of clock variables in $\lambda$
are reset.
In timed automata, the value of each clock variable increases at the
same rate.

Consider a timed word in Fig.~\ref{003016_31Jan16} and a timed automaton
in Fig.~\ref{235528_30Jan16}.
First we are in $s_0$ and the value of $t$ is 0.
At 1.0, we read $a$ and go to $s_1$.
Since $t$ is reset, the value of $t$ is 0.
At 1.9, we read $b$ and go to $s_1$.
The variable $t$ is reset again.
At 3.4, we read $a$ and go to $s_2$.
Since $s_2$ is an accepting state, this timed word is accepted.
}

\auxproof{
\paragraph{Timed Pattern Matching}
The \emph{timed pattern matching} is a problem introduced
in~\cite{Ulus2014}.
An input timed signal and a timed regular expression are given, and we
answer the set of all the intervals matching the given timed
expression. 
Since the domain of timed signals is a dense set, the answer of the
timed pattern matching can be an infinite set.
However, it can be represented as a finite union of \emph{zones}.
A zone is a special case of a convex polyhedron.
The offline pattern matching algorithm is introduced in~\cite{Ulus2014}
and~\cite{Ferrere2015}, and the online algorithm is introduced
in~\cite{Ulus2016}.

However, to the best of our knowledge, the timed pattern matching with
timed automata is not studied yet.
Since timed automata have strictly more expressive power than timed
regular expressions, algorithms solving timed pattern matching problem
with timed automata is also applicable to timed regular expressions, but
not conversely.
}

\auxproof{
\subsection{Our Contributions}
In this paper, we construct two algorithms to solve the timed pattern
matching problem with automata; a \emph{naive algorithm} and a
\emph{timed generalized Boyer-Moore type algorithm}.
We also implement these algorithms and compare their performance.
}

\auxproof{\paragraph{A Naive Algorithm for Timed Pattern Matching Problem}
Given an input timed word and a timed automaton, our algorithm answers
the set of all the intervals matching the given timed automaton. 
The naive algorithm solves the time constraints for any start and end
indices of the given timed word. 
Let $w = (\overline{a},\overline{\tau})$ be the input timed word and $i$ be the index
that points current beginning character.
For each $i \leq j \leq |w|$, we read $w_j$ step by step.
We calculate the reachable states, the possible values of clock
variables, and the set of possible start time $T \subseteq
[\tau_{i-1},\tau_i)$ after reading $w_j$ by examining the clock
constraints of the edges. 

For some technical reasons, we use timed automata such that any edge to
an accepting state is labeled with an extra character \$ to denote the
end of events. 
We do not take sub-words but insert the terminal character after the
last event.
More precisely, given a timed word $w = (\overline{a},\overline{\tau})$ and a timed
automaton $\mathcal{A}$, we search the set of pair $(t,t')$ such that
$w|_{(t,t')}$ is accepted by $\mathcal{A}$.
A timed word segment $w|_{(t,t')} = (\overline{a'},\overline{\tau'})$ is defined as
follows, where $i$ is the least index such that $t < \tau_i$ and $j$ is
the greatest index such that $\tau_j < t'$.
}

\auxproof{
\begin{wrapfigure}[8]{r}{0pt}
 \begin{tikzpicture}[shorten >=1pt,node distance=2.3cm,on grid,auto]
  \node[state,initial] (s_0)   {$s_0$}; 
  \node[state] (s_1) [right=of s_0] {$s_1$}; 
  \node[state,accepting] (s_2) [right=of s_1] {$s_2$}; 
  \path[->] 
  (s_0) edge  [above] node {
  \begin{tabular}{c}
   $a,0<t<1$\\
   $/t:=0$\\
  \end{tabular}
  } (s_1)
  (s_1) edge [loop above] node  {
  \begin{tabular}{c}
   $b,0<t<1$\\
   $/t:=0$
  \end{tabular}
  } (s_1)
  (s_1) edge  [above] node {$\$,t=1$} (s_2);  
 \end{tikzpicture}
 \caption{A timed automaton.}
 \label{114312_31Jan16}
\end{wrapfigure}

Assume the input timed word is $w = ((1.5,2.2,3.4,4.0,4.6),(a,b,a,b,b))$
and the timed automaton is as in Fig.~\ref{114312_31Jan16}.
In our naive algorithm, we read a timed word from the end in order to
extend to the timed generalized Boyer-Moore type algorithm later.
While the first character we read is $b$, we cannot move from $s_0$.
When $i = 3$, we read $a$ at 3.4.
Considering the clock constraint, we can fix the domain of $t$ to be
$(2.4,3.4)$.
After reading $(a_3,\tau_3)$, we try to insert \$ at 4.4, but we
cannot do it because $\tau_4 = 4.0$ is greater than 4.4.
We try to insert \$ at 5.0 after reading $(a_4,\tau_4)$, but
we cannot do it again. 
After considering the rest part in the same manner, the answer is
$\{(t,t') \mid t \in (0.5,1.5), t' = 3.2\} \cup \{(t,t') \mid t \in
(2.4,3.4), t' = 5.6)\}$.

The algorithm introduced in~\cite{Ulus2014} is a timed regular
expression based, inductive algorithm. 
Our algorithm is a timed automata based, essentially sequential
algorithm. 
We also construct an online algorithm with a few modifications.
}

\auxproof{
\paragraph{A Boyer-Moore Type Algorithm for Timed Pattern Matching}

Inspired by the original Boyer-Moore algorithm for
strings~\cite{Boyer1977} and the generalized Boyer-Moore type algorithm
for strings and regular languages~\cite{Watson2003}, we construct another
algorithm that avoids reading unnecessary characters. 
The timed generalized Boyer-Moore type algorithm is equipped with the
skip value functions. 
The skip value functions examine whether the already read substring can
be a substring of some accepting words. 
We skip to the next point that can be the beginning point of some
matching segments.
In our algorithm, we employs mainly two skip value functions: one of the
skip value functions is similar to that of the Boyer-Moore algorithm for
regular languages (we call them $\Delta_2$ in
section~\ref{201300_31Jan16}); and the other function is based on the
skip
value function in the Boyer-Moore algorithm for strings ($\Delta_1$ in
section~\ref{201300_31Jan16}).
We mainly use the former skip value function.
Adaptation of $\Delta_1$ to the (not timed) pattern matching problem is
also the novelty of our work. 

Again, assume the input timed word is $w =
((1.5,2.2,3.4,4.0,4.6),(a,b,a,b,b))$ and the timed automaton is as in
Fig.~\ref{114312_31Jan16}. 
After examining any sub-words when $i = 3$, we can move to $i = 1$.
It is because the event $a_3 = a$ does not appear except the first
event in the language of $\mathcal{A}$.

Our experiments show that our Boyer-Moore type optimization can make timed
pattern matching at most twice faster. While this constant speed-up is certainly
uninteresting from the complexity theory point of view, we believe 
it helps some real-world monitoring tasks. The latter tend to deal with literally
\emph{big data}: for example, in~\cite{DBLP:conf/rv/ColomboP12} a
 payment
transaction record of 300K users over almost a year is monitored---against various
 properties, some of them timed and others not---and they report the task took
hundreds of hours. Halving the processing time would then be a substantial benefit.
}

\auxproof{
\paragraph{Implementations of the Algorithms}
We show the practicality of our algorithms by implementing and testing
both the naive algorithm and the timed generalized Boyer-Moore type
algorithm. 
We test with some benchmarks and monitor the output torque of the
simulation of the engine of a car. 
We compare the performance of both implementations. 
}

\auxproof{
\subsection{Organization of the Paper}

The remainder of this paper is organized as follows.
In Subsection~\ref{221705_24Jan16}, we show formalisms used in this paper
such as timed words, timed automata, and region automata.
In Subsection~\ref{subsec:stringMatchingAndOriginalBoyerMoore}, we
review the ``generalized'' Boyer-Moore type algorithm
in~\cite{Watson2003} for pattern matching.
We construct an naive algorithm in Section~\ref{103806_9Apr16}, and
an Boyer-Moore type algorithm in Section~\ref{201300_31Jan16} for timed
pattern matching.
In Section~\ref{222310_24Jan16} we do some experiments, and in
Section~\ref{sec:conclusions} we discuss our contributions and future
work. 

Most proofs are deferred to the appendix due to lack of space. 

}

\section{Preliminaries}\label{sec:prelim}
\subsection{Timed Automata}
\label{221705_24Jan16}
Here we follow~\cite{Alur1994,Asarin2002}, possibly with a fix to accept finite
words instead of infinite.
For a sequence 
 $\overline{s}=s_1 s_2 \dotsc s_{n}$ we write $|\overline{s}|=n$; and 
for $i,j$ such that  $1 \le i \leq j \leq |s|$,
$\overline{s} (i)$ denotes the element $s_i$ and $\overline{s}(i,j)$ denotes the subsequence
$s_i s_{i+1} \dotsc s_j$.



\vspace{.3em}
\noindent
\begin{minipage}{\textwidth}
 \begin{mydefinition}[timed word]\label{def:timedWord}
 A \emph{timed word} over an alphabet $\Sigma$ is 
 an element of $(\Sigma \times \mathbb{R}_{> 0})^*$---which is denoted by
 $(\overline{a},\overline{\tau})$ using $\overline{a}\in\Sigma^{*}$,
 $\overline{\tau}\in (\mathbb{R}_{> 0})^{*}$ via the  embedding 
 $(\Sigma \times \mathbb{R}_{>0})^* \hookrightarrow \Sigma^* \times (\mathbb{R}_{>0})^*$---such that
 for any
 $i\in [1,|\overline{\tau}|-1]$ we have 
 $0 < \tau_i < \tau_{i+1}$.

 Let $(\overline{a},\overline{\tau})$  be 
 a timed word and $t \in \mathbb{R}$ be such that
 $-\tau_1 < t$. 
 The \emph{$t$-shift}
 $(\overline{a},\overline{\tau}) + t$
 of $(\overline{a},\overline{\tau})$ is the timed word
 $(\overline{a},\overline{\tau}+t)$, where $\overline{\tau}+t$ is the sequence $\tau_1+t,\tau_2+t,\dotsc,\tau_{|\overline{\tau}|}+t$.

  Let $(\overline{a},\overline{\tau})$ and $(\overline{a'},\overline{\tau'})$ be timed words over
  $\Sigma$ such that  $\tau_{|\tau|} < \tau'_{1}$. Their
 \emph{absorbing
  concatenation} $(\overline{a},\overline{\tau}) \circ (\overline{a'},\overline{\tau'})$ is 
  defined by
 \begin{math}
   (\overline{a},\overline{\tau}) \circ (\overline{a'},\overline{\tau'}) =
  (\overline{a}\circ\overline{a'}, \overline{\tau}\circ\overline{\tau'})
 \end{math}, where $\overline{a}\circ\overline{a'}$ and $\overline{\tau}\circ\overline{\tau'}$ denote
 (usual) concatenation of sequences over $\Sigma$ and $\mathbb{R}_{>0}$,
 respectively.

 Now let $(\overline{a},\overline{\tau})$ and $(\overline{a''},\overline{\tau''})$ be arbitrary timed
 words over $\Sigma$. 
 Their
  \emph{non-absorbing concatenation} $(\overline{a},\overline{\tau}) \cdot
  (\overline{a''},\overline{\tau''})$ is defined by $(\overline{a},\overline{\tau}) \cdot (\overline{a''},\overline{\tau''}) = (\overline{a},\overline{\tau}) \circ
  ((\overline{a''},\overline{\tau''}) + \tau_{|\overline{\tau}|})$. 

 A \emph{timed language} over an alphabet $\Sigma$ is  a set
 of timed words over  $\Sigma$.
\end{mydefinition}\end{minipage}

\vspace{.3em}
\noindent
\begin{minipage}{\textwidth} 
\begin{myremark}[signal]\label{rem:signal}
 \emph{Signal} is another formalization
  of records with a notion of time, used e.g.\ in~\cite{Ulus2014}; a signal
  over $\Sigma$ is a function $\mathbb{R}_{\geq 0} \to \Sigma$.  A timed
  word describes a time-stamped sequence of events, while a signal
  describes values of $\Sigma$ that change over time.  In this paper we
  shall work with timed words. This is for technical reasons and not
 important from the applicational point of view: when we restrict to those signals which
  exhibit only finitely many changes, there is a natural correspondence
  between such signals and timed words.
\end{myremark}
\end{minipage}
\vspace{.3em}


Let $C$ be a (fixed) finite set of \emph{clock variables}. 
The set $\Phi(C)$ of
\emph{clock constraints} 
is defined by the following BNF notation. 
 \[
\Phi(C) \;\ni\; \delta\; =\; x < c \mid x > c \mid x \leq c \mid x \geq
 c \mid \textbf{true}\mid \delta
 \land \delta \quad\text{where $x \in C$ and $c \in \mathbb{Z}_{\geq 0}$.}
 \]
 Absence of $\lor$ or $\lnot$ does not harm expressivity: $\lor$ can be
 emulated with nondeterminism (see Def.~\ref{def:timedAutom}); and $\lnot$ can be propagated
 down to atomic formulas by the de Morgan laws. Restriction to
 $\textbf{true}$ and $\land$ is technically useful, too, when we deal
 with intervals and zones (Def.~\ref{def:zone}). 

 A \emph{clock interpretation} $\nu$ over the set $C$ of clock variables
 is
 a function $\nu : C \to \mathbb{R}_{\geq 0}$. 
 Given a clock interpretation $\nu$ and $t \in \mathbb{R}_{\geq 0}$, $\nu
 + t$ denotes the clock interpretation that maps a clock variable $x\in C$ to
 $\nu (x) + t$.

\begin{mydefinition}[timed automaton]\label{def:timedAutom}
 A \emph{timed automaton (TA)} $\mathcal{A}$ is  a tuple
 $(\Sigma,S,S_0,C,E,F)$ where: $\Sigma$ is a finite alphabet; $S$ is a finite
 set of states; $S_0 \subseteq S$ is the set of initial states;
 $C$ is the set of clock variables; $E
 \subseteq S \times S \times \Sigma \times \mathcal{P}(C) \times
 \Phi(C)$ is the set of transitions; and
 $F \subseteq S$ is the set of accepting states.
\end{mydefinition}
The intuition for   $(s,s',a,\lambda,\delta)\in E$ is:
from $s$, also assuming that the clock constraint $\delta$ is
satisfied, we can move to the state $s'$ conducting the action $a$
and resetting the value of each clock variable $x\in \lambda$ 
to $0$. Examples of TAs are
in~(\ref{eq:TAExampleJustifyingDollar}) and
 Fig.~\ref{fig:exampleOfATimedAutom}--\ref{fig:case5} later.

The above notations (as well as the ones below) follow those 
in~\cite{Alur1994}. In the following definition~(\ref{eq:runOfTimedAutom}) of run, for example, 
the first transition occurs at (absolute) time $\tau_{1}$ and the second occurs at
time $\tau_{2}$; it is implicit that we stay at the state $s_{1}$ for
time
$\tau_{2}-\tau_{1}$.


\begin{mydefinition}[run]\label{def:runOfTimedAutom}
 A \emph{run} of a timed automaton $\mathcal{A} = (\Sigma,S,S_0,C,E,F)$
 over a timed word $(\overline{a},\overline{\tau}) \in (\Sigma \times \mathbb{R}_{>0})^*$ is
 a pair $(\overline{s},\overline{\nu}) \in S^* \times
 ((\mathbb{R}_{\geq 0})^C)^*$ of a sequence $\overline{s}$ of states and a sequence
 $\overline{\nu}$ of clock interpretations, subject to the following
 conditions: 1)
 $|\overline{s}| = |\overline{\nu}| = |\overline{a}| + 1$;
 2)
 $s_0 \in S_0$, and for any $x \in C$,
 $\nu_0 (x) = 0$; and 3)
 for any
 $i\in[0,|\overline{a}|-1]$
there exists a transition
 $(s_i,s_{i+1},a_{i+1},\lambda,\delta)\in E$
 such that the clock constraint $\delta$ holds under the clock interpretation
 $\nu_i + (\tau_{i+1} - \tau_{i})$ (here $\tau_{0}$ is
 defined to be $0$), and the clock interpretation $\nu_{i+1}$ has it that
 $\nu_{i+1} (x) = \nu_i (x) + \tau_{i+1} -
 \tau_{i}$  (if $x \notin \lambda$) and $\nu_{i+1} (x)
 = 0$ (if $x\in \lambda$). This run is depicted as follows.
%
\begin{equation}\label{eq:runOfTimedAutom}
 (s_0,\nu_0)
 \stackrel{(a_1,\tau_1)}{\longrightarrow}
 (s_1,\nu_1)
 \stackrel{(a_2,\tau_2)}{\longrightarrow}
 \cdots
 \longrightarrow
 (s_{|\overline{a}|-1},\nu_{|\overline{\tau}|-1})
 \stackrel{(a_{|\overline{a}|},\tau_{|\overline{\tau}|})}{\longrightarrow}
 (s_{|\overline{a}|},\nu_{|\overline{\tau}|})
%
\end{equation}

Such  a run $(\overline{s},\overline{\nu})$ 
of $\mathcal{A}$ is
 \emph{accepting} if $s_{|\overline{s}|-1} \in F$. The \emph{language}
 $L (\mathcal{A})$ of $\mathcal{A}$ is defined by
 $L (\mathcal{A}) = \{w
 \mid \text{there
 is an accepting run of  $\mathcal{A}$ over $w$}\}$. 
\end{mydefinition}

There is another 
specification formalism for timed languages called 
\emph{timed regular expressions (TREs)}~\cite{EugeneAsarin,Asarin2002}. 
Unlike in the classic Kleene theorem, in the timed case timed automata
are strictly more expressive than TREs. See~\cite[Prop.~2]{Herrmann1999}. 



\emph{Region automaton} is an important theoretical gadget in the theory
of timed automaton: it reduces the domain 
$S \times (\mathbb{R}_{\geq 0})^C$
of pairs $(s,\nu)$ in~(\ref{eq:runOfTimedAutom})---that is an \emph{infinite}
set---to its \emph{finite} abstraction, the latter being amenable to
algorithmic treatments. Specifically it relies on an equivalence
relation $\sim$ over clock interpretations.  Given 
a timed automaton $\mathcal{A} = (\Sigma,S,S_0,C,E,F)$---where, without
loss of generality,  we assume
that each clock variable $x \in C$ appears in at least one
clock constraint in $E$---let $c_x$ denote the greatest
number that is  compared with $x$ in the clock constraints in $E$.
(Precisely: $c_{x}=\max\{c\in\mathbb{Z}_{\ge 0}\mid \text{$x\bowtie c$
occurs in $E$, where ${\bowtie}\in\{<,>,\le,\ge\}$}\}$.)
Writing $\integer (\tau)$ and $\fractional (\tau)$ 
for the integer and fractional parts of 
 $\tau \in \mathbb{R}_{\geq 0}$, an
equivalence relation $\sim$ over clock interpretations $\nu,\nu'$ is
defined as follows. 
We have $\nu \sim \nu'$ if: 
\begin{itemize}
 \item for each $x\in C$ we have  $\integer(\nu(x)) =
       \integer(\nu'(x))$ or ($\nu (x) > c_x$ and $\nu' (x) > c_x$); 
 \item 
for any $x,y \in C$ such that $\nu (x) \leq c_x$ and $\nu (y) \leq
c_y$, $\fractional(\nu(x)) < \fractional(\nu(y))$ if and only if
$\fractional(\nu'(x)) < \fractional(\nu'(y))$; and
 \item 
for any $x \in C$ such that $\nu (x) \leq c_x$, $\fractional(\nu(x)) = 0$ if
and only if $\fractional(\nu'(x)) = 0$.
\end{itemize}
%
A \emph{clock region} 
is an  equivalence class of clock interpretations
modulo $\sim$; as usual the equivalence class of $\nu$ is denoted by
$[\nu]$. 
Let $\alpha,\alpha'$ be clock regions. We say 
 $\alpha'$ is a \emph{time-successor} of
$\alpha$ if for any $\nu \in \alpha$, there exists $t \in \mathbb{R}_{>
0 }$ such that $\nu + t \in \alpha'$. 



%
%
%

\vspace{.3em}
\noindent
\begin{minipage}{\textwidth} 
\begin{mydefinition}[region automaton]
\label{def:regionAutom}
 For a timed automaton $\mathcal{A} = (\Sigma,S,S_0,C,E,F)$, the
 \emph{region automaton} $R (\mathcal{A})$ is the NFA
 $(\Sigma,S',S'_0,E',F')$ defined as follows:
  $S' = S \times \bigl((\mathbb{R}_{\geq 0})^{C}/{\sim}\bigr)$;
on initial states $S'_0 = \{(s,[\nu]) \mid s \in S_0, \nu (x) = 0 \text{
       for each $x\in C$}\}$; on accepting states
$F' = \{(s,\alpha) \in S' \mid s \in F\}$. The transition relation
 $E'\subseteq S'\times S'\times \Sigma$ 
 is defined as follows:
       $((s,\alpha),(s',\alpha'),a) \in E'$ if there exist a
       clock region $\alpha''$ and $(s,s',a,\lambda,\delta) \in E$
       such that
\begin{itemize}
 \item  $\alpha''$ is a time-successor of $\alpha$, and
 \item for each $\nu \in \alpha''$, 1) $\nu$ satisfies $\delta$, and 2)
       there exists $\nu'
       \in \alpha'$ such that $\nu (x) = \nu' (x)$ (if $x\notin\lambda$)
       and $\nu' (x) = 0$ (if $x\in\lambda$).
\end{itemize}
\end{mydefinition}
\end{minipage}

\vspace{.3em}
\noindent
It is known~\cite{Alur1994} that the region automaton $R(\mathcal{A})$
 indeed has finitely many states.

The following notation for NFAs will be used later.

\vspace{.3em}
\noindent
\begin{minipage}{\textwidth} 
\begin{mydefinition}[$\Run_{\mathcal{A}} (s,s')$]\label{def:runNFA}
Let $\mathcal{A}$ be an NFA over $\Sigma$, and  $s$ and $s'$ be its states.
We let $\Run_{\mathcal{A}} (s,s')$ denote the set of runs from $s$ to
 $s'$, that is,
\begin{math}
 \Run_{\mathcal{A}} (s,s')
 =
 \{s_{0}s_{1}\dotsc s_{n}\mid n\in \mathbb{Z}_{\ge 0}, s_{0}=s,
 s_{n}=s', \forall i
.\,\exists a_{i+1}
.\,
 s_{i}\stackrel{a_{i+1}}{\rightarrow} s_{i+1} \text{ in $\mathcal{A}$}\}
\end{math}.
\end{mydefinition}
\end{minipage}

%

\subsection{String Matching and the Boyer-Moore Algorithm}
\label{subsec:stringMatchingAndOriginalBoyerMoore}
In~\S\ref{subsec:stringMatchingAndOriginalBoyerMoore}--\ref{subsec:patternMatchingAndWatsonWatsonAlgorithm}
we shall revisit the Boyer-Moore algorithm and its adaptation for
pattern matching~\cite{Watson2003}. We do so in 
considerable details, so as to provide both technical and intuitional bases
 for our timed adaptation.

\begin{wrapfigure}{r}{0pt}
\scriptsize
 \vspace*{-1em}
 \begin{tabular}{ccccccccccccccccccccccccccccccccccccccccccccccc}
  &\tiny 1&\tiny 2&\tiny 3&\tiny 4&\tiny 5&\tiny 6&\tiny 7&\tiny 8&\tiny
  9&\tiny 10&\tiny 11&\tiny 12&\tiny 13&\tiny 14&\tiny 15&\tiny 16&\tiny
  17&\tiny 18&\tiny 19&\tiny 20&\tiny 21&\tiny 22&\tiny 23&\tiny 24
  \\[-.3em]
  $\str=$&H&E&R&E&\ &I&S&\ &A&\ &S&I&M&P&L&E&\ &E&X&A&M&P&L&E\\
  $\pat=$&&&&&&&&&&&&&&&&&& E&X&A&M&P&L&E\\[-.5em]
  &&&&&&&&&&&&&&&&&& \tiny 1&\tiny 2&\tiny 3&\tiny 4&\tiny 5&\tiny 6&\tiny 7
 \end{tabular}
  \caption{The string matching problem}
  \label{fig:stringMatch}
\end{wrapfigure}
\emph{String matching} is a fundamental operation on strings:
given an input string $\str$ and a pattern string
$\pat$, it asks for the
\emph{match set} $\bigl\{(i,j)\,\bigl|\bigr.\,\str(i,j)=\pat \bigr\}$.
An example (from~\cite{MooreExample}) is in
Fig.~\ref{fig:stringMatch}, where the answer is $\{(18,24)\}$.


 A brute-force algorithm has the complexity
 $O(|\str||\pat|)$; known optimizations include ones by
Knuth, Morris, and Pratt~\cite{Knuth1977} and by 
Boyer and Moore~\cite{Boyer1977}. The former performs better in the
 worst case, but for practical
 instances the latter is commonly used. 
Let us now demonstrate how the Boyer-Moore algorithm for string
matching works, using the example in Fig.~\ref{fig:stringMatch}. Its
 main idea is to skip unnecessary matching of characters, using
 two \emph{skip value functions} $\Delta_{1}$ and $\Delta_{2}$ (that we
 define later).

The bottom line in the Boyer-Moore algorithm is that the pattern string
$\pat$ moves \emph{from left to right}, and matching between the
input string $\str$ and $\pat$ is conducted \emph{from
right to left}. In~(\ref{eq:config1}) is the initial configuration, and
we set out with comparing the characters $\str(7)$ and $\pat(7)$. They
turn out to be different.

\begin{wrapfigure}{r}{0pt}
\scriptsize
\begin{tabular}{c}
   \begin{tabular}[c]{ccccccccccccccccccccccccccccccccccccccccccccccc}
 \tiny 1&\tiny 2&\tiny 3&\tiny 4&\tiny 5&\tiny 6&\tiny 7&\tiny 8&\tiny
  9&\tiny 10&\tiny 11&\tiny 12&\tiny 13&\tiny 14&\tiny 15&\tiny 16&\tiny
  17&\tiny 18&\tiny 19&\tiny 20&\tiny 21&\tiny 22&\tiny 23&\tiny 24
  \\[-.3em]
 H&E&R&E&\ &I&S&\ &A&\ &S&I&M&P&L&E&\ &E&X&A&M&P&L&E\\
 E&X&A&M&P&L&E\\[-.5em]
 \tiny 1&\tiny 2&\tiny 3&\tiny 4&\tiny 5&\tiny 6&\tiny 7\\
 \end{tabular}
\hspace{-.6cm}\text{\begin{minipage}{1cm}\begin{equation}\label{eq:config1}\end{equation}\end{minipage}}\\
\scriptsize
 \begin{tabular}{ccccccccccccccccccccccccccccccccccccccccccccccc}
 \tiny 1&\tiny 2&\tiny 3&\tiny 4&\tiny 5&\tiny 6&\tiny 7&\tiny 8&\tiny
  9&\tiny 10&\tiny 11&\tiny 12&\tiny 13&\tiny 14&\tiny 15&\tiny 16&\tiny
  17&\tiny 18&\tiny 19&\tiny 20&\tiny 21&\tiny 22&\tiny 23&\tiny 24
  \\[-.3em]
 H&E&R&E&\ &I&S&\ &A&\ &S&I&M&P&L&E&\ &E&X&A&M&P&L&E\\
 &&&&&&& E&X&A&M&P&L&E\\[-.5em]
 &&&&&&& \tiny 1&\tiny 2&\tiny 3&\tiny 4&\tiny 5&\tiny 6&\tiny 7\\
 \end{tabular}
\hspace{-.6cm}\text{\begin{minipage}{1cm}\begin{equation}\label{eq:config2}\end{equation}\end{minipage}}\\
\scriptsize
 \begin{tabular}{ccccccccccccccccccccccccccccccccccccccccccccccc}
 \tiny 1&\tiny 2&\tiny 3&\tiny 4&\tiny 5&\tiny 6&\tiny 7&\tiny 8&\tiny
  9&\tiny 10&\tiny 11&\tiny 12&\tiny 13&\tiny 14&\tiny 15&\tiny 16&\tiny
  17&\tiny 18&\tiny 19&\tiny 20&\tiny 21&\tiny 22&\tiny 23&\tiny 24
  \\[-.3em]
 H&E&R&E&\ &I&S&\ &A&\ &S&I&M&P&L&E&\ &E&X&A&M&P&L&E\\
 &&&&&&&&& E&X&A&M&P&L&E\\[-.5em]
 &&&&&&&&& \tiny 1&\tiny 2&\tiny 3&\tiny 4&\tiny 5&\tiny 6&\tiny 7\\
 \end{tabular}
\hspace{-.6cm}\text{\begin{minipage}{1cm}\begin{equation}\label{eq:config3}\end{equation}\end{minipage}}
\end{tabular}
\end{wrapfigure}
A naive algorithm would then move the pattern to the right by one
position.  We can do better, however, realizing that the character
$\str(7)=\mathrm{S}$ (that we already read for comparison) never occurs
in the pattern $\pat$. This means the position $7$ cannot belong to
any matching interval $(i,j)$, and we thus jump to the
configuration~(\ref{eq:config2}).  Formally this argument is expressed
by the value $\Delta_{1}(\mathrm{S},7)=7$ of the first skip value
function $\Delta_{1}$, as we will see later.

Here again we compare characters from right to left, in~(\ref{eq:config2}), realizing
immediately that $\str(14)\neq\pat(7)$. It is time to shift the pattern;
given that $\str(14)=\mathrm{P}$ occurs as $\pat(5)$, we shift the
pattern by 
$\Delta_{1}(\mathrm{P},7)=7-5=2$.

We are now in the configuration~(\ref{eq:config3}), and some 
initial matching
succeeds ($\str(16)=\pat(7)$, $\str(15)=\pat(6)$, and so on). 
The matching  fails for $\str(12)\neq\pat(3)$. Following the same reasoning
as above---the character $\str(12)=\mathrm{I}$ does not occur in
$\pat(3)$, $ \pat(2)$ or $\pat(1)$---we would then shift
the pattern
by $\Delta_{1}(\mathrm{I},3)=3$. 

\begin{wrapfigure}[8]{r}{0pt}
\scriptsize
 \begin{tabular}{c|cccccccccccccccccccccccccccccccccccccccccccccc}
 & \tiny 1&\tiny 2&\tiny 3&\tiny 4&\tiny 5&\tiny 6&\tiny 7\\[-.3em]
 &E&X&A&\bf\underline M&\bf\underline P&\bf\underline L&\bf\underline E\\[-.0em]\hline
 \color{red}\xmark&*&E&X&\bf\underline A&\bf\underline M&\bf\underline P&\bf\underline L&E\\[-.0em]
 \color{red}\xmark&*&*&E&\bf\underline X&\bf\underline A&\bf\underline M&\bf\underline P&L&E\\[-.0em]
 $\vdots$&&&&$\vdots$\\[-.0em]
  \color{red}\xmark&*&*& *&\bf\underline *&\bf\underline *&\bf\underline E&\bf\underline X&A&M&P&L&E\\[-.0em]
 \color{dgreen}\cmark&*&*&*&\bf\underline *&\bf\underline *&\bf\underline *&\bf\underline E&X&A&M&P&L&E
 \end{tabular}
  \caption{Table for computing $\Delta_{2}$}
  \label{fig:tableForDelta2}
\end{wrapfigure}
However we can do even better. Consider the table on the right, where
we forget about the input $\str$ and shift the pattern $\pat$ one by one, trying to match it with $\pat$
itself. We are specifically interested in the segment $\mathrm{MPLE}$
from $\pat(4)$ to $\pat(7)$ (underlined in the first row)---because it is the partial match we
have discovered
in the configuration~(\ref{eq:config3}). 
The table shows that we need to shift at least by $6$ to get a potential
match (the last row); hence from the configuration~(\ref{eq:config3}) we can shift the
pattern $\pat$ by $6$, which is more than 
the skip value in the above ($\Delta_{1}(\mathrm{I},3)=3$).
 This argument---different from the one for $\Delta_{1}$---is
formalized as the second skip value function $\Delta_{2}(3)=6$.


\begin{wrapfigure}[3]{r}{0pt}
\scriptsize
 \begin{tabular}{ccccccccccccccccccccccccccccccccccccccccccccccc}
 \tiny 1&\tiny 2&\tiny 3&\tiny 4&\tiny 5&\tiny 6&\tiny 7&\tiny 8&\tiny
  9&\tiny 10&\tiny 11&\tiny 12&\tiny 13&\tiny 14&\tiny 15&\tiny 16&\tiny
  17&\tiny 18&\tiny 19&\tiny 20&\tiny 21&\tiny 22&\tiny 23&\tiny 24
  \\[-.3em]
 H&E&R&E&\ &I&S&\ &A&\ &S&I&M&P&L&E&\ &E&X&A&M&P&L&E\\
 &&&&&&&&&&&&&&& E&X&A&M&P&L&E\\[-.5em]
 &&&&&&&&&&&&&&& \tiny 1&\tiny 2&\tiny 3&\tiny 4&\tiny 5&\tiny 6&\tiny 7\\
 \end{tabular}
\hspace{-.6cm}\text{\begin{minipage}{1cm}\begin{equation}\label{eq:config4}\end{equation}\end{minipage}}
\end{wrapfigure}
We are led to the configuration on the right, only to find that the
first matching trial fails ($\str(22)\neq\pat(7)$). Since
$\str(22)=\mathrm{P}$ occurs in $\pat$ as $\pat(5)$, we shift $\pat$ by 
$\Delta_{1}(\mathrm{P},7)= 2$. This finally brings us to the
configuration
in Fig.~\ref{fig:stringMatch} and the match set $\{(18,24)\}$. 

Summarizing, the key in the Boyer-Moore algorithm is to use two skip
value
functions
$\Delta_{1},\Delta_{2}$ to shift the pattern faster than one-by-one. 
The precise definition of $\Delta_{1},\Delta_{2}$ is in
Appendix~\ref{appendix:skipValueFunctions}, for reference.

\subsection{Pattern Matching and a Boyer-Moore type Algorithm
}
\label{subsec:patternMatchingAndWatsonWatsonAlgorithm}
\newcommand{\calloutabove}[2]{\node[rectangle callout,draw,
inner sep=2pt,rounded corners=2pt,
callout absolute pointer={(#1.north)},above=1.5cm of #1]{#2};}

\emph{Pattern matching} is another fundamental operation that
generalizes string matching: given an input string $\str$ and a
regular language $L$ as a \emph{pattern}, it asks for the \emph{match
set} $\bigl\{(i,j)\,\bigl|\bigr.\,\str(i,j)\in L \bigr\}$.  For
example, for $\str$ in Fig.~\ref{fig:stringMatch} and the pattern
$\mathrm{[A\mathchar`-Z]}^*\mathrm{MPLE}$, the match set is $\{(11,16),(18,24)\}$.
In~\cite{Watson2003} an algorithm for pattern matching is introduced that
employs ``Boyer-Moore type'' optimization, much like the use of
$\Delta_{2}$ in~\S{}\ref{subsec:stringMatchingAndOriginalBoyerMoore}.

\begin{figure}[tbp]
\begin{minipage}[c]{.7\textwidth}
  \centering
 \scalebox{0.7}{
 \begin{tikzpicture}[shorten >=1pt,node distance=3cm,on grid,auto]
 \node[state,initial right] at (0,0) (s_0) {$s_0$};
 \node[state] (s_1) at (-2.2,1) {$s_1$};
 \node[state] (s_2) at (-2.2,-1) {$s_2$};
 \node[state] (s_3) at (-4.4,0) {$s_3$};
 \node[state,accepting] (s_4) at (-7.6,0) {$s_4$};

  \node[rectangle callout,draw,
  inner sep=2pt,rounded corners=2pt,
  callout absolute pointer={(s_0.north east)}] at (1.83,1.3) {
   \begin{tabular}{c}
    $m_{s_0} = 0$\\
    $L'_{s_0} = \{\varepsilon\}$\\
    $\Delta_2(\{s_0\}) = 1$\\
   \end{tabular}};
  \node[rectangle callout,draw,
  inner sep=2pt,rounded corners=2pt,
  callout absolute pointer={(s_1.north east)}] at (-0.5,1.5) {
   \begin{tabular}{c}
    $m_{s_1} = 1$\\
    $L'_{s_1} = \{\text{``a''}\}$\\
    $\Delta_2(\{s_1\}) = 3$\\
   \end{tabular}};
  \node[rectangle callout,draw,
  inner sep=2pt,rounded corners=2pt,
  callout absolute pointer={(s_2.east)}
  ] at (0.6,-1.18) {
  \begin{tabular}{c}
   $m_{s_2} = 1$\\
   $L'_{s_2} = \{\text{``c''}\}$\\
   $\Delta_2(\{s_2\}) = 2$\\
  \end{tabular}};
  \node[rectangle callout,draw,
  inner sep=2pt,rounded corners=2pt,
  callout absolute pointer={(s_3.south west)},below left=1.7cm of s_3]{
   \begin{tabular}{c}
    $m_{s_3} = 2$\\
    $L'_{s_3} = \{\text{``ba'',``dc''}\}$\\
    $\Delta_2(\{s_3\}) = 2$\\
   \end{tabular}};
  \calloutabove{s_4}{
   \begin{tabular}{c}
    $m_{s_4} = 3$\\
    $L'=L'_{s_4} = \{\text{``dba'',``ddc'',``cba'',``cdc''}\}$\\
    $\Delta_2(\{s_4\}) = 2$\\
   \end{tabular}}

 \path[->] 
  (s_0) edge [above] node {a} (s_1)
  (s_0) edge [above] node {c} (s_2)
  (s_1) edge [above] node {b} (s_3)
  (s_2) edge [above] node {d} (s_3)
  (s_3) edge [loop above] node {c} (s_3)
  (s_3) edge [above] node {d} (s_4);
 \end{tikzpicture}}
 \caption{The automaton $\mathcal{A}$ and skip values}
 \label{192552_1Feb16}
\end{minipage} 
\hfill
\begin{minipage}[c]{.23\textwidth}
\centering
\scalebox{.6}{
\begin{pgfpicture}
\pgfpathmoveto{\pgfpoint{85.848145bp}{3229.09292bp}}
\pgfpathlineto{\pgfpoint{186.071289bp}{3287.332666bp}}
\pgfusepath{use as bounding box}
\begin{pgfscope}
\pgftransformcm{1.0}{0.0}{0.0}{1.0}{\pgfpoint{100.0bp}{3280.393701bp}}
\pgftext[left,base]{\normalsize\color[rgb]{0.0,0.0,0.0}c}
\end{pgfscope}
\begin{pgfscope}
\pgftransformcm{1.0}{0.0}{0.0}{1.0}{\pgfpoint{110.0bp}{3280.393701bp}}
\pgftext[left,base]{\normalsize\color[rgb]{0.0,0.0,0.0}b}
\end{pgfscope}
\begin{pgfscope}
\pgftransformcm{1.0}{0.0}{0.0}{1.0}{\pgfpoint{120.0bp}{3280.393701bp}}
\pgftext[left,base]{\normalsize\color[rgb]{0.0,0.0,0.0}a}
\end{pgfscope}
\begin{pgfscope}
\pgftransformcm{1.0}{0.0}{0.0}{1.0}{\pgfpoint{130.0bp}{3280.393701bp}}
\pgftext[left,base]{\normalsize\color[rgb]{0.0,0.0,0.0}d}
\end{pgfscope}
\begin{pgfscope}
\pgftransformcm{1.0}{0.0}{0.0}{1.0}{\pgfpoint{140.0bp}{3280.393701bp}}
\pgftext[left,base]{\normalsize\color[rgb]{0.0,0.0,0.0}c}
\end{pgfscope}
\begin{pgfscope}
\pgftransformcm{1.0}{0.0}{0.0}{1.0}{\pgfpoint{150.0bp}{3280.393701bp}}
\pgftext[left,base]{\normalsize\color[rgb]{0.0,0.0,0.0}d}
\end{pgfscope}
\begin{pgfscope}
\pgftransformcm{1.0}{0.0}{0.0}{1.0}{\pgfpoint{160.0bp}{3280.393701bp}}
\pgftext[left,base]{\normalsize\color[rgb]{0.0,0.0,0.0}c}
\end{pgfscope}
\begin{pgfscope}
\pgftransformcm{1.0}{0.0}{0.0}{1.0}{\pgfpoint{135.0bp}{3270.393701bp}}
\pgftext[left,base]{\normalsize\color[rgb]{0.0,0.0,0.0}$\vdots$}
\end{pgfscope}
\begin{pgfscope}
\pgfsetlinewidth{1.0bp}
\pgfsetrectcap
\pgfsetmiterjoin
\pgfsetmiterlimit{10.0}
\pgfpathmoveto{\pgfpoint{110.0bp}{3265.393701bp}}
\pgfpathlineto{\pgfpoint{135.0bp}{3265.393701bp}}
\pgfpathlineto{\pgfpoint{135.0bp}{3255.393701bp}}
\pgfpathlineto{\pgfpoint{110.0bp}{3255.393701bp}}
\definecolor{strokepaint}{rgb}{0.0,0.0,0.0}\pgfsetstrokecolor{strokepaint}
\pgfusepath{stroke}
\end{pgfscope}
\begin{pgfscope}
\pgfsetlinewidth{1.0bp}
\pgfsetrectcap
\pgfsetmiterjoin
\pgfsetmiterlimit{10.0}
\pgfpathmoveto{\pgfpoint{120.0bp}{3250.393701bp}}
\pgfpathlineto{\pgfpoint{145.0bp}{3250.393701bp}}
\pgfpathlineto{\pgfpoint{145.0bp}{3240.393701bp}}
\pgfpathlineto{\pgfpoint{120.0bp}{3240.393701bp}}
\definecolor{strokepaint}{rgb}{0.0,0.0,0.0}\pgfsetstrokecolor{strokepaint}
\pgfusepath{stroke}
\end{pgfscope}
\begin{pgfscope}
\pgftransformcm{1.0}{0.0}{0.0}{1.0}{\pgfpoint{135.0bp}{3230.393701bp}}
\pgftext[left,base]{\sffamily\mdseries\upshape\normalsize\color[rgb]{0.0,0.0,0.0}$\vdots$}
\end{pgfscope}
\begin{pgfscope}
\pgftransformcm{1.0}{0.0}{0.0}{1.0}{\pgfpoint{95.0bp}{3255.393701bp}}
\pgftext[left,base]{\sffamily\mdseries\upshape\normalsize\color[rgb]{0.0,0.0,0.0}$\leftarrow$}
\end{pgfscope}
\begin{pgfscope}
\pgftransformcm{1.0}{0.0}{0.0}{1.0}{\pgfpoint{105.0bp}{3240.393701bp}}
\pgftext[left,base]{\sffamily\mdseries\upshape\normalsize\color[rgb]{0.0,0.0,0.0}$\leftarrow$}
\end{pgfscope}
\begin{pgfscope}
\pgftransformcm{1.0}{0.0}{0.0}{1.0}{\pgfpoint{100.0bp}{3260.393701bp}}
\pgftext[left,base]{\sffamily\mdseries\upshape\normalsize\color[rgb]{0.0,0.0,0.0}$i$}
\end{pgfscope}
\begin{pgfscope}
\pgftransformcm{1.0}{0.0}{0.0}{1.0}{\pgfpoint{110.0bp}{3245.393701bp}}
\pgftext[left,base]{\sffamily\mdseries\upshape\normalsize\color[rgb]{0.0,0.0,0.0}$i$}
\end{pgfscope}
\begin{pgfscope}
\pgftransformcm{1.0}{0.0}{0.0}{1.0}{\pgfpoint{170.0bp}{3258.393701bp}}
\pgftext[left,base]{\sffamily\mdseries\upshape\normalsize\color[rgb]{0.0,0.0,0.0}$j=4$}
\end{pgfscope}
\begin{pgfscope}
\pgftransformcm{1.0}{0.0}{0.0}{1.0}{\pgfpoint{170.0bp}{3243.393701bp}}
\pgftext[left,base]{\sffamily\mdseries\upshape\normalsize\color[rgb]{0.0,0.0,0.0}$j=5$}
\end{pgfscope}
\end{pgfpicture}
}
 \caption{A brute-force algorithm}
 \label{fig:bruteforce}
\end{minipage}
\end{figure}

Let $\str=\mathrm{cbadcdc}$ be an input string and
$\mathrm{dc}^*\{\mathrm{ba} \mid \mathrm{dc}\}$ be a pattern $L$, for
example.  We can solve pattern matching by the following brute-force
algorithm.  
\begin{itemize}
 \item 
 We express the pattern as an NFA
$\mathcal{A}=(\Sigma,S,S_{0},E,F)$
in Fig.~\ref{192552_1Feb16}. We reverse words---$w\in
L$ if and only if $w^{\Rev}\in L(\mathcal{A})$---following the 
Boyer-Moore algorithm
(\S{}\ref{subsec:stringMatchingAndOriginalBoyerMoore}) that matches a
       segment of  input and a pattern \emph{from right to left}.
 \item 
Also following~\S{}\ref{subsec:stringMatchingAndOriginalBoyerMoore} we
``shift the pattern from left to right.'' This technically means:
we conduct the following for $j=1$ first, then $j=2$, and so on.
 For fixed $j$ we search for $i\in[1,j]$ such that $\str(i,j)\in
 L$. This is done by computing the set $S_{(i,j)}$ of reachable states of $\mathcal{A}$
 when it is fed with $\str(i,j)^{\Rev}$. The computation is done
 step-by-step, decrementing $i$ from $j$ to  $1$:
 \begin{displaymath}
	S_{(j,j)}=\{s\mid \exists
 s'\in S_{0}.\, s'\stackrel{\str(j)}{\to} s\}\enspace,
 \;\text{and}\;
 S_{(i,j)}=\{s\mid \exists s'\in S_{(i+1,j)}.\, 
s'\stackrel{\str(i)}{\to} s\}\enspace.
       \end{displaymath}
\end{itemize}
Then $(i,j)$ is a matching interval if and only if $S_{(i,j)}\cap
F\neq\emptyset$. See Fig.~\ref{fig:bruteforce}. 

The Boyer-Moore type optimization in~\cite{Watson2003} tries to hasten
the shift of $j$. A key observation is as follows. Assume $j=4$;
then the above procedure would feed
$\str(1,4)^{\Rev}=\mathrm{dabc}$ to the automaton $\mathcal{A}$
(Fig.~\ref{192552_1Feb16}). We instantly see that this would not 
yield any matching interval---for a word to be accepted by $\mathcal{A}$ 
it must start with abd, cdd, abc or cdc. 

Precisely the algorithm in~\cite{Watson2003} works as follows. 
We first observe that the shortest word accepted by $\mathcal{A}$ is
$3$; therefore we can start right away with $j=3$, skipping $j=1,2$
in the above brute-force algorithm. 
Unfortunately $\str(1,3)=\mathrm{cba}$ does not match  $L$, with 
$\str(1,3)^{\Rev}=\mathrm{abc}$ only leading to $\{s_{3}\}$ in
$\mathcal{A}$. 

\begin{wrapfigure}[11]{r}{0pt}
\scriptsize
 \begin{tabular}{c|cccccccccccccccccccccccccccccccccccccccccccccc}
  & \tiny 2&\tiny 3&\tiny 4&\tiny 5\\[-.3em]
\multirow{2}{*}{$L'_{s_{3}}$}
 &b&a\\[-.0em]
 &d&c\\[-.0em]
 \hline\hline
 \multirow{4}{*}{
\begin{minipage}{3em}
  $L'$ shifted 
  \\ by $1$\\
 ({\color{red}\xmark})
\end{minipage}
}
 &d&b&a\\[-.0em]
 &d&d&c\\[-.0em]
 &c&b&a\\[-.0em]
 &c&d&c\\[-.0em]
 \hline
 \multirow{4}{*}{
\begin{minipage}{3em}
  $L'$ shifted 
  \\ by $2$\\
 ({\color{dgreen}\cmark})
\end{minipage}
}
 &*&d&b&a\\[-.0em]
 &*&d&d&c\\[-.0em]
 &\underline{*}&\underline{c}&b&a\\[-.0em]
  &\underline{*}&\underline{c}&d&c\\[-.0em]
 \end{tabular}
  \caption{Table for $\Delta_{2}(s_{3})$}
  \label{fig:tableForDelta2InPatternMatching}
\end{wrapfigure}
We now shift $j$ by $2$, from $3$ to $5$, following
Fig.~\ref{fig:tableForDelta2InPatternMatching}. Here 
\begin{equation}\label{eq:LPrimeExample}
 L'=\bigl\{\,\bigl(w(1,3)\bigr)^{\Rev}\,\bigl|\bigr.\, w\in L(\mathcal{A})\,\bigr\}=\{
 \mathrm{dba},
 \mathrm{ddc},
 \mathrm{cba},
 \mathrm{cdc}
\}\enspace;
\end{equation}
that is, for $\str(i,j)$ to match the pattern $L$ its last three
characters must match $L'$. Our previous ``key observation'' now
translates to the fact that
$\str(2,4)=\mathrm{bad}$ does not belong to $L'$; in the actual
algorithm in~\cite{Watson2003}, however, we do not use the string $\str(2,4)$ itself. Instead 
we \emph{overapproximate} it with the information that feeding
$\mathcal{A}$ with $\str(1,3)^{\Rev}=\mathrm{abc}$ led to $\{s_{3}\}$. 
Similarly to the case with $L'$, this implies that the last two
characters
of $\str(1,3)$ must have been in
$L'_{s_{3}}=\{\mathrm{ba},\mathrm{dc}\}$. 
The table shows that none in $L'_{s_{3}}$ matches any of $L'$ when
$j$ is shifted by $1$; when $j$ is shifted by $2$, we have matches
(underlined). Therefore we jump from $j=3$ to $j=5$. 

This is how the algorithm in~\cite{Watson2003} works: 
it accelerates the brute-force algorithm in Fig.~\ref{fig:bruteforce}, 
skipping some $j$'s,
with the help of a skip value function $\Delta_{2}$. The
overapproximation in the last paragraph allows  
 $\Delta_{2}$ to rely only on a pattern $L$ (but not on an
input string $\str$); this means that pre-processing is done once we fix
the pattern $L$, and it is reused for various input strings $\str$. This
is an
advantage in monitoring applications where one would deal with a number
of input strings $\str$, some of which are yet to come.
\auxproof{
 Concretely the pre-processing part computes: $m_{s}$ (the
 length of the shortest word that leads to a state $s$ in $\mathcal{A}$);
 $L'_{s}=\{w(1,m_{s})^{\Rev}\mid \text{$w$ leads to $s$ in
 $\mathcal{A}$}\}$, like in~(\ref{eq:LPrimeExample}); $m=\min_{s\in F}
 m_{s}$ (i.e.\ the length of the shortest accepted word);
 $L'=\{w(1,m)^{\Rev}\mid w\in L(\mathcal{A})\}$; and finally the skip
 value $\Delta_{2}(s)$ for each state $s$, computed like in
 Fig.~\ref{fig:tableForDelta2InPatternMatching}. In fact $\Delta_{2}$ is
 defined for each set $C\subseteq S$ of states; see
 Appendix~\ref{appendix:skipValueFunctionForWatsonsAlgorithm} for
 details.
}
See Appendix~\ref{appendix:skipValueFunctionForWatsonsAlgorithm} for
the precise definition of the skip value function $\Delta_{2}$.

In Fig.~\ref{192552_1Feb16} we annotate each state $s$ with the values
$m_{s}$ and $L'_{s}$ that is used in computing
$\Delta_{2}(\{s\})$. Here $ m_{s}
$ is
the length of a shortest word that leads to $s$; 
\begin{math}
  m
 =
 \min_{s\in F}m_{s}
\end{math}
(that is $3$ in the above example); 
and 
$ L'_{s}
 =
 \{w(1,\min\{m_{s},m\})^{\Rev}\mid w\in L(\mathcal{A}_{s})\}
$.

It is not hard to generalize 
the other skip value function $\Delta_1$
in~\S{}\ref{subsec:stringMatchingAndOriginalBoyerMoore} for pattern
matching, too: in place of $\pat$ we use the set $L'$
in the above~(\ref{eq:LPrimeExample}). See
Appendix~\ref{appendix:skipValueFunctionForWatsonsAlgorithm}.

\section{The Timed Pattern Matching  Problem
}
\label{103806_9Apr16}
Here we formulate our problem, following
the notations in~\S{}\ref{221705_24Jan16}. 

Given a timed word $w$, the timed word segment 
$w|_{(t,t')}$ is the result of clipping the parts outside the open interval
 $(t,t')$. For example, for 
 $w = \bigl((a,b,c),(0.7,1.2,1.5)\bigr)$, 
 we have
 $w|_{(1.0,1.7)} = \bigl((b,c,\$),(0.2,0.5,0.7)\bigr)$,
 $w|_{(1.0,1.5)} = \bigl((b,\$),(0.2,0.5)\bigr)$ and
 $w|_{(1.2,1.5)} = \bigl((\$),(0.3)\bigr)$.
 Here the (fresh) \emph{terminal character} $\$$ designates the end of a
 segment. Since we use open intervals $(t,t')$, for example, the word
 $w|_{(1.2,1.5)}$ does not contain the character $c$ at time $0.3$. 
 The formal definition is as follows.

\vspace{.3em}
\noindent
\begin{minipage}{\textwidth}
\begin{mydefinition}[timed word segment]\label{def:timedWordSegment}
 Let $w = (\overline{a},\overline{\tau})$ be a timed word over $\Sigma$,
 $t$ and $t'$ be reals such that $0 \leq t < t'$, and
 $i$ and $j$ be  indices such that $\tau_{i-1} \leq t < \tau_{i}$
 and $\tau_j < t' \leq \tau_{j+1}$ (we let $\tau_0=0$ and
 $\tau_{|\overline{\tau}|+1}=\infty$).
 The \emph{timed word segment} $w|_{(t,t')}$ of $w$ on the interval
 $(t,t')$ is the timed word
 $(\overline{a'},\overline{\tau'})$, over the extended alphabet $\Sigma
 \amalg \{\$\}$, 
 defined as follows:
 1) $\left| w|_{(t,t')} \right|= j-i+2$;  2) we have $a'_k =
 a_{i+k-1}$ and $\tau'_k = \tau_{i+k-1} - t$ for $k\in[1, j - i + 1]$;
 and 3)
  $a'_{j-i+2} = \$ $ and $\tau'_{j-i+2} = t'-t$.
 \end{mydefinition}
\end{minipage}
\vspace{.3em}
\noindent
\begin{minipage}{\textwidth}
\begin{mydefinition}[timed pattern matching] \label{TimedPatternMatching}
The  \emph{timed pattern matching} problem (over an alphabet $\Sigma$)
takes (as input)
 a timed word $w$ over $\Sigma$ and a timed automaton $\mathcal{A}$
 over $\Sigma \amalg \{\$\}$;  and
it requires (as output) the \emph{match set}
 $\mathcal{M} (w,\mathcal{A}) = \bigl\{(t,t') \in (\mathbb{R}_{\ge 0})^{2} \mid t < t', w|_{(t,t')}
 \in L (\mathcal{A})\bigr\}$.
\end{mydefinition}
\end{minipage}

\begin{wrapfigure}{r}{0pt}
 \parbox{0.3\textwidth}{
 \vspace{-1.0em}
 \scalebox{0.6}{
 \begin{tikzpicture}[shorten >=1pt,node distance=2cm,on grid,auto]
 \node[state,initial] (s_0) {$s_0$};
 \node[state] (s_1) [right of=s_0] {$s_1$};
 \node[state,accepting] (s_2) [right of=s_1]{$s_2$};

 \path[->] 
  (s_0) edge [above] node {
  \begin{tabular}{c}
   $a,\mathbf{true}$\\
   $/t:=0$\\
  \end{tabular}} (s_1)
  (s_1) edge [above] node {$\$,t\ge 2$} (s_2);
 \end{tikzpicture}}
}
\hspace{-.6cm}\text{\begin{minipage}{1cm}\begin{equation}\label{eq:TAExampleJustifyingDollar}\end{equation}\end{minipage}}
\end{wrapfigure}
\noindent Our formulation in Def.~\ref{TimedPatternMatching} slightly differs from 
that in~\cite{Ulus2014} in that: 1) we use timed words in place of
signals (Rem.~\ref{rem:signal}); 2) for specification we use timed
automata
rather than timed regular expressions; and 3) we use an 
explicit terminal character $\$$. While none of these differences is major,
introduction of $\$$ enhances expressivity, e.g.\ in specifying
 ``an event $a$ occurs, and after that, no
other event occurs within 2s.''
(see~(\ref{eq:TAExampleJustifyingDollar})). It is also easy to ignore
$\$$---when one is not interested in it---by having the clock constraint
\textbf{true} on the $\$$-labeled transitions leading to the accepting states.
\vspace{.3em}
\noindent
\begin{minipage}{\textwidth}
\begin{myassumption}
 In what follows we assume the following. Each  timed automaton $\mathcal{A}$ over the alphabet
 $\Sigma\amalg\{\$\}$ is such that: every $\$$-labeled transition is
 into an accepting state; and no other transition is $\$$-labeled.
 And there exists no transition from any accepting states.
\end{myassumption}
\end{minipage}

\section{
 A Naive Algorithm and Its Online Variant } \label{sec:naiveAlgo} 

Here
we present a naive algorithm for timed pattern matching (without a
Boyer-Moore type optimization), also indicating how to make it into an
online one.  Let us fix a timed word $w$ over $\Sigma$ and a timed
automaton $\mathcal{A} = (\Sigma \amalg \{\$\},S,S_0,C,E,F)$ as the
input.

First of all, a match set (Def.~\ref{TimedPatternMatching}) is in
general an infinite set, and we need its finitary representation 
for an algorithmic treatment. We follow~\cite{Ulus2014} and use 
(2-dimensional) zones for that purpose.

\vspace{.3em}
\noindent
\begin{minipage}{\textwidth}
\begin{mydefinition}[zone]\label{def:zone}
 Consider the 2-dimensional plane $\mathbb{R}^{2}$ whose axes are denoted
 by $t$ and $t'$. A \emph{zone} is a convex polyhedron  specified by 
constraints of the form
 $t\bowtie c$, 
 $t'\bowtie c$ and  
 $t'-t\bowtie c$, where 
 ${\bowtie} \in
\{<,>,\leq,\geq\}$ and $c \in \mathbb{Z}_{\geq 0}$.

 It is not hard to see that each zone is specified by three intervals
 (that may or may not include their endpoints):
 $T_{0}$ for $t$, $T_{f}$ for $t'$ and $T_{\Delta}$ for $t'-t$. 
 We let a triple $(T_{0},T_{f},T_{\Delta})$ represent a zone, and write
  $(t,t') \in (T_0,T_f,T_{\Delta})$ if  $t \in
T_0$, $t' \in T_f$ and $t' - t \in T_{\Delta}$.
\end{mydefinition}
\end{minipage}


In our algorithms we shall use the following constructs.

\vspace{.3em}
\noindent
\begin{minipage}{\textwidth}
\begin{mydefinition}[$\reset,\eval,\solConstr,\rhoEmpty,\Conf$]\label{def:alg_funcs}
 Let $\rho\colon C\rightharpoonup\mathbb{R}_{>0}$ be a partial function 
 that carries a clock variable $x\in C$ to a positive real;
 the intention is that $x$ was reset at time $\rho(x)$ (in the absolute
 clock). 
Let $x\in C$ and $t_{r}\in\mathbb{R}_{>0}$; then the partial function
 $\reset(\rho,x,t_r)\colon C\rightharpoonup\mathbb{R}_{>0}$ 
 is defined by: $\reset(\rho,x,t_r)(x)=t_{r}$ and 
$\reset(\rho,x,t_r)(y)=\rho(y)$ for each $y\in C$ such that
 $y\neq x$. (The last is Kleene's equality between partial functions, to be precise.)

Now let $\rho$ be as above, and $t,t_{0}\in\mathbb{R}_{\ge 0}$,
with the intention that $t$ is the current (absolute) time and $t_{0}$
 is the epoch (absolute) time for a timed word segment
 $w|_{(t_{0},t')}$. We further assume $t_{0}\le t$
 and $t_{0}\le \rho(x)\le t$ for each $x\in C$ for which $\rho(x)$ is
 defined. 
The clock interpretation 
 $\eval (\rho,t,t_0)\colon C\to \mathbb{R}_{\ge 0}$
 is defined by: $\eval (\rho,t,t_0)(x)=t-\rho(x)$ (if $\rho(x)$
 is defined); and 
$\eval (\rho,t,t_0)(x)=t-t_{0}$ (if $\rho(x)$ is undefined).

For intervals $T,T'\subseteq\mathbb{R}_{\ge 0}$, a partial function 
$\rho\colon C\rightharpoonup\mathbb{R}_{>0}$ and a clock constraint
 $\delta$ (\S{}\ref{221705_24Jan16}), we define
 $\solConstr(T,T',\rho,\delta)=\bigl\{\,(t,t')\,\bigl|\bigr.\,t\in T, t'\in T',
 \eval(\rho,t',t)\models \delta\,\bigr\}$.

We let $\rhoEmpty\colon C\rightharpoonup\mathbb{R}_{>0}$ denote the
 partial function that is nowhere defined.

For a timed word $w$, a timed automaton $\mathcal{A}$ and each $i,j\in
 [1, |w|]$, we 
define the  set of ``configurations'':
 \begin{math}
  \Conf (i,j) = \bigl\{ (s,\rho,T) \,\bigl|\bigr.\, \forall t_0 \in
  T.\, \exists (\overline{s},\overline{\nu}).\,  
  (\overline{s},\overline{\nu}) \text{ is a run  over } w (i,j) -
  t_0 ,
   s_{|\overline{s}|-1} = s \text{, and }
  \nu_{|\overline{\nu}|-1} = \eval (\rho,\tau_j,t_0)\bigr\}
 \end{math}.
Further details are in
 Appendix~\ref{appendix:futherExplOfNaiveAlgorithm}. 
\end{mydefinition}
\end{minipage}


\begin{wrapfigure}[8]{r}{0pt}
\centering
\scalebox{.7}{
\begin{pgfpicture}
\begin{pgfscope}
\pgftransformcm{1.0}{0.0}{0.0}{1.0}{\pgfpoint{220.0bp}{3170.393701bp}}
\pgftext[left,base]{\mdseries\upshape\normalsize\color[rgb]{0.0,0.0,0.0}$a_2$}
\end{pgfscope}
\begin{pgfscope}
\pgftransformcm{1.0}{0.0}{0.0}{1.0}{\pgfpoint{200.0bp}{3170.393701bp}}
\pgftext[left,base]{\mdseries\upshape\normalsize\color[rgb]{0.0,0.0,0.0}$a_1$}
\end{pgfscope}
\begin{pgfscope}
\pgftransformcm{1.0}{0.0}{0.0}{1.0}{\pgfpoint{240.0bp}{3170.393701bp}}
\pgftext[left,base]{\mdseries\upshape\normalsize\color[rgb]{0.0,0.0,0.0}$a_3$}
\end{pgfscope}
\begin{pgfscope}
\pgftransformcm{1.0}{0.0}{0.0}{1.0}{\pgfpoint{260.0bp}{3170.393701bp}}
\pgftext[left,base]{\mdseries\upshape\normalsize\color[rgb]{0.0,0.0,0.0}$a_4$}
\end{pgfscope}
\begin{pgfscope}
\pgftransformcm{1.0}{0.0}{0.0}{1.0}{\pgfpoint{280.0bp}{3170.393701bp}}
\pgftext[left,base]{\mdseries\upshape\normalsize\color[rgb]{0.0,0.0,0.0}$a_5$}
\end{pgfscope}
\begin{pgfscope}
\pgftransformcm{1.0}{0.0}{0.0}{1.0}{\pgfpoint{300.0bp}{3170.393701bp}}
\pgftext[left,base]{\mdseries\upshape\normalsize\color[rgb]{0.0,0.0,0.0}$a_6$}
\end{pgfscope}
\begin{pgfscope}
\pgftransformcm{1.0}{0.0}{0.0}{1.0}{\pgfpoint{200.0bp}{3180.393701bp}}
\pgftext[left,base]{\mdseries\upshape\normalsize\color[rgb]{0.0,0.0,0.0}$\tau_1$}
\end{pgfscope}
\begin{pgfscope}
\pgftransformcm{1.0}{0.0}{0.0}{1.0}{\pgfpoint{220.0bp}{3180.393701bp}}
\pgftext[left,base]{\mdseries\upshape\normalsize\color[rgb]{0.0,0.0,0.0}$\tau_2$}
\end{pgfscope}
\begin{pgfscope}
\pgftransformcm{1.0}{0.0}{0.0}{1.0}{\pgfpoint{240.0bp}{3180.393701bp}}
\pgftext[left,base]{\mdseries\upshape\normalsize\color[rgb]{0.0,0.0,0.0}$\tau_3$}
\end{pgfscope}
\begin{pgfscope}
\pgftransformcm{1.0}{0.0}{0.0}{1.0}{\pgfpoint{260.0bp}{3180.393701bp}}
\pgftext[left,base]{\mdseries\upshape\normalsize\color[rgb]{0.0,0.0,0.0}$\tau_4$}
\end{pgfscope}
\begin{pgfscope}
\pgftransformcm{1.0}{0.0}{0.0}{1.0}{\pgfpoint{280.0bp}{3180.393701bp}}
\pgftext[left,base]{\mdseries\upshape\normalsize\color[rgb]{0.0,0.0,0.0}$\tau_5$}
\end{pgfscope}
\begin{pgfscope}
\pgftransformcm{1.0}{0.0}{0.0}{1.0}{\pgfpoint{300.0bp}{3180.393701bp}}
\pgftext[left,base]{\mdseries\upshape\normalsize\color[rgb]{0.0,0.0,0.0}$\tau_6$}
\end{pgfscope}
\begin{pgfscope}
\pgftransformcm{1.0}{0.0}{0.0}{1.0}{\pgfpoint{250.0bp}{3150.393701bp}}
\pgftext[left,base]{\mdseries\upshape\normalsize\color[rgb]{0.0,0.0,0.0}$\vdots$}
\end{pgfscope}
\begin{pgfscope}
\pgfsetlinewidth{1.0bp}
\pgfsetrectcap
\pgfsetmiterjoin
\pgfsetmiterlimit{10.0}
\pgfpathmoveto{\pgfpoint{290.0bp}{3148.393701bp}}
\pgfpathlineto{\pgfpoint{250.0bp}{3148.393701bp}}
\pgfpathlineto{\pgfpoint{250.0bp}{3136.393701bp}}
\pgfpathlineto{\pgfpoint{290.0bp}{3136.393701bp}}
\definecolor{strokepaint}{rgb}{0.0,0.0,0.0}\pgfsetstrokecolor{strokepaint}
\pgfusepath{stroke}
\end{pgfscope}
\begin{pgfscope}
\pgftransformcm{1.0}{0.0}{0.0}{1.0}{\pgfpoint{170.0bp}{3140.393701bp}}
\pgftext[left,base]{\mdseries\upshape\normalsize\color[rgb]{0.0,0.0,0.0}$i=4$}
\end{pgfscope}
\begin{pgfscope}
\pgftransformcm{1.0}{0.0}{0.0}{1.0}{\pgfpoint{295.0bp}{3140.393701bp}}
\pgftext[left,base]{\mdseries\upshape\normalsize\color[rgb]{0.0,0.0,0.0}$\stackrel{j}{\rightarrow}$}
\end{pgfscope}
\begin{pgfscope}
\pgfsetlinewidth{1.0bp}
\pgfsetrectcap
\pgfsetmiterjoin
\pgfsetmiterlimit{10.0}
\pgfpathmoveto{\pgfpoint{270.0bp}{3130.393701bp}}
\pgfpathlineto{\pgfpoint{230.0bp}{3130.393701bp}}
\pgfpathlineto{\pgfpoint{230.0bp}{3118.393701bp}}
\pgfpathlineto{\pgfpoint{270.0bp}{3118.393701bp}}
\definecolor{strokepaint}{rgb}{0.0,0.0,0.0}\pgfsetstrokecolor{strokepaint}
\pgfusepath{stroke}
\end{pgfscope}
\begin{pgfscope}
\pgftransformcm{1.0}{0.0}{0.0}{1.0}{\pgfpoint{170.0bp}{3122.393701bp}}
\pgftext[left,base]{\mdseries\upshape\normalsize\color[rgb]{0.0,0.0,0.0}$i=3$}
\end{pgfscope}
\begin{pgfscope}
\pgftransformcm{1.0}{0.0}{0.0}{1.0}{\pgfpoint{275.0bp}{3122.393701bp}}
\pgftext[left,base]{\mdseries\upshape\normalsize\color[rgb]{0.0,0.0,0.0}$\stackrel{j}{\rightarrow}$}
\end{pgfscope}
\begin{pgfscope}
\pgftransformcm{1.0}{0.0}{0.0}{1.0}{\pgfpoint{250.0bp}{3100.393701bp}}
\pgftext[left,base]{\mdseries\upshape\normalsize\color[rgb]{0.0,0.0,0.0}$\vdots$}
\end{pgfscope}
\end{pgfpicture}
}
\caption{$i,j$ in our algorithms for timed pattern
 matching}
\label{fig:orderOfIJInTheNaiveAlgorithm}
\end{wrapfigure}
Our first (naive) algorithm for timed pattern matching is 
in
Alg.~\ref{naive_alg}.
We conduct a brute-force breadth-first search, computing
$\bigl\{\,(t,t')\in\mathcal{M}(w,\mathcal{A})\,\bigl|\bigr.\,
  \tau_{i-1}\le t < \tau_{i}, \tau_{j} < t'\le \tau_{j+1}\,\bigr\}$
for each $i,j$, with the aid of $\Conf(i,j)$ in
Def.~\ref{def:alg_funcs}. 
(The singular case of $\forall i.\, \tau_{i}\not\in(t,t')$ is separately
taken care
of  by $\Init$.)
We do so  in the order illustrated in
Fig.~\ref{fig:orderOfIJInTheNaiveAlgorithm}: we decrement $i$, and for
each $i$ we increment $j$. This order---that flips the one in
Fig.~\ref{fig:bruteforce}---is for the purpose of the Boyer-Moore
type optimization later in~\S{}\ref{201300_31Jan16}. In
Appendix~\ref{appendix:futherExplOfNaiveAlgorithm} we provide further details.



\vspace{.3em}
\noindent
\begin{minipage}{\textwidth}
\begin{mytheorem}[termination and correctness of Alg.~\ref{naive_alg}]
\label{thm:terminationAndCorrectnessOfNaiveAlg}
\begin{enumerate}
 \item Alg.~\ref{naive_alg} terminates and its answer $Z$
is a finite union of zones. 
 \item  For any $t, t' \in \mathbb{R}_{>0}$ such that $t < t'$, the following
 are equivalent: 1) there is a zone $(T_0,T_f,T_{\Delta}) \in Z$ such that
        $(t,t') \in (T_0,T_f,T_{\Delta})$; and 2)
there is an accepting run $(\overline{s},\overline{\nu})$ over
        $w|_{(t,t')}$ of $\mathcal{A}$.
  \myqed
\end{enumerate}
\end{mytheorem}
\end{minipage}

%

As an immediate corollary, we conclude that a match set
 $\mathcal{M} (w,\mathcal{A})$ always allows representation by finitely
 many zones.

Changing the order of examination of $i,j$
(Fig.~\ref{fig:orderOfIJInTheNaiveAlgorithm}) gives us an \emph{online}
variant of Alg.~\ref{naive_alg}. It is presented in 
Appendix~\ref{appendix:onlinenNaiveAlgorithm}; nevertheless our Boyer-Moore type
algorithm is based on the original
Alg.~\ref{naive_alg}.

\section{A Timed Boyer-Moore Type Algorithm}
\label{201300_31Jan16}
Here we describe our main contribution, namely a Boyer-Moore type
algorithm for timed pattern matching. Much like the algorithm
in~\cite{Watson2003} skips some $j$'s in Fig.~\ref{fig:bruteforce} (\S{}\ref{subsec:patternMatchingAndWatsonWatsonAlgorithm}), we
wish to skip some $i$'s in Fig.~\ref{fig:orderOfIJInTheNaiveAlgorithm}. 
Let us henceforth fix a timed word $w = (\overline{a},\overline{\tau})$ and a timed automaton 
$\mathcal{A} = (\Sigma \amalg \{\$\},S,S_0,C,E,F)$ as the input of
the problem.


Let us define the \emph{optimal} skip value function by
 $\Opt (i) = \min \{n \in \mathbb{R}_{>0} \mid \exists t \in
 [\tau_{i-n-1},\tau_{i-n}).\, \exists t' \in (t,\infty).\,(t,t') \in \mathcal{M}
 (w,\mathcal{A})\}\label{optimum_amount_of_decrease}$;
the value $\Opt(i)$ designates the biggest skip value, at each $i$ in
 Fig.~\ref{fig:orderOfIJInTheNaiveAlgorithm}, that does not change the
 outcome of the algorithm. Since the function  $\Opt$ is not amenable to 
efficient computation in general, our goal is its underapproximation 
that is easily computed. 

Towards that goal we follow the (untimed) pattern matching
 algorithm in~\cite{Watson2003};
 see~\S{}\ref{subsec:patternMatchingAndWatsonWatsonAlgorithm}. 
 In applying the same idea as in
 Fig.~\ref{fig:tableForDelta2InPatternMatching} to define a skip value, however, the first obstacle
 is that the language $L'_{s_{3}}$---the set of (suitable prefixes of) all words
 that lead to $s_{3}$---becomes an \emph{infinite} set in the current timed
 setting. Our countermeasure is to use a \emph{region automaton}
 $R(\mathcal{A})$ (Def.~\ref{def:regionAutom}) for representing the set.



We shall first introduce some constructs used in our algorithm.
\vspace{-1em}
\begin{mydefinition}[$\mathcal{W}(r),\mathcal{W}(\overline{s},\overline{\alpha})$]
\label{def:mathcalW}
 Let $r$ be a set of runs of the timed automaton $\mathcal{A}$. We
 define a timed language $\mathcal{W}(r)$  by:
\begin{math}
 \mathcal{W}(r)=
  \bigl\{\,(\overline{a},\overline{\tau})\,\bigl|\bigr.\,\text{in $r$ there is a run of
 $\mathcal{A}$ over $(\overline{a},\overline{\tau})$}\,\bigr\}
\end{math}.

 For the region automaton $R(\mathcal{A})$, each run 
$(\overline{s},\overline{\alpha})$ of $R(\mathcal{A})$---where $s_{k}\in S$ and
 $\alpha_{k}\in (\mathbb{R}_{\geq 0})^{C}/{\sim}$, recalling the state
 space of $R(\mathcal{A})$ from Def.~\ref{def:regionAutom}---is
 naturally identified with a set of runs of $\mathcal{A}$, namely
$\{ (\overline{s},\overline{\nu})\in \bigl(S\times
(\mathbb{R}_{\geq 0})^{C}\bigr)^{*}  \mid
 \nu_k \in \alpha_k \text{ for each $k$}\}$. Under this identification
 we shall sometimes write $\mathcal{W}(\overline{s},\overline{\alpha})$ for
 a suitable timed language, too.

 The above definitions of $\mathcal{W}(r)$ and
 $\mathcal{W}(\overline{s},\overline{\alpha})$ naturally extends to a
 set $r$ of \emph{partial runs} of $\mathcal{A}$, and to a \emph{partial
 run}  $(\overline{s},\overline{\alpha})$ of $R(\mathcal{A})$,
 respectively. Here a partial run is a run but we do not require:
 it 
 start at an initial state; or its initial clock interpretation be $0$. 
\end{mydefinition}
\vspace{-1em}
The next optimization of  $R(\mathcal{A})$ is similar to so-called \emph{trimming},
but we leave those states that do not lead to any final state (they become
necessary later).

\vspace{.3em}
\noindent
\begin{minipage}{\textwidth}
\begin{mydefinition}[$\noredundancy{R}(\mathcal{A})$]
\label{def:regionAutomWithoutRedundancy}
For a timed automaton $\mathcal{A}$, we let
$\noredundancy{R}(\mathcal{A})$ denote its \emph{reachable region automaton}. It is the NFA
 $\noredundancy{R}(\mathcal{A})=(\Sigma,
\noredundancy{S},
\noredundancy{S_{0}},
\noredundancy{E},
\noredundancy{F})$
 obtained from 
$R(\mathcal{A})$ (Def.~\ref{def:regionAutom}) by removing all those
 states which are unreachable from any initial state.
\end{mydefinition}
\end{minipage}


%

We are ready to describe our Boyer-Moore type algorithm.
We use a skip value function $\Delta_{2}$ that is similar to
the one in~\S{}\ref{subsec:patternMatchingAndWatsonWatsonAlgorithm} (see
Fig.~\ref{192552_1Feb16} \&~\ref{fig:tableForDelta2InPatternMatching}),
computed with the aid of 
$m_{s}$ and $L'_{s}$ defined for each state $s$. We define
$m_{s}$ and $L'_{s}$ 
using the NFA $\noredundancy{R}(\mathcal{A})$. Notable differences  are:
1) here $L'_{s}$ and $L'$ are sets of \emph{runs}, not  of
\emph{words}; and 2) since the orders are flipped between
Fig.~\ref{fig:bruteforce} \&~\ref{fig:orderOfIJInTheNaiveAlgorithm}, 
$\Rev$ e.g.\ in~(\ref{eq:LPrimeExample}) is gone now.

The precise definitions are as follows. Here 
$s \in S$ is a state of the (original) timed automaton $\mathcal{A}$;
and we let
\begin{math}
 \noredundancy{R} (s) = \{ (s,\alpha) \in \noredundancy{S}
   \}
\end{math}.
\begin{equation}\label{eq:mLPrimeForTimedBoyerMoore}
 \small
 \begin{aligned}
 m &= \min
  \{|w'|\mid w' \in L (\mathcal {A})\}
 \quad
 m_s = 
\min \bigl\{\,|r| \,\bigl|\bigr.\,
 \beta_0 \in \noredundancy{S_{0}}, \beta \in \noredundancy{R}(s),
 r \in \Run_{\noredundancy{R} (\mathcal{A})} (\beta_0,\beta)\, \bigr\}
 \!\!\! \!\!\! \!\!\! \!\!\! \!\!\! \!\!\! \!\!\! \!\!\! \!\!\!
\\
 L' &= \bigl\{r (0,m-1) \mid \beta_0 \in \noredundancy{S_{0}}, \beta_f \in \noredundancy{F}, 
  r \in \Run_{\noredundancy{R} (\mathcal{A})} (\beta_0,\beta_f)\bigr\}
\\
  L'_s &= \bigl\{\,r (0,\min\{m,m_s\}-1) \,\bigl|\bigr.\,
 \beta_0 \in \noredundancy{S_{0}}, \beta \in \noredundancy{R}(s), 
 r \in \Run_{\noredundancy{R}(\mathcal{A})} (\beta_0,\beta)\,\bigr\}
\end{aligned}
\end{equation}
\begin{algorithm}[H]
 \caption{Our naive algorithm for timed pattern matching.
 See~Def.~\ref{def:alg_funcs} and
 Appendix~\ref{appendix:futherExplOfNaiveAlgorithm} 
 for details.
%
}
 \label{naive_alg}
 \scalebox{1.0}{
 \parbox{\textwidth}{
 \begin{algorithmic}[1]
  \Require A timed word $w = (\overline{a},\overline{\tau})$, and a
  timed automaton $\mathcal{A} = (\Sigma,S,S_0,C,E,F)$.
  \Ensure $\bigcup Z$ is the match set $\mathcal{M} (w,\mathcal{A})$ in
  Def.~\ref{TimedPatternMatching}.
  \State $i \gets |w|;\; \CurrConf \gets \emptyset;\;
  \Init\gets\emptyset;\; Z\gets\emptyset$
  \label{naive_alg:initialization}
  \Comment{$\Init$ and $Z$ are finite sets of zones.}
  \For{$s \in S_0$} \label{naive_alg:computingInitBegin}
  \Comment{Lines~\ref{naive_alg:computingInitBegin}--\ref{naive_alg:computingInitEnd}
  compute $\Init$.}
  \For{$s_f \in F$} 
  \For{$(s,s_f,\$,\lambda,\delta) \in E$} 
  \State $\Init \gets \Init \cup \solConstr
  ([0,\infty),(0,\infty),\rhoEmpty,\delta)$ 
         \label{naive_alg:computingInitEnd}
  \EndFor
  \EndFor 
  \EndFor
  \State $Z \gets Z \cup \bigl\{\,(T_0 \cap
  [\tau_{|w|},\infty),T_f \cap (\tau_{|w|},\infty),T_{\Delta}) 
  \,\bigl|\bigr.\,
  (T_0,T_f,T_\Delta) \in \Init\,\bigr\}$
  \State{}
  \Comment{We have added
  $\bigl\{\,(t,t')\in\mathcal{M}(w,\mathcal{A})\,\bigl|\bigr.\,
  t,t'\in[\tau_{|w|},\infty)
 \,\bigr\}$ to $Z$.}
  \While{$i > 0$}
  \State $Z \gets Z \cup \{(T_0 \cap
  [\tau_{i-1},\tau_{i}),T_f \cap (\tau_{i-1},\tau_{i}],T_{\Delta})  \mid
  (T_0,T_f,T_\Delta) \in \Init\}$
  \State \Comment{We have added $\bigl\{\,(t,t')\in\mathcal{M}(w,\mathcal{A})\,\bigl|\bigr.\,
  t,t'\in[\tau_{i-1},\tau_{i}]
\,\bigr\}$ to $Z$.}
  \label{naive_alg:specialCaseEnd}
  \State \Comment{Now, for each $j$, we shall add
  $\bigl\{\,(t,t')\in\mathcal{M}(w,\mathcal{A})\,\bigl|\bigr.\,
  t\in[\tau_{i-1},\tau_{i}), t'\in (\tau_{j},\tau_{j+1}]
\,\bigr\}$.}
  \State $j \gets i$
  \State $\CurrConf \gets \{(s,\rhoEmpty,[\tau_{i-1},\tau_i)) \mid s \in S_0\}$
  \label{naive_alg:CurrConfInitialize}
  \While{$\CurrConf \neq \emptyset \And j \leq |w|$}
  \State $(\PrevConf , \CurrConf) \gets (\CurrConf , \emptyset)$
  \Comment{Here $\PrevConf = \Conf (i,j-1)$.}
  \For{$(s,\rho,T) \in \PrevConf$}
  \For{$(s,s',a_j,\lambda,\delta) \in E$}
  \Comment{Read $(a_{j},\tau_{j})$.} \label{naive_alg:readWordBegin}
  \State $T' \gets \{t_0 \in T \mid \eval(\rho,\tau_j,t_0) \models \delta\}$
  \label{naive_alg:update_start_interval}
  \State
  \Comment  {Narrow the interval $T$ to satisfy the clock constraint $\delta$.}
  \If{$T' \neq \emptyset$}
  \State $\rho' \gets \rho$

  \For{$x \in \lambda$}
  \State $\rho' \gets \reset(\rho',x,\tau_j)$ \label{naive_alg:reset_clock}
  \Comment {Reset the clock variables in $\lambda$.}
  \EndFor

  \State $\CurrConf \gets \CurrConf \cup (s',\rho',T')$\label{naive_alg:readWordEnd}

  \For{$s_f \in F$} \label{naive_alg:insertTermBegin}\label{naive_alg:start_of_accepting_state}
  \Comment{Lines~\ref{naive_alg:insertTermBegin}--\ref{naive_alg:insertTermEnd}
  try to insert \$ in $(\tau_{j},\tau_{j+1}]$.}
  \For{$(s',s_f,\$,\lambda',\delta') \in E$} \label{naive_alg:loop_over_e_2}
  \If {$j = |w|$}
  \State $T'' \gets (\tau_j,\infty)$
  \Else
  \State $T'' \gets (\tau_j,\tau_{j+1}]$
  \EndIf
  \State $Z \gets Z \cup \solConstr(T',T'',\rho',\delta')$
  \label{naive_alg:add_interval} \label{naive_alg:insertTermEnd}
  \EndFor
  \EndFor \label{naive_alg:end_of_accepting_state}
  \EndIf
  \EndFor
  \EndFor

  \State $j \gets j + 1$
  \EndWhile
  \State $i \gets i - 1$ \label{naive_alg:decrement_i}
  \EndWhile
 \end{algorithmic}}}
\end{algorithm}

Note again that these data are defined over
$\noredundancy{R}(\mathcal{A})$
(Def.~\ref{def:regionAutomWithoutRedundancy}); 
 $\Run_{\noredundancy{R}(\mathcal{A})} (\beta_0,\beta_f)$ 
 is from Def.~\ref{def:runNFA}.

\begin{figure}[tbp]
 \begin{minipage}{.5\textwidth}
   \scalebox{0.65}{
 \begin{tikzpicture}[shorten >=1pt,node distance=2cm,on grid,auto] 
  \node[state,initial] (s_0)   {$s_0$}; 
  \node[state] (s_1) [right=of s_0] {$s_1$}; 
  \node[state] (s_2) [above right=of s_1] {$s_2$};
  \node[state,node distance=3cm] (s_3) [right=of s_1] {$s_3$};
  \node[accepting,state] (s_4) [right=of s_3] {$s_4$}; 
  \path[->] 
  (s_0) edge  [below] node {
  \begin{tabular}{c}
   $a,\mathbf{true}$\\
   $/y:=0$\\
  \end{tabular}
  } (s_1)
  (s_1) edge node  {$b,y=1$} (s_2)
  (s_1) edge [above,bend left=5,bend right=-5] node {$c,x<1$} (s_3) 
  (s_3) edge [below,bend left=5,bend right=-5] node {
  \begin{tabular}{c}
   $a,y<1$\\
   $/y:=0$\\
  \end{tabular}
  } (s_1) 
  (s_2) edge  node {$c,x<1$} (s_3) 
  (s_3) edge [loop below] node {$d,x>1$}(s_3)
  (s_3) edge node {$\$,x=1$} (s_4);
 \end{tikzpicture}
 }

\vspace{1em}
{\scriptsize  
\begin{minipage}{\textwidth}
 Here
 $m_{s_0}=1$, $m_{s_1}=2$, $m_{s_2}=3$, $m_{s_3}=3$, $m_{s_4}=4$
 and   $m=3$.
\end{minipage}
}
  \scalebox{.85}{ \begin{minipage}{.5\textwidth}
  \begin{align*}
   L' &=
 \left\{
 \begin{array}{l}
  \scalebox{0.70}{
  \tikz[baseline=(s_0.base),auto,node distance=0.5cm]{
   \regionstate{s_0}{$s_0$}{$x=y=0$};
   \regionstate[right=of s_0]{s_1}{$s_1$}{$y=0<x<1$};
   \regionstate[right=of s_1]{s_3}{$s_3$}{$0<y<x<1$};
   \path[->](s_0) edge  node {$a$} (s_1);
   \path[->](s_1) edge  node {$c$} (s_3);
   }}
 \end{array} 
\right\},\\
   L'_{s_3} &=
 \left\{
 \begin{array}{l}
  \scalebox{0.70}{
  \tikz[baseline=(s_0.base),auto,node distance=0.5cm]{
   \regionstate{s_0}{$s_0$}{$x=y=0$};
   \regionstate[right=of s_0]{s_1}{$s_1$}{$y=0<x<1$};
   \regionstate[right=of s_1]{s_3}{$s_3$}{$0<y<x<1$};
   \path[->](s_0) edge  node {$a$} (s_1);
   \path[->](s_1) edge  node {$c$} (s_3);
   }}
 \end{array}
 \right\}\\
   L'_{s_4} &=
 \left\{
 \begin{array}{l}
  \scalebox{0.70}{
  \tikz[baseline=(s_0.base),auto,node distance=0.5cm]{
   \regionstate{s_0}{$s_0$}{$x=y=0$};
   \regionstate[right=of s_0]{s_1}{$s_1$}{$y=0<x<1$};
   \regionstate[right=of s_1]{s_3}{$s_3$}{$0<y<x<1$};
   \path[->](s_0) edge  node {$a$} (s_1);
   \path[->](s_1) edge  node {$c$} (s_3);
   }}
 \end{array} 
\right\}\\
  \end{align*}
 \end{minipage}}
 \end{minipage}
 \hfill
\scalebox{.85}{ \begin{minipage}{.5\textwidth}
\begin{align*}
  L'_{s_0} &=
  \left\{
  \scalebox{0.70}{
  \tikz[baseline=(s_0.base),auto,node distance=0.5cm]{
  \regionstate{s_0}{$s_0$}{$x=y=0$};
  }}
 \right\}\\
 L'_{s_1} &=
 \left\{
 \begin{array}{l}
  \scalebox{0.70}{
  \tikz[baseline=(s_0.base),auto,node distance=0.5cm]{
   \regionstate{s_0}{$s_0$}{$x=y=0$};
   \regionstate[right=of s_0]{s_1}{$s_1$}{$y=0<x<1$};
   \path[->](s_0) edge  node {$a$} (s_1);
   }}, \\
   \scalebox{0.70}{
   \tikz[baseline=(s_0.base),auto,node distance=0.5cm]{
   \regionstate{s_0}{$s_0$}{$x=y=0$};
   \regionstate[right=of s_0]{s_1}{$s_1$}{$y=0,x=1$};
   \path[->](s_0) edge  node {$a$} (s_1);
   }},\\ 
  \scalebox{0.70}{
   \tikz[baseline=(s_0.base),auto,node distance=0.5cm]{
   \regionstate{s_0}{$s_0$}{$x=y=0$};
   \regionstate[right=of s_0]{s_1}{$s_1$}{$y=0,1<x$};
   \path[->](s_0) edge  node {$a$} (s_1);
   }}
 \end{array}
 \right\}\\ 
 L'_{s_2} &=
 \left\{
 \begin{array}{l}
  \scalebox{0.70}{
  \tikz[baseline=(s_0.base),auto,node distance=0.5cm]{
   \regionstate{s_0}{$s_0$}{$x=y=0$};
   \regionstate[right=of s_0]{s_1}{$s_1$}{$y=0<x<1$};
   \regionstate[right=of s_1]{s_3}{$s_3$}{$0<y<x<1$};
   \path[->](s_0) edge  node {$a$} (s_1);
   \path[->](s_1) edge  node {$c$} (s_3);
   }}, \\
  \scalebox{0.70}{
  \tikz[baseline=(s_0.base),auto,node distance=0.5cm]{
   \regionstate{s_0}{$s_0$}{$x=y=0$};
   \regionstate[right=of s_0]{s_1}{$s_1$}{$y=0<x<1$};
   \regionstate[right=of s_1]{s_2}{$s_2$}{$y=1<x$};
   \path[->](s_0) edge  node {$a$} (s_1);
   \path[->](s_1) edge  node {$b$} (s_2);
   }}, \\
   \scalebox{0.70}{
   \tikz[baseline=(s_0.base),auto,node distance=0.5cm]{
   \regionstate{s_0}{$s_0$}{$x=y=0$};
   \regionstate[right=of s_0]{s_1}{$s_1$}{$y=0,x=1$};
   \regionstate[right=of s_1]{s_2}{$s_2$}{$y=1<x$};
   \path[->](s_0) edge  node {$a$} (s_1);
   \path[->](s_1) edge  node {$b$} (s_2);
   }},\\ 
  \scalebox{0.70}{
   \tikz[baseline=(s_0.base),auto,node distance=0.5cm]{
   \regionstate{s_0}{$s_0$}{$x=y=0$};
   \regionstate[right=of s_0]{s_1}{$s_1$}{$y=0,1<x$};
   \regionstate[right=of s_1]{s_2}{$s_2$}{$y=1<x$};
   \path[->](s_0) edge  node {$a$} (s_1);
   \path[->](s_1) edge  node {$b$} (s_2);
   }}
 \end{array}
 \right\}\\ 
 \end{align*}
 \end{minipage}}

\vspace{-1.5em}
 \caption{An example of a timed automaton $\mathcal{A}$, and the values
 $m_{s}, L'_{s}, L'$}
 \label{fig:exampleOfATimedAutom}
\end{figure}

\vspace*{.5em}
\noindent
\begin{minipage}{\textwidth}
 \begin{mydefinition}[$\Delta_{2}$]
 \label{def:Delta2ForTimedBoyerMoore}
 Let $\Conf$ be a set of triples $(s,\rho,T)$ of: a state $s\in S$ of
 $\mathcal{A}$, $\rho\colon C\rightharpoonup \mathbb{R}_{>0}$, and an
 interval $T$. (This is much like $\Conf(i,j)$ in Def.~\ref{def:alg_funcs}.) We define 
 the skip value  $\Delta_2 (\Conf)$  as follows. 
 \end{mydefinition}
\end{minipage}
 \begin{align*}
 \begin{array}{rl}
   d_1 (r) &= \min_{r' \in L'} \min \{ n \in \mathbb{Z}_{>0} \mid
  \mathcal{W}(r) \cap
  \bigl(\,\bigcup_{r''\in \pref \bigl(r'(n,|r'|)\bigr)}\mathcal{W}(r'') \,\bigr)
 \neq \emptyset\,\}\\  
  d_2 (r) &= \min_{r' \in L'} \min \{ n \in \mathbb{Z}_{>0} \mid
  \bigl(\bigcup_{r''\in \pref(r)}\mathcal{W}(r'')\bigr)
  \cap \mathcal{W}\bigl(r'(n,|r'|)\bigr) \neq \emptyset\,\}\\
  \Delta_2 (\Conf) &= \max_{(s,\rho,T) \in \Conf}
  \min_{r \in L'_s} \min\{d_1 (r),d_2 (r)\}
 \end{array} 
 \end{align*}
 Here $r\in \Run_{\noredundancy{R}(\mathcal{A})}$; $L'$ is
 from~(\ref{eq:mLPrimeForTimedBoyerMoore}); $\mathcal{W}$ is from
 Def.~\ref{def:mathcalW}; $r'(n,|r'|)$ is a subsequence
 of $r'$  (that is a partial run); and $\pref(r)$ denotes the set of all
 prefixes of $r$. 

\noindent
\begin{minipage}{\textwidth}
\begin{mytheorem}[correctness of $\Delta_{2}$]
 \label{thm:Delta2_is_skipvalue}
 Let $i\in [1,|w|\,]$, and $j=\max\{j\in [i,|w|\,]\mid
 \Conf(i,j)\neq\emptyset\}$, where $\Conf(i,j)$ is from Def.~\ref{def:alg_funcs}. (In case $\Conf(i,j)$ is everywhere empty
 we let $j=i$.) Then we have $  \Delta_2(\Conf (i,j))\le \Opt (i)$.
 \qed
\end{mytheorem}
\end{minipage}

\vspace*{.5em}


The remaining issue in Def.~\ref{def:Delta2ForTimedBoyerMoore} is that
the sets like $\mathcal{W}(r)$ and $\mathcal{W}(r'')$ can be
infinite---we need to avoid their direct computation.  We rely on the
usual automata-theoretic trick: the intersection of languages is
recognized by a product automaton. 

Given two timed automata
$\mathcal{B}$
and $\mathcal{C}$, we let $\mathcal{B}\times\mathcal{C}$ denote their
\emph{product} defined in the standard way (see e.g.~\cite{PandyaS12}). 
The following is straightforward.

\begin{myproposition}
\label{prop:productAutomAndPartialRun}
 Let $r = (\overline{s},\overline{\alpha})$ and $r' =
 (\overline{s'},\overline{\alpha'})$ be partial runs of
 $R(\mathcal{B})$ and $R(\mathcal{C})$, respectively; 
 they are naturally identified with sets of partial runs  of 
 $\mathcal{B}$ and $\mathcal{C}$ (Def.~\ref{def:mathcalW}). 
 Assume further that $|r| =
 |r'|$. 
 Then we have 
 $\mathcal{W} (r) \cap \mathcal{W} (r') = \mathcal{W}
 (r,r')$, where $(r,r')$ is the
 following
 set of runs of $\mathcal{B}\times \mathcal{C}$:
 $(r,r') = \bigl\{\,
 \bigl((\overline{s},\overline{s'}),(\overline{\nu} , \overline{\nu'})\bigr) \mid
 (\overline{s},\overline{\nu}) \in (\overline{s},\overline{\alpha})
 \text{ and  } (\overline{s'},\overline{\nu'}) \in
 (\overline{s'},\overline{\alpha'}) \,\bigr\}$.
\myqed
\end{myproposition}
The proposition allows the following algorithm for the emptiness check
required in computing $d_{1}$
(Def.~\ref{def:Delta2ForTimedBoyerMoore}).  Firstly we distribute $\cap$
over $\bigcup$; then we need to check if
$\mathcal{W}(r)\cap\mathcal{W}(r'')\neq\emptyset$ for each $r''$. The
proposition reduces this
 to checking if 
$\mathcal{W}(r,r'')\neq\emptyset$, that is, if $(r,r'')$ is a
(legitimate) partial run of the region automaton 
 $R (\mathcal{A} \times \mathcal{A})$. The last check is obviously
 decidable since $R (\mathcal{A} \times \mathcal{A})$ is finite. For
 $d_{2}$ the situation is similar. 

We also note that the computation of $\Delta_{2}$
(Def.~\ref{def:Delta2ForTimedBoyerMoore}) can be accelerated by
memorizing the values $\min_{r \in
L'_s}\min\{d_1 (r),d_2 (r)\}$ for each $s$. 

Finally our Boyer-Moore type algorithm for timed pattern matching is 
Alg.~\ref{boyer-moore_alg}  in Appendix~\ref{appendix:timedBoyerMoore}. 
Its main differences from the naive one (Alg.~\ref{naive_alg}) are: 1)
initially we start with $i=|w|-m+1$ instead of $i=|w|$
(line~\ref{naive_alg:initialization} of Alg.~\ref{naive_alg}); and 2) 
we decrement $i$ by the skip value computed by $\Delta_{2}$, instead of
by $1$ (line~\ref{naive_alg:decrement_i} of Alg.~\ref{naive_alg}).


It is also possible to employ an analog of the 
skip value function $\Delta_1$ in
\S{}\ref{subsec:stringMatchingAndOriginalBoyerMoore}
\&~\ref{subsec:patternMatchingAndWatsonWatsonAlgorithm}. 
For $c\in\Sigma$ and $p\in \mathbb{Z}_{>0}$, we define
\begin{math}
   \Delta_1 (c,p) = 
  \min_{k > 0} \{k - p \mid k > m \text{ or }
  \exists (\overline{a},\overline{\tau}) \in \mathcal{W}(L').\, a_k = c \}
\end{math}. Here $m$ and $L'$ are
  from~(\ref{eq:mLPrimeForTimedBoyerMoore}). Then we can possibly skip more $i$'s
  using both $\Delta_{1}$ and $\Delta_{2}$; see
  Appendix~\ref{appendix:timedBoyerMoore} for details. In our
  implementation we do not use $\Delta_{1}$, though, following the
  (untimed) pattern matching algorithm in~\cite{Watson2003}. 
 Investigating the effect of additionally using $\Delta_{1}$ is future work.


One may think of the following alternative for pattern matching: we
first forget about time stamps, time constraints, etc.; the resulting
``relaxed'' untimed problem can be solved by the algorithm
in~\cite{Watson2003}
(\S{}\ref{subsec:patternMatchingAndWatsonWatsonAlgorithm}); and then 
we introduce the time constraints back and refine the interim result to the
correct one.  Our \emph{timed} Boyer-Moore algorithm has greater skip
values in general, however, because by using region automata
$R(\mathcal{A}),\noredundancy{R}(\mathcal{A})$ we also take time
constraints 
into account when computing skip values.
 

\marginpar{Ichiro added this remark, instead of adding a paragraph in
``Future Work''}
\begin{myremark}\label{rem:zoneAutom}
 It was suggested by multiple reviewers that our use of region
 automata be replaced with that
 of \emph{zone automata} (see e.g.~\cite{BehrmannBLP06}). This 
 can result in 
a much smaller automaton $R(\mathcal{A})$ for
 calculating skip values (cf.\ Def.~\ref{def:Delta2ForTimedBoyerMoore}
 and Case~2 of~\S{}\ref{222310_24Jan16}).
 More importantly, zone automata are insensitive to the time
 unit size---unlike region automata where the numbers $c_{x}$ in
 Def.~\ref{def:regionAutom} govern their size---a property desired in 
 actual deployment of timed pattern matching. This is a topic of our
 imminent
 future work.
\end{myremark}

 \section{Experiments}
 \label{222310_24Jan16}
 We implemented both of our naive offline algorithm and our
 Boyer-Moore type algorithm (without $\Delta_1$) in C++ \cite{TimedBoyerMooreSample}.
 We ran our experiments on MacBook Air Mid 2011 with Intel Core
 i7-2677M 1.80GHz CPU with 3.7GB RAM and
 Arch Linux (64-bit).
 Our programs were compiled with GCC 5.3.0 with optimization level O3.

 An execution of our  Boyer-Moore type algorithm consists of two phases: 
 in the first \emph{pre-processing} phase we compute the skip value function
 $\Delta_{2}$---to be precise the value
 $\min_{r \in L'_s}\min\{d_1 (r),d_2 (r)\}$ for each $s$, on which
 $\Delta_{2}$
 relies---and in the latter ``matching'' phase we actually compute the match set. 

 As input we used five test cases: each case consists of 
a timed automaton $\mathcal{A}$ and multiple timed words $w$ of varying
length $|w|$. 
 Cases~1 \&~4 are from a previous
 work~\cite{Ulus2014} on timed pattern matching; in Case~2 the timed
 automaton $\mathcal{A}$ is not expressible with a timed
 regular expression (TRE, see~\cite{Herrmann1999}
 and~\S{}\ref{221705_24Jan16}); and Cases~3 \&~5 are our original. In
 particular Case~5 comes from an automotive example.

Our principal interest is in the
 relationship between execution time and the length $|w|$ (i.e.\ the
 number of events), for both of the naive and Boyer-Moore algorithms. 
 For each test case we ran our programs 30 times; the presented
 execution time is the average. 
 We measured the execution time separately for the
 pre-processing phase and the (main) matching phase; in the figures we
 present the time for the latter.

 We present an overview of the results. Further
 details are found in 
 Appendix~\ref{appendix:detailedResultsOfOurExperiments}.



\begin{figure}[tbh]
 \begin{minipage}[b]{.5\textwidth}
  \centering
   \scalebox{0.7}{
 \begin{tikzpicture}[shorten >=1pt,node distance=2cm,on grid,auto]
 \node[state,initial] (s_0) {$s_0$};
 \node[state] (s_1) [right of=s_0] {$s_1$};
 \node[state] (s_2) [right of=s_1] {$s_2$};
 \node[state,accepting] (s_3) [right of=s_2]{$s_3$};

 \path[->] 
  (s_0) edge [above] node {$a,\mathbf{true}$} (s_1)
  (s_1) edge [above] node {$b,\mathbf{true}$} (s_2)
  (s_2) edge [above] node {$\$,\mathbf{true}$} (s_3);
 \end{tikzpicture}}

    \scalebox{0.50}{
   \begin{tikzpicture}[gnuplot]
\path (0.000,0.000) rectangle (12.500,8.750);
\gpcolor{color=gp lt color border}
\gpsetlinetype{gp lt border}
\gpsetdashtype{gp dt solid}
\gpsetlinewidth{1.00}
\draw[gp path] (1.136,0.985)--(1.316,0.985);
\draw[gp path] (11.947,0.985)--(11.767,0.985);
\node[gp node right] at (0.952,0.985) {$0$};
\draw[gp path] (1.136,1.807)--(1.316,1.807);
\draw[gp path] (11.947,1.807)--(11.767,1.807);
\node[gp node right] at (0.952,1.807) {$10$};
\draw[gp path] (1.136,2.629)--(1.316,2.629);
\draw[gp path] (11.947,2.629)--(11.767,2.629);
\node[gp node right] at (0.952,2.629) {$20$};
\draw[gp path] (1.136,3.450)--(1.316,3.450);
\draw[gp path] (11.947,3.450)--(11.767,3.450);
\node[gp node right] at (0.952,3.450) {$30$};
\draw[gp path] (1.136,4.272)--(1.316,4.272);
\draw[gp path] (11.947,4.272)--(11.767,4.272);
\node[gp node right] at (0.952,4.272) {$40$};
\draw[gp path] (1.136,5.094)--(1.316,5.094);
\draw[gp path] (11.947,5.094)--(11.767,5.094);
\node[gp node right] at (0.952,5.094) {$50$};
\draw[gp path] (1.136,5.916)--(1.316,5.916);
\draw[gp path] (11.947,5.916)--(11.767,5.916);
\node[gp node right] at (0.952,5.916) {$60$};
\draw[gp path] (1.136,6.737)--(1.316,6.737);
\draw[gp path] (11.947,6.737)--(11.767,6.737);
\node[gp node right] at (0.952,6.737) {$70$};
\draw[gp path] (1.136,7.559)--(1.316,7.559);
\draw[gp path] (11.947,7.559)--(11.767,7.559);
\node[gp node right] at (0.952,7.559) {$80$};
\draw[gp path] (1.136,8.381)--(1.316,8.381);
\draw[gp path] (11.947,8.381)--(11.767,8.381);
\node[gp node right] at (0.952,8.381) {$90$};
\draw[gp path] (1.136,0.985)--(1.136,1.165);
\draw[gp path] (1.136,8.381)--(1.136,8.201);
\node[gp node center] at (1.136,0.677) {$0$};
\draw[gp path] (2.938,0.985)--(2.938,1.165);
\draw[gp path] (2.938,8.381)--(2.938,8.201);
\node[gp node center] at (2.938,0.677) {$20$};
\draw[gp path] (4.740,0.985)--(4.740,1.165);
\draw[gp path] (4.740,8.381)--(4.740,8.201);
\node[gp node center] at (4.740,0.677) {$40$};
\draw[gp path] (6.542,0.985)--(6.542,1.165);
\draw[gp path] (6.542,8.381)--(6.542,8.201);
\node[gp node center] at (6.542,0.677) {$60$};
\draw[gp path] (8.343,0.985)--(8.343,1.165);
\draw[gp path] (8.343,8.381)--(8.343,8.201);
\node[gp node center] at (8.343,0.677) {$80$};
\draw[gp path] (10.145,0.985)--(10.145,1.165);
\draw[gp path] (10.145,8.381)--(10.145,8.201);
\node[gp node center] at (10.145,0.677) {$100$};
\draw[gp path] (11.947,0.985)--(11.947,1.165);
\draw[gp path] (11.947,8.381)--(11.947,8.201);
\node[gp node center] at (11.947,0.677) {$120$};
\draw[gp path] (1.136,8.381)--(1.136,0.985)--(11.947,0.985)--(11.947,8.381)--cycle;
\node[gp node center,rotate=-270] at (0.246,4.683) {Execution Time [ms]};
\node[gp node center] at (6.541,0.215) {Number of Events [$\times 10000$]};
\node[gp node right] at (10.479,8.047) {naive};
\gpcolor{rgb color={0.580,0.000,0.827}}
\draw[gp path] (10.663,8.047)--(11.579,8.047);
\draw[gp path] (1.136,0.985)--(1.136,0.986)--(1.137,0.986)--(1.137,0.988)--(1.139,0.989)%
  --(1.142,0.995)--(1.145,0.999)--(1.148,1.004)--(1.154,1.005)--(1.159,1.015)--(1.163,1.014)%
  --(1.172,1.017)--(1.181,1.031)--(1.182,1.046)--(1.190,1.075)--(1.199,1.085)--(1.208,1.071)%
  --(1.217,1.100)--(1.226,1.138)--(1.228,1.102)--(1.280,1.196)--(1.424,1.363)--(1.713,1.548)%
  --(2.037,1.808)--(2.289,1.936)--(2.938,2.463)--(3.442,2.785)--(3.839,3.288)--(4.740,3.857)%
  --(5.641,4.439)--(5.749,4.559)--(6.542,5.587)--(7.442,6.031)--(8.343,6.581)--(9.244,7.153)%
  --(10.145,7.795)--(10.361,7.894);
\gpsetpointsize{4.00}
\gppoint{gp mark 1}{(1.136,0.985)}
\gppoint{gp mark 1}{(1.136,0.986)}
\gppoint{gp mark 1}{(1.137,0.986)}
\gppoint{gp mark 1}{(1.137,0.988)}
\gppoint{gp mark 1}{(1.139,0.989)}
\gppoint{gp mark 1}{(1.142,0.995)}
\gppoint{gp mark 1}{(1.145,0.999)}
\gppoint{gp mark 1}{(1.148,1.004)}
\gppoint{gp mark 1}{(1.154,1.005)}
\gppoint{gp mark 1}{(1.159,1.015)}
\gppoint{gp mark 1}{(1.163,1.014)}
\gppoint{gp mark 1}{(1.172,1.017)}
\gppoint{gp mark 1}{(1.181,1.031)}
\gppoint{gp mark 1}{(1.182,1.046)}
\gppoint{gp mark 1}{(1.190,1.075)}
\gppoint{gp mark 1}{(1.199,1.085)}
\gppoint{gp mark 1}{(1.208,1.071)}
\gppoint{gp mark 1}{(1.217,1.100)}
\gppoint{gp mark 1}{(1.226,1.138)}
\gppoint{gp mark 1}{(1.228,1.102)}
\gppoint{gp mark 1}{(1.280,1.196)}
\gppoint{gp mark 1}{(1.424,1.363)}
\gppoint{gp mark 1}{(1.713,1.548)}
\gppoint{gp mark 1}{(2.037,1.808)}
\gppoint{gp mark 1}{(2.289,1.936)}
\gppoint{gp mark 1}{(2.938,2.463)}
\gppoint{gp mark 1}{(3.442,2.785)}
\gppoint{gp mark 1}{(3.839,3.288)}
\gppoint{gp mark 1}{(4.740,3.857)}
\gppoint{gp mark 1}{(5.641,4.439)}
\gppoint{gp mark 1}{(5.749,4.559)}
\gppoint{gp mark 1}{(6.542,5.587)}
\gppoint{gp mark 1}{(7.442,6.031)}
\gppoint{gp mark 1}{(8.343,6.581)}
\gppoint{gp mark 1}{(9.244,7.153)}
\gppoint{gp mark 1}{(10.145,7.795)}
\gppoint{gp mark 1}{(10.361,7.894)}
\gppoint{gp mark 1}{(11.121,8.047)}
\gpcolor{color=gp lt color border}
\node[gp node right] at (10.479,7.739) {Boyer-Moore (without pre-processing)};
\gpcolor{rgb color={0.000,0.620,0.451}}
\draw[gp path] (10.663,7.739)--(11.579,7.739);
\draw[gp path] (1.136,0.986)--(1.137,0.986)--(1.137,0.989)--(1.139,0.989)--(1.142,0.992)%
  --(1.145,0.997)--(1.148,1.000)--(1.154,1.010)--(1.159,1.009)--(1.163,1.011)--(1.172,1.020)%
  --(1.181,1.042)--(1.182,1.032)--(1.190,1.050)--(1.199,1.051)--(1.208,1.043)--(1.217,1.083)%
  --(1.226,1.069)--(1.228,1.074)--(1.280,1.169)--(1.424,1.284)--(1.713,1.545)--(2.037,1.715)%
  --(2.289,1.832)--(2.938,2.305)--(3.442,2.574)--(3.839,3.048)--(4.740,3.538)--(5.641,4.063)%
  --(5.749,4.075)--(6.542,4.968)--(7.442,5.474)--(8.343,5.933)--(9.244,6.433)--(10.145,6.895)%
  --(10.361,7.061);
\gppoint{gp mark 2}{(1.136,0.986)}
\gppoint{gp mark 2}{(1.136,0.986)}
\gppoint{gp mark 2}{(1.137,0.986)}
\gppoint{gp mark 2}{(1.137,0.989)}
\gppoint{gp mark 2}{(1.139,0.989)}
\gppoint{gp mark 2}{(1.142,0.992)}
\gppoint{gp mark 2}{(1.145,0.997)}
\gppoint{gp mark 2}{(1.148,1.000)}
\gppoint{gp mark 2}{(1.154,1.010)}
\gppoint{gp mark 2}{(1.159,1.009)}
\gppoint{gp mark 2}{(1.163,1.011)}
\gppoint{gp mark 2}{(1.172,1.020)}
\gppoint{gp mark 2}{(1.181,1.042)}
\gppoint{gp mark 2}{(1.182,1.032)}
\gppoint{gp mark 2}{(1.190,1.050)}
\gppoint{gp mark 2}{(1.199,1.051)}
\gppoint{gp mark 2}{(1.208,1.043)}
\gppoint{gp mark 2}{(1.217,1.083)}
\gppoint{gp mark 2}{(1.226,1.069)}
\gppoint{gp mark 2}{(1.228,1.074)}
\gppoint{gp mark 2}{(1.280,1.169)}
\gppoint{gp mark 2}{(1.424,1.284)}
\gppoint{gp mark 2}{(1.713,1.545)}
\gppoint{gp mark 2}{(2.037,1.715)}
\gppoint{gp mark 2}{(2.289,1.832)}
\gppoint{gp mark 2}{(2.938,2.305)}
\gppoint{gp mark 2}{(3.442,2.574)}
\gppoint{gp mark 2}{(3.839,3.048)}
\gppoint{gp mark 2}{(4.740,3.538)}
\gppoint{gp mark 2}{(5.641,4.063)}
\gppoint{gp mark 2}{(5.749,4.075)}
\gppoint{gp mark 2}{(6.542,4.968)}
\gppoint{gp mark 2}{(7.442,5.474)}
\gppoint{gp mark 2}{(8.343,5.933)}
\gppoint{gp mark 2}{(9.244,6.433)}
\gppoint{gp mark 2}{(10.145,6.895)}
\gppoint{gp mark 2}{(10.361,7.061)}
\gppoint{gp mark 2}{(11.121,7.739)}
\gpcolor{color=gp lt color border}
\node[gp node right] at (10.479,7.431) {Boyer-Moore (with pre-processing)};
\gpcolor{rgb color={1.000,0.000,0.000}}
\draw[gp path] (10.663,7.431)--(11.579,7.431);
\draw[gp path] (1.136,0.992)--(1.136,0.989)--(1.137,0.990)--(1.137,0.995)--(1.139,0.993)%
  --(1.142,0.996)--(1.145,1.003)--(1.148,1.004)--(1.154,1.016)--(1.159,1.013)--(1.163,1.016)%
  --(1.172,1.024)--(1.181,1.047)--(1.182,1.037)--(1.190,1.056)--(1.199,1.056)--(1.208,1.047)%
  --(1.217,1.089)--(1.226,1.073)--(1.228,1.078)--(1.280,1.176)--(1.424,1.290)--(1.713,1.551)%
  --(2.037,1.720)--(2.289,1.837)--(2.938,2.311)--(3.442,2.579)--(3.839,3.053)--(4.740,3.544)%
  --(5.641,4.068)--(5.749,4.081)--(6.542,4.973)--(7.442,5.479)--(8.343,5.938)--(9.244,6.438)%
  --(10.145,6.901)--(10.361,7.067);
\gppoint{gp mark 3}{(1.136,0.992)}
\gppoint{gp mark 3}{(1.136,0.989)}
\gppoint{gp mark 3}{(1.137,0.990)}
\gppoint{gp mark 3}{(1.137,0.995)}
\gppoint{gp mark 3}{(1.139,0.993)}
\gppoint{gp mark 3}{(1.142,0.996)}
\gppoint{gp mark 3}{(1.145,1.003)}
\gppoint{gp mark 3}{(1.148,1.004)}
\gppoint{gp mark 3}{(1.154,1.016)}
\gppoint{gp mark 3}{(1.159,1.013)}
\gppoint{gp mark 3}{(1.163,1.016)}
\gppoint{gp mark 3}{(1.172,1.024)}
\gppoint{gp mark 3}{(1.181,1.047)}
\gppoint{gp mark 3}{(1.182,1.037)}
\gppoint{gp mark 3}{(1.190,1.056)}
\gppoint{gp mark 3}{(1.199,1.056)}
\gppoint{gp mark 3}{(1.208,1.047)}
\gppoint{gp mark 3}{(1.217,1.089)}
\gppoint{gp mark 3}{(1.226,1.073)}
\gppoint{gp mark 3}{(1.228,1.078)}
\gppoint{gp mark 3}{(1.280,1.176)}
\gppoint{gp mark 3}{(1.424,1.290)}
\gppoint{gp mark 3}{(1.713,1.551)}
\gppoint{gp mark 3}{(2.037,1.720)}
\gppoint{gp mark 3}{(2.289,1.837)}
\gppoint{gp mark 3}{(2.938,2.311)}
\gppoint{gp mark 3}{(3.442,2.579)}
\gppoint{gp mark 3}{(3.839,3.053)}
\gppoint{gp mark 3}{(4.740,3.544)}
\gppoint{gp mark 3}{(5.641,4.068)}
\gppoint{gp mark 3}{(5.749,4.081)}
\gppoint{gp mark 3}{(6.542,4.973)}
\gppoint{gp mark 3}{(7.442,5.479)}
\gppoint{gp mark 3}{(8.343,5.938)}
\gppoint{gp mark 3}{(9.244,6.438)}
\gppoint{gp mark 3}{(10.145,6.901)}
\gppoint{gp mark 3}{(10.361,7.067)}
\gppoint{gp mark 3}{(11.121,7.431)}
\gpcolor{color=gp lt color border}
\draw[gp path] (1.136,8.381)--(1.136,0.985)--(11.947,0.985)--(11.947,8.381)--cycle;
\gpdefrectangularnode{gp plot 1}{\pgfpoint{1.136cm}{0.985cm}}{\pgfpoint{11.947cm}{8.381cm}}
\end{tikzpicture}
   }

 \caption{Case 1: $\mathcal{A}$ and execution time}
 \label{fig:case1}
\scalebox{0.57}{
 \begin{tikzpicture}[shorten >=1pt,node distance=2.5cm,on grid,auto] 
     \node[state,initial] (s_0)   {$s_0$}; 
     \node[state,node distance=2cm] (s_1) [right=of s_0] {$s_1$}; 
     \node[state] (s_2) [right=of s_1] {$s_2$};
     \node[state] (s_3) [right=of s_2] {$s_3$};
     \node[state,node distance=2.0cm] (s_4) [below=of s_3] {$s_4$}; 
     \node[state,node distance=3cm] (s_5) [left=of s_4] {$s_5$}; 
     \node[state,node distance=3cm] (s_6) [left=of s_5] {$s_6$};
     \node[state,accepting] (s_7) [left=of s_6] {$s_7$};
     \path[->] 
     (s_0) edge  [above] node {\begin{tabular}{c}
                                $a,1 < x$\\
                                $/x := 0$
                               \end{tabular}} (s_1)
     (s_1) edge  [above] node {\begin{tabular}{c}
                                $a,0 < x < 1$\\
                                $/x := 0$
                               \end{tabular}} (s_2)
     (s_2) edge  [above] node {\begin{tabular}{c}
                                $a,0 < x < 1$\\
                                $/x := 0$
                               \end{tabular}} (s_3)
     (s_3) edge  [right] node {\begin{tabular}{c}
                                $a,0 < x < 1$\\
                                $/x := 0$
                               \end{tabular}} (s_4)
     (s_4) edge  [below] node {\begin{tabular}{c}
                                $a,0 < x < 1$\\
                                $/x := 0$
                               \end{tabular}} (s_5)
     (s_5) edge  [below] node {\begin{tabular}{c}
                                $a,1 < x$\\
                                $/x := 0$
                               \end{tabular}} (s_6)
     (s_5) edge  [above, loop above] node {$a,0 < x < 1/x := 0$} (s_5)
     (s_6) edge  [above, loop above] node {$a,1< x/x := 0$} (s_6)
     (s_6) edge  [above] node {$\$,x = 1$} (s_7);
 \end{tikzpicture}}

    \scalebox{0.50}{
   \input{4-0622123306.tikz}
   }

 \caption{Case 3: $\mathcal{A}$ and execution time}
 \label{fig:case3}
 \end{minipage}
 \begin{minipage}[b]{.5\textwidth}
  \centering
       \scalebox{0.7}{
       \begin{tikzpicture}[shorten >=1pt,node distance=2.3cm,on grid,auto] 
        \node[state,initial] (s_0)   {$s_0$}; 
        \node[state] (s_1) [right=of s_0] {$s_1$}; 
        \node[state,node distance=2.3cm] (s_2) [right=of s_1] {$s_2$};
        \node[state,accepting,node distance=2.0cm] (s_3) [right=of s_2] {$s_3$};
        \path[->] 
        (s_0) edge  [loop above] node {$a,\mathbf{true}/ y:= 0$} (s_0)
        (s_0) edge node  {$a, x=1$} (s_1)
        (s_1) edge  [loop above] node {$a,\mathbf{true}$} (s_1)
        (s_1) edge node  {$a,y=1$} (s_2)
        (s_2) edge node  {$\$,\mathbf{true}$} (s_3);
       \end{tikzpicture}
       }

    \scalebox{0.50}{
   \begin{tikzpicture}[gnuplot]
\path (0.000,0.000) rectangle (12.500,8.750);
\gpcolor{color=gp lt color border}
\gpsetlinetype{gp lt border}
\gpsetdashtype{gp dt solid}
\gpsetlinewidth{1.00}
\draw[gp path] (1.504,0.985)--(1.684,0.985);
\draw[gp path] (11.947,0.985)--(11.767,0.985);
\node[gp node right] at (1.320,0.985) {$0$};
\draw[gp path] (1.504,1.807)--(1.684,1.807);
\draw[gp path] (11.947,1.807)--(11.767,1.807);
\node[gp node right] at (1.320,1.807) {$500$};
\draw[gp path] (1.504,2.629)--(1.684,2.629);
\draw[gp path] (11.947,2.629)--(11.767,2.629);
\node[gp node right] at (1.320,2.629) {$1000$};
\draw[gp path] (1.504,3.450)--(1.684,3.450);
\draw[gp path] (11.947,3.450)--(11.767,3.450);
\node[gp node right] at (1.320,3.450) {$1500$};
\draw[gp path] (1.504,4.272)--(1.684,4.272);
\draw[gp path] (11.947,4.272)--(11.767,4.272);
\node[gp node right] at (1.320,4.272) {$2000$};
\draw[gp path] (1.504,5.094)--(1.684,5.094);
\draw[gp path] (11.947,5.094)--(11.767,5.094);
\node[gp node right] at (1.320,5.094) {$2500$};
\draw[gp path] (1.504,5.916)--(1.684,5.916);
\draw[gp path] (11.947,5.916)--(11.767,5.916);
\node[gp node right] at (1.320,5.916) {$3000$};
\draw[gp path] (1.504,6.737)--(1.684,6.737);
\draw[gp path] (11.947,6.737)--(11.767,6.737);
\node[gp node right] at (1.320,6.737) {$3500$};
\draw[gp path] (1.504,7.559)--(1.684,7.559);
\draw[gp path] (11.947,7.559)--(11.767,7.559);
\node[gp node right] at (1.320,7.559) {$4000$};
\draw[gp path] (1.504,8.381)--(1.684,8.381);
\draw[gp path] (11.947,8.381)--(11.767,8.381);
\node[gp node right] at (1.320,8.381) {$4500$};
\draw[gp path] (1.504,0.985)--(1.504,1.165);
\draw[gp path] (1.504,8.381)--(1.504,8.201);
\node[gp node center] at (1.504,0.677) {$0$};
\draw[gp path] (3.245,0.985)--(3.245,1.165);
\draw[gp path] (3.245,8.381)--(3.245,8.201);
\node[gp node center] at (3.245,0.677) {$20$};
\draw[gp path] (4.985,0.985)--(4.985,1.165);
\draw[gp path] (4.985,8.381)--(4.985,8.201);
\node[gp node center] at (4.985,0.677) {$40$};
\draw[gp path] (6.726,0.985)--(6.726,1.165);
\draw[gp path] (6.726,8.381)--(6.726,8.201);
\node[gp node center] at (6.726,0.677) {$60$};
\draw[gp path] (8.466,0.985)--(8.466,1.165);
\draw[gp path] (8.466,8.381)--(8.466,8.201);
\node[gp node center] at (8.466,0.677) {$80$};
\draw[gp path] (10.207,0.985)--(10.207,1.165);
\draw[gp path] (10.207,8.381)--(10.207,8.201);
\node[gp node center] at (10.207,0.677) {$100$};
\draw[gp path] (11.947,0.985)--(11.947,1.165);
\draw[gp path] (11.947,8.381)--(11.947,8.201);
\node[gp node center] at (11.947,0.677) {$120$};
\draw[gp path] (1.504,8.381)--(1.504,0.985)--(11.947,0.985)--(11.947,8.381)--cycle;
\node[gp node center,rotate=-270] at (0.246,4.683) {Execution Time [ms]};
\node[gp node center] at (6.725,0.215) {Number of Events [$\times 100$]};
\node[gp node right] at (10.479,8.047) {naive};
\gpcolor{rgb color={0.580,0.000,0.827}}
\draw[gp path] (10.663,8.047)--(11.579,8.047);
\draw[gp path] (1.521,0.985)--(1.539,0.985)--(1.574,0.986)--(1.643,0.987)--(1.782,0.992)%
  --(2.061,1.012)--(2.374,1.096)--(2.618,1.130)--(3.245,1.266)--(3.732,1.394)--(4.115,1.554)%
  --(4.985,1.926)--(5.855,2.379)--(5.960,2.481)--(6.726,3.041)--(7.596,3.718)--(8.466,4.532)%
  --(9.336,5.453)--(10.207,6.492)--(10.415,6.809);
\gpsetpointsize{4.00}
\gppoint{gp mark 1}{(1.521,0.985)}
\gppoint{gp mark 1}{(1.539,0.985)}
\gppoint{gp mark 1}{(1.574,0.986)}
\gppoint{gp mark 1}{(1.643,0.987)}
\gppoint{gp mark 1}{(1.782,0.992)}
\gppoint{gp mark 1}{(2.061,1.012)}
\gppoint{gp mark 1}{(2.374,1.096)}
\gppoint{gp mark 1}{(2.618,1.130)}
\gppoint{gp mark 1}{(3.245,1.266)}
\gppoint{gp mark 1}{(3.732,1.394)}
\gppoint{gp mark 1}{(4.115,1.554)}
\gppoint{gp mark 1}{(4.985,1.926)}
\gppoint{gp mark 1}{(5.855,2.379)}
\gppoint{gp mark 1}{(5.960,2.481)}
\gppoint{gp mark 1}{(6.726,3.041)}
\gppoint{gp mark 1}{(7.596,3.718)}
\gppoint{gp mark 1}{(8.466,4.532)}
\gppoint{gp mark 1}{(9.336,5.453)}
\gppoint{gp mark 1}{(10.207,6.492)}
\gppoint{gp mark 1}{(10.415,6.809)}
\gppoint{gp mark 1}{(11.121,8.047)}
\gpcolor{color=gp lt color border}
\node[gp node right] at (10.479,7.739) {Boyer-Moore (without pre-processing)};
\gpcolor{rgb color={0.000,0.620,0.451}}
\draw[gp path] (10.663,7.739)--(11.579,7.739);
\draw[gp path] (1.521,0.985)--(1.539,0.986)--(1.574,0.986)--(1.643,0.987)--(1.782,0.993)%
  --(2.061,1.012)--(2.374,1.048)--(2.618,1.087)--(3.245,1.229)--(3.732,1.394)--(4.115,1.538)%
  --(4.985,1.958)--(5.855,2.507)--(5.960,2.639)--(6.726,3.224)--(7.596,4.013)--(8.466,4.944)%
  --(9.336,6.009)--(10.207,7.124)--(10.415,7.470);
\gppoint{gp mark 2}{(1.521,0.985)}
\gppoint{gp mark 2}{(1.539,0.986)}
\gppoint{gp mark 2}{(1.574,0.986)}
\gppoint{gp mark 2}{(1.643,0.987)}
\gppoint{gp mark 2}{(1.782,0.993)}
\gppoint{gp mark 2}{(2.061,1.012)}
\gppoint{gp mark 2}{(2.374,1.048)}
\gppoint{gp mark 2}{(2.618,1.087)}
\gppoint{gp mark 2}{(3.245,1.229)}
\gppoint{gp mark 2}{(3.732,1.394)}
\gppoint{gp mark 2}{(4.115,1.538)}
\gppoint{gp mark 2}{(4.985,1.958)}
\gppoint{gp mark 2}{(5.855,2.507)}
\gppoint{gp mark 2}{(5.960,2.639)}
\gppoint{gp mark 2}{(6.726,3.224)}
\gppoint{gp mark 2}{(7.596,4.013)}
\gppoint{gp mark 2}{(8.466,4.944)}
\gppoint{gp mark 2}{(9.336,6.009)}
\gppoint{gp mark 2}{(10.207,7.124)}
\gppoint{gp mark 2}{(10.415,7.470)}
\gppoint{gp mark 2}{(11.121,7.739)}
\gpcolor{color=gp lt color border}
\node[gp node right] at (10.479,7.431) {Boyer-Moore (with pre-processing)};
\gpcolor{rgb color={1.000,0.000,0.000}}
\draw[gp path] (10.663,7.431)--(11.579,7.431);
\draw[gp path] (1.521,1.247)--(1.539,1.210)--(1.574,1.207)--(1.643,1.201)--(1.782,1.211)%
  --(2.061,1.211)--(2.374,1.294)--(2.618,1.321)--(3.245,1.474)--(3.732,1.621)--(4.115,1.746)%
  --(4.985,2.211)--(5.855,2.749)--(5.960,2.848)--(6.726,3.476)--(7.596,4.228)--(8.466,5.170)%
  --(9.336,6.224)--(10.207,7.366)--(10.415,7.692);
\gppoint{gp mark 3}{(1.521,1.247)}
\gppoint{gp mark 3}{(1.539,1.210)}
\gppoint{gp mark 3}{(1.574,1.207)}
\gppoint{gp mark 3}{(1.643,1.201)}
\gppoint{gp mark 3}{(1.782,1.211)}
\gppoint{gp mark 3}{(2.061,1.211)}
\gppoint{gp mark 3}{(2.374,1.294)}
\gppoint{gp mark 3}{(2.618,1.321)}
\gppoint{gp mark 3}{(3.245,1.474)}
\gppoint{gp mark 3}{(3.732,1.621)}
\gppoint{gp mark 3}{(4.115,1.746)}
\gppoint{gp mark 3}{(4.985,2.211)}
\gppoint{gp mark 3}{(5.855,2.749)}
\gppoint{gp mark 3}{(5.960,2.848)}
\gppoint{gp mark 3}{(6.726,3.476)}
\gppoint{gp mark 3}{(7.596,4.228)}
\gppoint{gp mark 3}{(8.466,5.170)}
\gppoint{gp mark 3}{(9.336,6.224)}
\gppoint{gp mark 3}{(10.207,7.366)}
\gppoint{gp mark 3}{(10.415,7.692)}
\gppoint{gp mark 3}{(11.121,7.431)}
\gpcolor{color=gp lt color border}
\draw[gp path] (1.504,8.381)--(1.504,0.985)--(11.947,0.985)--(11.947,8.381)--cycle;
\gpdefrectangularnode{gp plot 1}{\pgfpoint{1.504cm}{0.985cm}}{\pgfpoint{11.947cm}{8.381cm}}
\end{tikzpicture}
   }

 \caption{Case 2: $\mathcal{A}$ and execution time}
 \label{fig:case2}
 \scalebox{0.57}{
 \begin{tikzpicture}[shorten >=1pt,node distance=2.5cm,on grid,auto] 
    \node[state,initial] (s_0)  {$s_0$}; 
    \node[state,node distance=2cm] (s_1) [right=of s_0] {$s_1$}; 
    \node[state] (s_2) [right=of s_1] {$s_2$};
    \node[state] (s_3) [right=of s_2] {$s_3$};
    \node[state,node distance=2cm] (s_4) [below=of s_3] {$s_4$}; 
    \node[state,node distance=2.5cm] (s_5) [left=of s_4] {$s_5$}; 
    \node[state,node distance=2cm] (s_6) [left=of s_5] {$s_6$};
    \node[state,accepting,node distance=2cm] (s_7) [left=of s_6] {$s_7$};
    \path[->] 
    (s_0) edge  [above] node {\begin{tabular}{c}
                               $\mathit{low},\mathbf{true}$\\
                               $/x := 0$
                              \end{tabular}} (s_1)
    (s_1) edge  [above] node {\begin{tabular}{c}
                               $\mathit{high},$\\
                               $0 < x < 1$
                              \end{tabular}} (s_2)
    (s_2) edge  [above] node {\begin{tabular}{c}
                               $\mathit{high},$\\
                               $0 < x < 1$
                              \end{tabular}} (s_3)
    (s_3) edge  [right] node {\begin{tabular}{c}
                               $\mathit{high},$\\
                               $0 < x < 1$
                              \end{tabular}} (s_4)
    (s_4) edge  [above] node {\begin{tabular}{c}
                               $\mathit{high},$\\
                               $0 < x < 1$
                              \end{tabular}} (s_5)
    (s_5) edge  [above] node {\begin{tabular}{c}
                               $\mathit{high},$\\
                               $1 < x$
                              \end{tabular}} (s_6)
    (s_5) edge  [above, loop above] node {$\mathit{high},\mathbf{true}$} (s_5)
    (s_6) edge  [above] node {$\$,\mathbf{true}$} (s_7);
 \end{tikzpicture}}

    \scalebox{0.50}{
   \begin{tikzpicture}[gnuplot]
\path (0.000,0.000) rectangle (12.500,8.750);
\gpcolor{color=gp lt color border}
\gpsetlinetype{gp lt border}
\gpsetdashtype{gp dt solid}
\gpsetlinewidth{1.00}
\draw[gp path] (1.320,0.985)--(1.500,0.985);
\draw[gp path] (11.947,0.985)--(11.767,0.985);
\node[gp node right] at (1.136,0.985) {$0$};
\draw[gp path] (1.320,2.218)--(1.500,2.218);
\draw[gp path] (11.947,2.218)--(11.767,2.218);
\node[gp node right] at (1.136,2.218) {$50$};
\draw[gp path] (1.320,3.450)--(1.500,3.450);
\draw[gp path] (11.947,3.450)--(11.767,3.450);
\node[gp node right] at (1.136,3.450) {$100$};
\draw[gp path] (1.320,4.683)--(1.500,4.683);
\draw[gp path] (11.947,4.683)--(11.767,4.683);
\node[gp node right] at (1.136,4.683) {$150$};
\draw[gp path] (1.320,5.916)--(1.500,5.916);
\draw[gp path] (11.947,5.916)--(11.767,5.916);
\node[gp node right] at (1.136,5.916) {$200$};
\draw[gp path] (1.320,7.148)--(1.500,7.148);
\draw[gp path] (11.947,7.148)--(11.767,7.148);
\node[gp node right] at (1.136,7.148) {$250$};
\draw[gp path] (1.320,8.381)--(1.500,8.381);
\draw[gp path] (11.947,8.381)--(11.767,8.381);
\node[gp node right] at (1.136,8.381) {$300$};
\draw[gp path] (1.320,0.985)--(1.320,1.165);
\draw[gp path] (1.320,8.381)--(1.320,8.201);
\node[gp node center] at (1.320,0.677) {$0$};
\draw[gp path] (2.383,0.985)--(2.383,1.165);
\draw[gp path] (2.383,8.381)--(2.383,8.201);
\node[gp node center] at (2.383,0.677) {$50$};
\draw[gp path] (3.445,0.985)--(3.445,1.165);
\draw[gp path] (3.445,8.381)--(3.445,8.201);
\node[gp node center] at (3.445,0.677) {$100$};
\draw[gp path] (4.508,0.985)--(4.508,1.165);
\draw[gp path] (4.508,8.381)--(4.508,8.201);
\node[gp node center] at (4.508,0.677) {$150$};
\draw[gp path] (5.571,0.985)--(5.571,1.165);
\draw[gp path] (5.571,8.381)--(5.571,8.201);
\node[gp node center] at (5.571,0.677) {$200$};
\draw[gp path] (6.634,0.985)--(6.634,1.165);
\draw[gp path] (6.634,8.381)--(6.634,8.201);
\node[gp node center] at (6.634,0.677) {$250$};
\draw[gp path] (7.696,0.985)--(7.696,1.165);
\draw[gp path] (7.696,8.381)--(7.696,8.201);
\node[gp node center] at (7.696,0.677) {$300$};
\draw[gp path] (8.759,0.985)--(8.759,1.165);
\draw[gp path] (8.759,8.381)--(8.759,8.201);
\node[gp node center] at (8.759,0.677) {$350$};
\draw[gp path] (9.822,0.985)--(9.822,1.165);
\draw[gp path] (9.822,8.381)--(9.822,8.201);
\node[gp node center] at (9.822,0.677) {$400$};
\draw[gp path] (10.884,0.985)--(10.884,1.165);
\draw[gp path] (10.884,8.381)--(10.884,8.201);
\node[gp node center] at (10.884,0.677) {$450$};
\draw[gp path] (11.947,0.985)--(11.947,1.165);
\draw[gp path] (11.947,8.381)--(11.947,8.201);
\node[gp node center] at (11.947,0.677) {$500$};
\draw[gp path] (1.320,8.381)--(1.320,0.985)--(11.947,0.985)--(11.947,8.381)--cycle;
\node[gp node center,rotate=-270] at (0.246,4.683) {Execution Time [ms]};
\node[gp node center] at (6.633,0.215) {Number of Events [$\times 10000$]};
\node[gp node right] at (10.479,8.047) {naive};
\gpcolor{rgb color={0.580,0.000,0.827}}
\draw[gp path] (10.663,8.047)--(11.579,8.047);
\draw[gp path] (1.836,1.349)--(2.355,1.720)--(2.876,2.086)--(3.915,2.828)--(5.210,3.812)%
  --(6.500,4.730)--(7.794,5.613)--(9.088,6.572)--(10.385,7.394)--(11.678,8.360);
\gpsetpointsize{4.00}
\gppoint{gp mark 1}{(1.836,1.349)}
\gppoint{gp mark 1}{(2.355,1.720)}
\gppoint{gp mark 1}{(2.876,2.086)}
\gppoint{gp mark 1}{(3.915,2.828)}
\gppoint{gp mark 1}{(5.210,3.812)}
\gppoint{gp mark 1}{(6.500,4.730)}
\gppoint{gp mark 1}{(7.794,5.613)}
\gppoint{gp mark 1}{(9.088,6.572)}
\gppoint{gp mark 1}{(10.385,7.394)}
\gppoint{gp mark 1}{(11.678,8.360)}
\gppoint{gp mark 1}{(11.121,8.047)}
\gpcolor{color=gp lt color border}
\node[gp node right] at (10.479,7.739) {Boyer-Moore (without pre-processing)};
\gpcolor{rgb color={0.000,0.620,0.451}}
\draw[gp path] (10.663,7.739)--(11.579,7.739);
\draw[gp path] (1.836,1.182)--(2.355,1.375)--(2.876,1.576)--(3.915,1.950)--(5.210,2.480)%
  --(6.500,2.906)--(7.794,3.392)--(9.088,3.867)--(10.385,4.377)--(11.678,4.892);
\gppoint{gp mark 2}{(1.836,1.182)}
\gppoint{gp mark 2}{(2.355,1.375)}
\gppoint{gp mark 2}{(2.876,1.576)}
\gppoint{gp mark 2}{(3.915,1.950)}
\gppoint{gp mark 2}{(5.210,2.480)}
\gppoint{gp mark 2}{(6.500,2.906)}
\gppoint{gp mark 2}{(7.794,3.392)}
\gppoint{gp mark 2}{(9.088,3.867)}
\gppoint{gp mark 2}{(10.385,4.377)}
\gppoint{gp mark 2}{(11.678,4.892)}
\gppoint{gp mark 2}{(11.121,7.739)}
\gpcolor{color=gp lt color border}
\node[gp node right] at (10.479,7.431) {Boyer-Moore (with pre-processing)};
\gpcolor{rgb color={1.000,0.000,0.000}}
\draw[gp path] (10.663,7.431)--(11.579,7.431);
\draw[gp path] (1.836,1.430)--(2.355,1.614)--(2.876,1.816)--(3.915,2.188)--(5.210,2.729)%
  --(6.500,3.147)--(7.794,3.647)--(9.088,4.110)--(10.385,4.625)--(11.678,5.140);
\gppoint{gp mark 3}{(1.836,1.430)}
\gppoint{gp mark 3}{(2.355,1.614)}
\gppoint{gp mark 3}{(2.876,1.816)}
\gppoint{gp mark 3}{(3.915,2.188)}
\gppoint{gp mark 3}{(5.210,2.729)}
\gppoint{gp mark 3}{(6.500,3.147)}
\gppoint{gp mark 3}{(7.794,3.647)}
\gppoint{gp mark 3}{(9.088,4.110)}
\gppoint{gp mark 3}{(10.385,4.625)}
\gppoint{gp mark 3}{(11.678,5.140)}
\gppoint{gp mark 3}{(11.121,7.431)}
\gpcolor{color=gp lt color border}
\draw[gp path] (1.320,8.381)--(1.320,0.985)--(11.947,0.985)--(11.947,8.381)--cycle;
\gpdefrectangularnode{gp plot 1}{\pgfpoint{1.320cm}{0.985cm}}{\pgfpoint{11.947cm}{8.381cm}}
\end{tikzpicture}
   }

 \caption{Case 5: $\mathcal{A}$ and execution time}
 \label{fig:case5}
 \end{minipage}
\end{figure}


\noindent
\textbf{Case 1: No Clock Constraints}\quad
In Fig.~\ref{fig:case1} we present a timed automaton $\mathcal{A}$ and
the execution time (excluding pre-processing) for $37$ timed words $w$ whose lengths range from
20 to 1,024,000. Each timed word $w$ is an alternation of $a,b\in
\Sigma$, and its time stamps are randomly generated according to a certain 
uniform distribution. 

The automaton $\mathcal{A}$ is without any clock constraints, so in this
case the problem is almost that of \emph{untimed} pattern matching.  The
Boyer-Moore algorithm outperforms the naive one; but the gap is
approximately 1/10 when $|w|$ is large enough, which is smaller than
what one would expect from the fact that $i$ is always decremented by 2.
This is because, as some combinatorial investigation would reveal, those
$i$'s which are skipped are for which we examine fewer $j$'s


The pre-processing phase (that relies only on $\mathcal{A}$) took 
$6.38 \cdot 10^{-2}$ms.\ on average.



\noindent
\textbf{Case 2: Beyond Expressivity of TREs}\quad
In Fig.~\ref{fig:case2} are a timed automaton $\mathcal{A}$---one
that is not expressible with TREs~\cite{Herrmann1999}---and
the execution time for $20$ timed words $w$
whose lengths range from 20 to 10,240. Each $w$ is a
repetition of $a \in \Sigma$, and its time stamps are randomly generated
according to the uniform distribution in the interval $(0,0.1)$.

One can easily see that the skip value is always $1$, so our Boyer-Moore
algorithm is slightly slower due to the overhead of repeatedly reading
the result of pre-processing.
 The naive algorithm (and hence the Boyer-Moore one too) exhibits non-linear increase in
Fig.~\ref{fig:case2}; this is because its worst-case complexity is
bounded by $|w||E|^{{|w|}+1}$ (where $|E|$ is the number of edges in
$\mathcal{A}$). See the proof of Thm.~\ref{thm:terminationAndCorrectnessOfNaiveAlg}
(Appendix~\ref{subsec:proofOfThmTerminationAndCorrectnessOfNaiveAlg}). 
The factor $|E|$ in the above complexity bound stems essentially from nondeterminism.


The pre-processing phase
took 
$1.39 \cdot 10^{2}$ms.\ on average.

\noindent
\textbf{Case 3: Accepting Runs Are Long}\quad
In Fig.~\ref{fig:case3} are a timed automaton $\mathcal{A}$ and
the execution time for $49$ timed words $w$
whose lengths range from 8,028 to 10,243,600.
Each  $w$ is randomly generated as follows: it is a repetition of 
 $a\in \Sigma$; $a$ is repeated according to the exponential
 distribution with a parameter $\lambda$; and we do so for a fixed duration
$\tau_{|\overline{\tau}|}$, generating a timed word of length $|w|$. See 
Table~\ref{131638_23Jan16} in
Appendix~\ref{appendix:detailedResultsOfOurExperiments}.


In the automaton  $\mathcal{A}$ the length $m$ of the shortest accepting
run is large; hence so are the skip values in the Boyer-Moore optimization.
(Specifically the skip value is $5$ if both $\tau_i - \tau_{i-1}$ and
$\tau_{i+1} - \tau_{i}$ are greater than $1$.) Indeed, as we see from
the figure, 
the Boyer-Moore algorithm outperforms the naive one roughly by twice.


The pre-processing phase
took 
$7.02$ms.\ on average. This is in practice negligible; recall that pre-processing is done only once when
$\mathcal{A}$ is given. 
    

\noindent
\textbf{Case 4: Region Automata are Big}\quad
Here   $\mathcal{A}$ is a translation of the TRE\\
 $\bigl\langle\bigl(\,\bigl(\langle\mathrm{p}\rangle_{(0,10]}\langle\mathrm{\neg
    p}\rangle_{(0,10]}\bigr)^*\land\bigl(\langle\mathrm{q}\rangle_{(0,10]}\langle\mathrm{\neg
    q}\rangle_{(0,10]}\bigr)^*\,\bigr)\$\bigr\rangle_{(0,80]}$. 
We executed our two algorithms for $12$ timed words $w$ whose lengths range from
1,934 to 31,935.
Each  $w$ is generated randomly as follows: it is the interleaving
combination of an alternation of $p,\lnot p$ and one of $q,\lnot q$; in
each alternation the time stamps are governed by the exponential
distribution with a parameter $\lambda$; and its duration $\tau_{|\tau|}$ is fixed.  See Table~\ref{162108_21Jan16}
in Appendix~\ref{appendix:detailedResultsOfOurExperiments}.


This $\mathcal{A}$ is bad for our Boyer-Moore type algorithm since
its region automaton $R(\mathcal{A})$ is very big. Specifically: the
numbers $c_{x}$ in Def.~\ref{def:regionAutom} are big (10 and 80) and we
have to have many states accordingly in $R(\mathcal{A})$---recall that
in  $\sim$ we care about the coincidence of 
integer part. 
Indeed, the construction of $\noredundancy{R} (\mathcal{A})$
took ca.\ 74s., and the construction of
$R(\mathcal{A} \times \mathcal{A})$ did not complete due to
RAM shortage. Therefore we couldn't complete pre-processing for Boyer-Moore.

 We note however that our naive algorithm 
 worked fine.
 See Table~\ref{162108_21Jan16}
in Appendix~\ref{appendix:detailedResultsOfOurExperiments}. 

\noindent
\textbf{Case 5: An Automotive Example}\quad
 This final example (Fig.~\ref{fig:case5}) is about anomaly detection of 
engines. The execution time is shown for  $10$ timed words $w$
whose lengths range from 242,808 to 4,873,207.
Each  $w$ is obtained as a discretized log of the simulation of the model
\texttt{sldemo\_enginewc.slx} in the Simulink Demo
palette~\cite{SimulinkGuide}: here the input of the model (desired
rpm) is generated randomly  according to the Gaussian distribution 
with $\mu=\text{2,000}$rpm and $\sigma^{2}=10^{6}$$\text{rpm}^{2}$;
we discretized the output of the model (engine  torque) into two statuses,
$\mathit{high}$ and $\mathit{low}$, with
the threshold of $40\mathrm{N \cdot m}$.

This test case is meant to be a practical example in automotive
applications---our original motivation for the current work. The
automaton $\mathcal{A}$ expresses: the engine torque is $\mathit{high}$
for more than $1$s. (the kind of anomaly we are interested in) and 
the log is not too sparse (which means the log is a credible one). 


Here the Boyer-Moore algorithm outperforms the naive one roughly by
twice. 
The pre-processing phase took 
$9.94$ms.\ on average.

\marginpar{Ichiro added this paragraph}
Lacking in the current section are: detailed comparison with the existing
implementations  (e.g.\ in~\cite{Ulus2014}, modulo  the word-signal
difference in Rem.~\ref{rem:signal}); and performance analysis when
the specification $\mathcal{A}$, instead of the input timed word $w$,
grows. We intend to address these issues in the coming extended version.

\auxproof{ \section{Conclusions}
 \label{sec:conclusions}

 The main contribution in this paper is the construction of our
 Boyer-Moore type algorithm for timed pattern matching. In this
 algorithm, we avoid some unnecessary matching using the information of
 the already read string. This algorithm is about twice faster than our
 naive algorithm for some problem instances.

 However, the space complexity is too large for some instances in our
 implementation of our Boyer-Moore type algorithm. It is because the
 number of states of region automata increases exponentially. In order to
 reduce the space complexity, we can construct the necessary part of the
 region automaton for matching on-the-fly instead of having the whole
 part of the region automaton. This will be similar to the trick to
 reduce the complexity of the emptiness problem of timed B\"{u}chi automaton
 to PSPACE in~\cite{Alur1994}.
}

\paragraph*{Acknowledgments}
Thanks are due to the anonymous referees for their careful reading and
expert comments.
The authors are supported by
Grant-in-Aid No.\ 15KT0012, JSPS; T.A.\ is supported by
Grant-in-Aid for JSPS Fellows.

\bibliographystyle{ichiro}
\bibliography{dblp_refs}

\begin{thebibliography}{10}

\bibitem{Alur1994}
R.~Alur and D.L. Dill.
\newblock A theory of timed automata.
\newblock \emph{Theor. Comput. Sci.}, 126(2):183--235, 1994.

\bibitem{EugeneAsarin}
E.~Asarin, P.~Caspi and O.~Maler.
\newblock A {Kleene} theorem for timed automata.
\newblock In \emph{Proceedings, 12th Annual {IEEE} Symposium on Logic in
  Computer Science, Warsaw, Poland, June 29 - July 2, 1997}, pp. 160--171.
  {IEEE} Computer Society, 1997.

\bibitem{Asarin2002}
E.~Asarin, P.~Caspi and O.~Maler.
\newblock Timed regular expressions.
\newblock \emph{J. {ACM}}, 49(2):172--206, 2002.

\bibitem{DBLP:conf/rv/2015}
E.~Bartocci and R.~Majumdar, editors.
\newblock \emph{Runtime Verification - 6th International Conference, {RV} 2015
  Vienna, Austria, September 22-25, 2015. Proceedings}, vol. 9333 of
  \emph{Lecture Notes in Computer Science}. Springer, 2015.

\bibitem{BehrmannBLP06}
G.~Behrmann, P.~Bouyer, K.G. Larsen and R.~Pel{\'{a}}nek.
\newblock Lower and upper bounds in zone-based abstractions of timed automata.
\newblock \emph{{STTT}}, 8(3):204--215, 2006.

\bibitem{DBLP:conf/rv/2014}
B.~Bonakdarpour and S.A. Smolka, editors.
\newblock \emph{Runtime Verification - 5th International Conference, {RV} 2014,
  Toronto, ON, Canada, September 22-25, 2014. Proceedings}, vol. 8734 of
  \emph{Lecture Notes in Computer Science}. Springer, 2014.

\bibitem{Boyer1977}
R.S. Boyer and J.S. Moore.
\newblock A fast string searching algorithm.
\newblock \emph{Commun. {ACM}}, 20(10):762--772, 1977.

\bibitem{DBLP:conf/rv/ColomboP12}
C.~Colombo and G.J. Pace.
\newblock Fast-forward runtime monitoring - an industrial case study.
\newblock In S.~Qadeer and S.~Tasiran, editors, \emph{Runtime Verification,
  Third International Conference, {RV} 2012, Istanbul, Turkey, September 25-28,
  2012, Revised Selected Papers}, vol. 7687 of \emph{Lecture Notes in Computer
  Science}, pp. 214--228. Springer, 2012.

\bibitem{DBLP:conf/rv/DeshmukhDGJJS15}
J.V. Deshmukh, A.~Donz{\'{e}}, S.~Ghosh, X.~Jin, G.~Juniwal and S.A. Seshia.
\newblock Robust online monitoring of signal temporal logic.
\newblock In Bartocci and Majumdar  \cite{DBLP:conf/rv/2015}, pp. 55--70.

\bibitem{DBLP:conf/rv/DokhanchiHF14}
A.~Dokhanchi, B.~Hoxha and G.E. Fainekos.
\newblock On-line monitoring for temporal logic robustness.
\newblock In Bonakdarpour and Smolka  \cite{DBLP:conf/rv/2014}, pp. 231--246.

\bibitem{DBLP:conf/cav/DonzeFM13}
A.~Donz{\'{e}}, T.~Ferr{\`{e}}re and O.~Maler.
\newblock Efficient robust monitoring for {STL}.
\newblock In N.~Sharygina and H.~Veith, editors, \emph{Computer Aided
  Verification - 25th International Conference, {CAV} 2013, Saint Petersburg,
  Russia, July 13-19, 2013. Proceedings}, vol. 8044 of \emph{Lecture Notes in
  Computer Science}, pp. 264--279. Springer, 2013.

\bibitem{Ferrere2015}
T.~Ferr{\`{e}}re, O.~Maler, D.~Nickovic and D.~Ulus.
\newblock Measuring with timed patterns.
\newblock In D.~Kroening and C.S. Pasareanu, editors, \emph{Computer Aided
  Verification - 27th International Conference, {CAV} 2015, San Francisco, CA,
  USA, July 18-24, 2015, Proceedings, Part {II}}, vol. 9207 of \emph{Lecture
  Notes in Computer Science}, pp. 322--337. Springer, 2015.

\bibitem{DBLP:conf/rv/GeistRS14}
J.~Geist, K.Y. Rozier and J.~Schumann.
\newblock Runtime observer pairs and bayesian network reasoners on-board fpgas:
  Flight-certifiable system health management for embedded systems.
\newblock In Bonakdarpour and Smolka  \cite{DBLP:conf/rv/2014}, pp. 215--230.

\bibitem{Herrmann1999}
P.~Herrmann.
\newblock Renaming is necessary in timed regular expressions.
\newblock In C.P. Rangan, V.~Raman and R.~Ramanujam, editors, \emph{Foundations
  of Software Technology and Theoretical Computer Science, 19th Conference,
  Chennai, India, December 13-15, 1999, Proceedings}, vol. 1738 of
  \emph{Lecture Notes in Computer Science}, pp. 47--59. Springer, 1999.

\bibitem{DBLP:conf/rv/HoOW14}
H.~Ho, J.~Ouaknine and J.~Worrell.
\newblock Online monitoring of metric temporal logic.
\newblock In Bonakdarpour and Smolka  \cite{DBLP:conf/rv/2014}, pp. 178--192.

\bibitem{DBLP:conf/rv/KaneCDK15}
A.~Kane, O.~Chowdhury, A.~Datta and P.~Koopman.
\newblock A case study on runtime monitoring of an autonomous research vehicle
  {(ARV)} system.
\newblock In Bartocci and Majumdar  \cite{DBLP:conf/rv/2015}, pp. 102--117.

\bibitem{Knuth1977}
D.E. Knuth, J.H.M. Jr. and V.R. Pratt.
\newblock Fast pattern matching in strings.
\newblock \emph{{SIAM} J. Comput.}, 6(2):323--350, 1977.

\bibitem{SimulinkGuide}
The MathWorks, Inc., Natick, MA, USA.
\newblock \emph{{Simulink User's Guide}}, 2015.

\bibitem{MooreExample}
{Boyer-Moore Fast String Searching Example}.
\newblock
  \url{http://www.cs.utexas.edu/users/moore/best-ideas/string-searching/fstrpos-example.html}.

\bibitem{DBLP:journals/corr/abs-cs-0702120}
J.~Ouaknine and J.~Worrell.
\newblock On the decidability and complexity of metric temporal logic over
  finite words.
\newblock \emph{Logical Methods in Computer Science}, 3(1), 2007.

\bibitem{PandyaS12}
P.K. Pandya and P.V. Suman.
\newblock An introduction to timed automata.
\newblock In \emph{Modern Applications of Automata Theory}, pp. 111--148. World
  Scientific, 2012.

\bibitem{Ulus2014}
D.~Ulus, T.~Ferr{\`{e}}re, E.~Asarin and O.~Maler.
\newblock Timed pattern matching.
\newblock In A.~Legay and M.~Bozga, editors, \emph{Formal Modeling and Analysis
  of Timed Systems - 12th International Conference, {FORMATS} 2014, Florence,
  Italy, September 8-10, 2014. Proceedings}, vol. 8711 of \emph{Lecture Notes
  in Computer Science}, pp. 222--236. Springer, 2014.

\bibitem{Ulus2016}
D.~Ulus, T.~Ferr{\`{e}}re, E.~Asarin and O.~Maler.
\newblock Online timed pattern matching using derivatives.
\newblock In M.~Chechik and J.~Raskin, editors, \emph{Tools and Algorithms for
  the Construction and Analysis of Systems - 22nd International Conference,
  {TACAS} 2016, Held as Part of the European Joint Conferences on Theory and
  Practice of Software, {ETAPS} 2016, Eindhoven, The Netherlands, April 2-8,
  2016, Proceedings}, vol. 9636 of \emph{Lecture Notes in Computer Science},
  pp. 736--751. Springer, 2016.

\bibitem{TimedBoyerMooreSample}
M.~Waga, T.~Akazaki and I.~Hasuo.
\newblock {Code that Accompanies "A Boyer-Moore Type Algorithm for Timed
  Pattern Matching."}.
\newblock \url{https://github.com/MasWag/timed-pattern-matching}.

\bibitem{Watson2003}
B.W. Watson and R.E. Watson.
\newblock A boyer-moore-style algorithm for regular expression pattern
  matching.
\newblock \emph{Sci. Comput. Program.}, 48(2-3):99--117, 2003.

\end{thebibliography}

\newpage
\appendix
\section{Skip Value Functions in the (Original) Boyer-Moore Algorithm}
\label{appendix:skipValueFunctions}
Here $a\in\Sigma$ is  a character, and $p\in\mathbb{Z}_{>0}$ is a positive integer.
\begin{align*} 
 \Delta_1 (
a,p) &=
  \min
  \{p - k\mid \text{$\pat (k) = 
  a$ or $k=0$}\}
 \\
 d_1 (p) &=
 \min\{
 n\in\mathbb{Z}_{>0}
  \mid 
  \pat
  (p+1-n,|\pat|-n) = \pat (p+1,|\pat|)
  \text{ or }
 n
  \geq |\pat|
 \}
 \\ 
 d_2 &= 
\min\bigl\{\, n\in\mathbb{Z}_{>0}\mid
 \pat
 (1,|\pat|-n) = \pat (n+1,|\pat|\,)
\,\bigr\}
\\ 
  \Delta_2 (p) &= \min\{d_1 (p), d_2\}
\end{align*}

\section{Skip Value Function in the Boyer-Moore Type Algorithm for
 Pattern Matching}
\label{appendix:skipValueFunctionForWatsonsAlgorithm}


The  definitions below follow those in~\cite{Watson2003}. 

Here $s$ is a state of the automaton 
$\mathcal{A}=(\Sigma,S,S_{0},E,F)$ that accepts the reversed language
$L^{\Rev}$ of the given pattern $L$; we let
$\mathcal{A}_{s}=(\Sigma,S,S_{0},E,\{s\})$ be the automaton where $s$ is
the only accepting state.
\begin{align*}
 m_{s}
 &= \min\{|w|\mid w\in L(\mathcal{A}_{s})\}
 \qquad\text{(the length of a shortest word that leads to $s$)}
 \\
 m
 &=
 \min_{s\in F}m_{s}
 \qquad\text{(the length of a shortest accepted word)}
\\
 L'
 &=
 \{w(1,m)^{\Rev}\mid w\in L(\mathcal{A})\}
 \qquad\text{(for $w'$ to be accepted we need $w'(1,m)\in L'$)}
 \\
 L'_{s}
 &=
 \{w(1,\min\{m_{s},m\})^{\Rev}\mid w\in L(\mathcal{A}_{s})\}
 \\
 &
 \qquad\text{(for $w'$ to lead to $s$ we need $w'(1,\min\{m_{s},m\})\in L'_{s}$)}
 \\
 d_1 (w) &= \min
  \{n\in \mathbb{Z}_{>0} \mid \Sigma^* w\Sigma^n\cap L' \neq
 \emptyset \}
 \\
 d_2 (w) &= \min 
 \{n\in \mathbb{Z}_{>0}\mid w\Sigma^n\cap \Sigma^*L' \neq
 \emptyset \}
 \\
  \Delta_2 (s) &= \min_{w \in L'_s} \min\{d_1 (w), d_2
 (w)\}
\end{align*}
For example Fig.~\ref{fig:tableForDelta2InPatternMatching} allows us to
conclude that $d_{2}(\mathrm{dc})=2$, and hence that
$\Delta_{2}(s_{3})=2$. 

Finally we define, for each $C\subseteq S$, 
\begin{displaymath}
 \Delta_2 (C) = \max_{s \in C} \Delta_2 (s)\enspace.
\end{displaymath}
\auxproof{Q: Why is this $\max$? I would suspect this must be $\min$. 

A: This is $\max$. Let $u$ be the exact word the
automaton reads, that is $u = w (i,j)^{\Rev}$. By definition of $C$, for
any $s \in C$, there is a run $s_0,s_1,\cdots,s$ over $u$. So we can
choose $s\in C$ arbitrarily to overapproximate $u$.
}

The other skip value function $\Delta_{1}$ is defined as follows. 
For $a\in \Sigma$ and $p\in \mathbb{Z}_{>0}$,
\begin{displaymath}
\Delta_1 (a,p) = \min 
\{
p
- k
\mid 
k = 0 \text{ or } \exists {w \in L'}.\, w (k) =
a
\}
 \enspace.
\end{displaymath}


\section{Further Details on Algorithm~\ref{naive_alg}}
\label{appendix:futherExplOfNaiveAlgorithm}
Here are some further explanations on Algorithm~\ref{naive_alg}. 
\begin{itemize}
 \item We separately compute: 1) those 
  ``immediately ending''  matching intervals 
such that $w|_{(t,t')}$ is of length $1$ (in which case only the
       terminal character $\$$ would occur); and 2) the others, i.e.\ 
  such that there exists $i\in[1,|w|\,]$ with $\tau_{i}\in (t,t')$. 
 \item 
 Lines~\ref{naive_alg:computingInitBegin}--\ref{naive_alg:computingInitEnd}
   compute a finite set $\Init$ of zones, such that $\bigcup\Init=\{(t,t')\mid \varepsilon|_{(t,t')}\in L
   (\mathcal{A})\}$. This is used as an overapproximation of 
 the ``immediately ending'' matching intervals. 
  Precisely: 
 $\mathcal{M}(w,\mathcal{A})\cap \{(t,t')\mid
  |w|_{(t,t')}|=1\} $ coincides with $(\bigcup \Init) \cap \{(t,t')\mid \forall i\in [1,|w|\,].\,
       \tau_{i}\not\in(t,t')\}$.
 
 \item For computing any other matching interval $(t,t')$---say
       $t\in[\tau_{i-1},\tau_{i})$ and $t'\in
       (\tau_{j},\tau_{j+1}]$---we use the 
       following set.
 \begin{multline*}
  \Conf (i,j) = \bigl\{ (s,\rho,T) \,\bigl|\bigr.\, \forall t_0 \in
  T.\, \exists (\overline{s},\overline{\nu}).\,  
  (\overline{s},\overline{\nu}) \text{ is a run  over } w (i,j) -
  t_0 ,
  \\
   s_{|\overline{s}|-1} = s \text{, and }
  \nu_{|\overline{\nu}|-1} = \eval (\rho,\tau_j,t_0)\bigr\}
 \end{multline*}
 Recall from~(\ref{eq:runOfTimedAutom}) in Def.~\ref{def:runOfTimedAutom} that $s_{|\overline{s}|-1}$ is
       the last element of a sequence $\overline{s}$; similarly for
       $\overline{\nu}$. The intuition of a ``configuration''
       $(s,\rho,T)\in \Conf (i,j)$ is: a nonempty interval $T$ (that is a
       subset of
       $[\tau_{i-1},\tau_{i})$) is the set of epoch times $t$ such that, 
       after reading
       the timed word $w(i,j)$, we end up 
 in a state $s\in S$ and have $\rho$ as the record of clock reset.
 In the algorithm
       we use $\CurrConf$ and $\PrevConf$  for computing $\Conf(i,j)$. 

  More specifically, we initialize $T$ to be the whole
       $[\tau_{i-1},\tau_{i})$ 
  (line~\ref{naive_alg:CurrConfInitialize}). We  gradually narrow it
       down, so that the relevant clock constraint is satisfied
       (line~\ref{naive_alg:update_start_interval}). 
   Here we exploit our definition of clock constraint---we do not allow
       $\lnot$ or $\lor$---to ensure that the new set $T'$ is indeed an
       interval.
  
On lines~\ref{naive_alg:insertTermBegin}--\ref{naive_alg:insertTermEnd} we try to insert the
 terminal character \$ in $(\tau_j,\tau_{j+1}]$. It is also easy to see
       that $\solConstr(T',T'',\rho',\delta')$  on
       line~\ref{naive_alg:end_of_accepting_state} is indeed a zone,
       due to the definition of a clock constraint $\delta'$. 
\end{itemize}

 \section{An Online Variant of Our (Naive) Timed Pattern Matching Algorithm}
 \label{appendix:onlinenNaiveAlgorithm}

 An online variant of our naive algorithm (Alg.~\ref{naive_alg}) is in
 Alg.~\ref{naive_alg_online}.
 Here, unlike in Alg.~\ref{naive_alg},
 we do not divide the match set into ``immediately ending'' ones and others.
 The bottom line in this online algorithm is to compute $\mathcal{M}
 (w,\mathcal{A}) \cap \{(t,t') \mid t' \in (\tau_{i},\tau_{i+1}]\}$
 for each $i$.
 This algorithm is ``online'' because after reading a prefix $w (1,n)$
 of $w$, we have the partial match set
 $\mathcal{M} (w,\mathcal{A}) \cap \{(t,t') \mid t' \in
 (\tau_{0},\tau_{n}]\}$ at the point $n$.
 This algorithm consists of two main parts.
 In the former part---the main loop---we conduct the following
 procedure for each $i \in [1,|w|]$.
 \begin{enumerate}
  \item Lines~\ref{online_alg:insertTermBegin}--\ref{online_alg:insertTermEnd}
   try to insert \$ in $(\tau_{i-1},\tau_{i}]$.
  \item
       Lines~\ref{online_alg:readLoopBegin}--\ref{online_alg:readLoopEnd}
       read $(a_i,\tau_i)$ and build the set $\bigcup_{k\in [1,i]}\Conf
       (k,i)$ of ``configurations''.
 \end{enumerate}
 In the latter part---post-processing---on
 lines~\ref{online_alg:insertTermRemainBegin}--\ref{online_alg:insertTermRemainEnd},
 we
 try to insert \$ in $(\tau_{|w|},\infty)$.

 \begin{algorithm}
  \caption{An online variant of our naive algorithm for timed pattern matching}
  \label{naive_alg_online}
  \scalebox{1.0}{
  \parbox{\textwidth}{
  \begin{algorithmic}[1]
  \Require A timed word $w = (\overline{a},\overline{\tau})$, and a timed
   automaton $\mathcal{A} = (\Sigma,S,S_0,C,E,F)$.
  \Ensure $\bigcup Z$ is the match set $\mathcal{M} (w,\mathcal{A})$ in
  Def.~\ref{TimedPatternMatching}.

   \State $\CurrConf \gets \emptyset;\; Z\gets\emptyset$
   \For{$i \gets 1\, \mathbf{to}\, |w|$}
   \Comment{Here $\CurrConf = \bigcup_{k\in [1,i-1]}\Conf (k,i-1)$.}
   \State $\CurrConf \gets \CurrConf \cup \{(s,\rhoEmpty,[\tau_{i-1},\tau_i)) \mid s \in S_0\}$

   \For{$(s,\rho,T) \in \CurrConf$}\label{online_alg:insertTermBegin}
   \For{$s_f \in F$} \label{naive_alg_online:start_of_accepting_state} 
   \For{$(s,s_f,\$,\lambda,\delta) \in E$} \label{naive_alg_online:loop_over_e_2}
   \State $T' \gets (\tau_{i-1},\tau_{i}]$
   \State $Z \gets Z \cup \solConstr (T,T',\rho,\delta)$
   \label{naive_alg_online:add_interval}
   \EndFor
   \EndFor
   \EndFor\label{online_alg:insertTermEnd}
   \Comment{Lines~\ref{online_alg:insertTermBegin}--\ref{online_alg:insertTermEnd}
   try to insert \$ in $(\tau_{i-1},\tau_{i}]$.}
   
   \State $(\PrevConf , \CurrConf) \gets (\CurrConf , \emptyset)$\par
   \For{$(s,\rho,T) \in \PrevConf$} \label{online_alg:readLoopBegin}
   \For{$(s,s',a_i,\lambda,\delta) \in E$} 
   \Comment{Read $(a_{i},\tau_{i})$.} \label{naive_alg_online:loop_over_e_1}
   \State $T' \gets \{t_0 \in T \mid \eval(\rho,\tau_i,t_0) \models \delta\}$
   \label{naive_alg_online:update_start_interval}     
   \State
   \Comment  {Narrow the interval $T$ to satisfy the clock constraint $\delta$.}
   \If{$T' \neq \emptyset$}
   \State $\rho' \gets \rho$

   \For{$x \in \lambda$}
   \State $\rho' \gets \reset(\rho',x,\tau_i)$ \label{naive_alg_online:reset_clock} \Comment
   {Reset the clock variables in $\lambda$.}
   \EndFor

   \State $\CurrConf \gets \CurrConf \cup (s',\rho',T')$

   \EndIf
   \EndFor
   \EndFor
   \EndFor   \label{online_alg:readLoopEnd}
   \State $\CurrConf \gets \CurrConf \cup \{(s,\rhoEmpty,[\tau_{|w|},\infty)) \mid s \in S_0\}$

   \For{$(s,\rho,T) \in \CurrConf$} \label{online_alg:insertTermRemainBegin}
   \For{$s_f \in F$}
   \For{$(s,s_f,\$,\lambda,\delta) \in E$}
   \State {$T' \gets (\tau_{|w|},\infty)$}
   \State $Z \gets Z \cup \solConstr (T,T',\rho,\delta)$
   \EndFor
   \EndFor
   \EndFor
   \label{online_alg:insertTermRemainEnd}
   \Comment{Lines~\ref{online_alg:insertTermRemainBegin}--\ref{online_alg:insertTermRemainEnd}
   try to insert \$ in $(\tau_{|w|},\infty)$.}
  \end{algorithmic}}}
 \end{algorithm}

\section{Our Boyer-Moore Type Algorithm for Timed Pattern Matching}
\label{appendix:timedBoyerMoore}

See Alg.~\ref{boyer-moore_alg}. 

Its main differences from the naive one (Alg.~\ref{naive_alg}) are: 1)
initially we start with $i=|w|-m+1$ 
(line~\ref{boyer-moore_alg:initialization})
instead of $i=|w|$
(line~\ref{naive_alg:initialization} of Alg.~\ref{naive_alg}); and 2) 
we decrement $i$ by the skip value computed by $\Delta_{2}$
(line~\ref{boyer-moore_alg:decrement_i}), instead of
by $1$ (line~\ref{naive_alg:decrement_i} of Alg.~\ref{naive_alg}).

Compared to Alg.~\ref{naive_alg} we have added
lines~\ref{boyer-moore_alg:additionalForLoopStart}--\ref{boyer-moore_alg:additionalForLoopEnd},
too. This is needed to take care of matching intervals in which no event occurs.

When one wishes to incorporate the other skip value function
$\Delta_{1}$ too, line~\ref{boyer-moore_alg:decrement_i} of
Alg.~\ref{boyer-moore_alg} should be replaced by
\begin{math}
i \leftarrow i-
 \max\{\Delta_1 (a_j,j-i+1),\Delta_2
 (\CurrConf)\}
\end{math}.
\begin{algorithm}
  \caption{Our Boyer-Moore type algorithm for timed pattern matching}
  \label{boyer-moore_alg}
  \scalebox{0.95}{
  \parbox{\textwidth}{
  \begin{algorithmic}[1]
   \Require A timed word $w = (\overline{a},\overline{\tau})$, and a
   timed automaton $\mathcal{A} = (\Sigma,S,S_0,C,E,F)$.
  \Ensure $\bigcup Z$ is the match set $\mathcal{M} (w,\mathcal{A})$ in
   Def.~\ref{TimedPatternMatching}.
   \State $i \gets |w| - m + 1$;\; $\CurrConf \gets \emptyset$;\;
   $\Init\gets\emptyset$;\; $Z\gets\emptyset$
   \label{boyer-moore_alg:initialization}
   \For{$s \in S_0$} \label{boyer-moore_alg:computingInitBegin}
   \For{$s_f \in F$} 
   \For{$(s,s_f,\$,\lambda,\delta) \in E$} 
   \State $\Init \gets \Init \cup \solConstr
   ([0,\infty),(0,\infty),\rhoEmpty,\delta)$
   \label{boyer-moore_alg:computingInitEnd}
   \EndFor
   \EndFor
   \EndFor

   \For{$k = 1$ \textbf{to} $|w|$} \label{boyer-moore_alg:additionalForLoopStart}
   \State $Z \gets Z \cup \{(T_0 \cap
   [\tau_{k-1},\tau_{k}),T_f \cap (\tau_{k-1},\tau_{k}],T_{\Delta})  \mid
   (T_0,T_f,T_\Delta) \in \Init\}$
   \label{boyer-moore_alg:additionalForLoopEnd}
   \EndFor

   \State $Z \gets Z \cup \{(T_0 \cap
   [\tau_{|w|},\infty),T_f \cap (\tau_{|w|},\infty),T_{\Delta}) \mid
   (T_0,T_f,T_\Delta) \in \Init\}$

   \While{$i > 0$}
   \State $j \gets i;$
   \State $\CurrConf \gets \{(s,\rhoEmpty,[\tau_{i-1},\tau_i)) \mid s \in S_0\}$

   \While{$\CurrConf \neq \emptyset \land j \leq |w|$}
   \State $(\PrevConf , \CurrConf) \gets (\CurrConf , \emptyset)$
   \For{$(s,\rho,T) \in \PrevConf$}
   \For{$(s,s',a_j,\lambda,\delta) \in E$} \label{boyer-moore_alg:loop_over_e_1}
   \State $T' \gets \{t_0 \in T \mid
   \delta(\eval(\rho,\tau_j,t_0))\}$
   \label{boyer-moore_alg:update_start_interval} \Comment  {Update available
   start interval.}
   \If{$T' \neq \emptyset$}
   \State $\rho' \gets \rho$

   \For{$x \in \lambda$}
   \State $\rho' \gets \reset(\rho',x,\tau_j)$ \label{boyer-moore_alg:reset_clock} \Comment
   {Reset clock variables in $\lambda$.}
   \EndFor

   \State $\CurrConf \gets \CurrConf \cup (s',\rho',T')$
   \EndIf
   \EndFor
   \EndFor
   
   \For{$(s,\rho,T) \in \CurrConf$}
   \For{$s_f \in F$} \label{boyer-moore_alg:start_of_accepting_state} 
   \For{$(s',s_f,\$,\lambda',\delta') \in E$} \label{boyer-moore_alg:loop_over_e_2}
   \If {$j = |w|$}
   \State {$T'' \gets (\tau_j,\infty)$}
   \Else
   \State $T'' \gets (\tau_j,\tau_{j+1}]$
   \EndIf
   \State $Z \gets Z \cup \solConstr (T',T'',\rho',\delta')$
   \Comment{Solve the linear inequalities.}  
   \EndFor
   \EndFor \label{boyer-moore_alg:end_of_accepting_state}

   \EndFor

   \State $j \gets j + 1$
   \EndWhile
   \If {$j \leq |w|$}
   \State {$\CurrConf \gets \PrevConf$}
   \EndIf
   \State $i \gets i - \Delta_2 (\CurrConf)$ \label{boyer-moore_alg:decrement_i}
   \EndWhile
  \end{algorithmic}}}
 \end{algorithm}

\section{Omitted Proofs}

\subsection{Proof of Thm.~\ref{thm:terminationAndCorrectnessOfNaiveAlg}}
\label{subsec:proofOfThmTerminationAndCorrectnessOfNaiveAlg}
\begin{proof}
For termination, 
it suffices to observe that 
 line~\ref{naive_alg:add_interval} of Alg.~\ref{naive_alg} is  executed no more
than $|w||E|^{{|w|}+1}$ times. 

 For correctness,  assume first that  there is a zone $(T_0,T_f,T_{\Delta}) \in Z$ such that
 $(t,t') \in (T_0,T_f,T_{\Delta})$.
 For $(T_0,T_f,T_{\Delta})$, let $\overline{e}$ be the sequence of
 transitions of $\mathcal{A}$ visited by the algorithm
 (line~\ref{naive_alg:readWordBegin}). 
 This $\overline{e}$ obviously yields an accepting run of $\mathcal{A}$
 over
 $w|_{(t,t')}$: line~\ref{naive_alg:update_start_interval} ensures that
 all the clock constraints in $\overline{e}$ are satisfied; and
 lines~\ref{naive_alg:insertTermBegin}--\ref{naive_alg:insertTermEnd}
 ensure that the run is accepting.


 Conversely, assume that there is an accepting run $(\overline{s},\overline{\nu})$ of
 $\mathcal{A}$ over $w|_{(t,t')}$. We extract a sequence  $\overline{e}$
 of transitions from the run  $(\overline{s},\overline{\nu})$; and 
 let $(T_0,T_f,T_{\Delta})$ be the zone that is added to $Z$ in the
 execution of Alg.~\ref{naive_alg} along the sequence  $\overline{e}$
 (cf.\ line~\ref{naive_alg:readWordBegin}). Then it is straightforward to see
 that $(t,t')\in (T_0,T_f,T_{\Delta})$. 
 \qed
\end{proof}

\subsection{Proof of Thm.~\ref{thm:Delta2_is_skipvalue}}
\label{appendix:thmproof}

We consider a timed automaton
$\mathcal{A} = (\Sigma \amalg \{\$\},S,S_0,C,E,F)$
as the input of timed pattern matching.

\begin{mylemma}
 For any run $r = (\overline{s},\overline{\alpha})$ of the region
 automata $R (\mathcal{A})$, the following equation holds.
 \[
 \pref \bigl( \mathcal{W}(r)\bigr) = \mathcal{W} \bigl( \pref (r) \bigr)
 \]
\end{mylemma}

\begin{proof}
 Let $w \in \pref (\mathcal{W} (r))$ and $w'$ be a timed word such that
 $w \cdot w' \in \mathcal{W} (r)$.
 Let $(\overline{s},\overline{\nu}) \in
 (\overline{s},\overline{\alpha})$ be a run over $w \cdot w'$.
 Since $(\overline{s} (0,|w|),\overline{\nu} (0,|w|)) \in
 (\overline{s} (0,|w|),\overline{\alpha} (0,|w|))$ is a run over
 $w$ over $\mathcal{A}$, $w$ is a member of $\mathcal{W} (\pref (r))$.

 Let $w \in \mathcal{W} (\pref (r))$ and $n$ be an index of $r$ such
 that $w \in \mathcal{W} (\overline{s} (0,n),\overline{\alpha} (0,n))$.
 Let $(\overline{s} (0,n),\overline{\nu})$ be a run of $\mathcal{A}$
 over $w$ such that for any $0 \leq i \leq n$, $\overline{\nu}_i \in
 \overline{\alpha}_i$ holds. 
 Since $(\overline{s},\overline{\alpha})$ is a run of $R
 (\mathcal{A})$, there is a sequence of interpretations
 $\overline{\nu'}$ such that $(\overline{s},\overline{\nu} \cdot
 \overline{\nu'}) \in (\overline{s},\overline{\alpha})$.
 Since $(\overline{s},\overline{\nu} \cdot \overline{\nu'})$ is a run
 of $\mathcal{A}$, there is a timed word $w'$ such that
 $(\overline{s},\overline{\nu} \cdot \overline{\nu'})$ is a run over $w
 \cdot w'$.
 Thus, $w$ is a member of $\pref (\mathcal{W} (r))$.
\end{proof}

\begin{mylemma}
 Let $r$ and $r'$ be runs of $R (\mathcal{A})$.
 For any integer $1 \leq n < |r'|$, the following property holds.
 
 \[
 (\Sigma^n \times (\mathbb{R}_{>0})^n) \cdot \mathcal{W} (r) \cap \mathcal{W}
 (r') \neq \emptyset
 \Rightarrow \mathcal{W} (r) \cap \mathcal{W} (r'(n,|r'|)) \neq \emptyset
 \]
\end{mylemma}

\begin{proof}
 Assume there is a timed word $(\overline{a},\overline{\tau})$ over
 $(\Sigma \amalg \{\$\})$ such that $(\overline{a},\overline{\tau})
 \in (\Sigma \times (\mathbb{R}_{>0}))^n \cdot \mathcal{W} (r) \cap \mathcal{W}
 (r')$.
 By definition of non-absorbing concatenations, $(\overline{a} (n+1,|\overline{a}|),\overline{\tau}
 (n+1,|\overline{\tau}|)) - \tau_{n} \in \mathcal{W} (r)$ holds.
 Since $(\overline{a},\overline{\tau}) \in \mathcal{W} (r')$,
 $(\overline{a}(n+1,|\overline{a}|),\overline{\tau}(n+1,|\overline{\tau}|)) - \tau_n \in \mathcal{W}
 (r'(n,|r'|))$ holds.
 Thus $(\overline{a} (n+1,|\overline{a}|),\overline{\tau} (n+1,|\overline{\tau}|)) - \tau_{n}$ is an
 element of $\mathcal{W} (r) \cap \mathcal{W} (r'(n,|r'|))$.

\end{proof}

We define the \emph{pre-accepted language} $L_{-\$} (\mathcal{A}) = \{w
(1,|w| - 1) \mid w \in L (\mathcal{A})\}$. Since we consider timed
pattern matching, the removed event of $L_{-\$}$ is \$. The language
$L_{-\$}$ represent the pattern without the terminate character.
By definition of $L'$, $L_{-\$} (\mathcal{A}) \subseteq \bigcup_{r \in L'} \mathcal{W} (r)
\cdot (\Sigma \times \mathbb{R}_{>0})^*$ holds.
For each $s \in S$, we define the language of words leading to $s$ as
$L_s = L ((\Sigma \amalg \{\$\},S,S_0,C,E,\{s\}))$.

\begin{mylemma}
 For any $s \in S \setminus F$, we have the following property.
 \[
 L_s \subseteq \bigcup_{r \in L'_s} \mathcal{W} (r) \cdot (\Sigma \times \mathbb{R}_{>0})^*
 \]
\end{mylemma}

\begin{proof}
 Let $w \in L_s$, and $(\overline{s},\overline{\nu})$ be a
 corresponding run of $w$ satisfying $s_{|\overline{s}|-1} = s$.
 We have
 \begin{align*}
  (\overline{s},[\overline{\nu}]) 
  &\in \{r \mid \beta_0 \in \noredundancy{S}_0, \beta
  \in \noredundancy{R} (s), r \in \mathrm{Run}_{\noredundancy{R}
  (\mathcal{A})} (\beta_0,\beta)\}\\
  & \subseteq \{r (0,\min\{m,m'_s\} - 1)\mid \beta_0 \in
  \noredundancy{S}_0, \beta \in \noredundancy{R} (s), r \in
  \mathrm{Run}_{\noredundancy{R} (\mathcal{A})} (\beta_0,\beta)\} \cdot
  (\noredundancy{S})^*\\
  & = L'_s \cdot (\noredundancy{S})^*
 \end{align*}
 where $[\overline{\nu}] = [\nu_0],[\nu_1],\cdots,[\nu_{|\overline{\nu}|-1}]$.
 Thus, $w \in \mathcal{W} (\overline{s},[\overline{\nu}]) \subseteq
 \bigcup_{r \in L'_s} \mathcal{W} (r) \cdot
 (\Sigma \times \mathbb{R}_{>0})^*$ holds.
\end{proof}

For any $i,j$ such that $1 \leq i < j \leq |w|$, 
for any $(s,\rho,T) \in \Conf (i,j)$, we have $s\not\in F$.
The proof of the theorem.~\ref{thm:Delta2_is_skipvalue} is as follows.
%

 \begin{proof}
  \begin{align*}
   & \exists t \in [\tau_{i-n-1},\tau_{i-n}).\, \exists t' \in
   (t,\infty).\,(t,t') \in \mathcal{M} (w,\mathcal{A})\\ 
   &\Leftrightarrow \exists t \in [\tau_{i-n-1},\tau_{i-n}).\, \exists
   t' \in (t,\infty).\, w|_{(t,t')} \in L(\mathcal{A})\\ 
   &\Rightarrow \exists t \in [\tau_{i-n-1},\tau_{i-n}).\, \exists t'
   \in (t,\infty).\, \exists k \in [i-n,|w|].\, (w (i-n,k) - t) \circ
   (\$,t') \in L(\mathcal{A})\\
   &\Rightarrow \exists t \in [\tau_{i-n-1},\tau_{i-n}).\, \exists k \in
   [i-n,|w|].\, (w (i-n,k) - t) \in L_{-\$} (\mathcal{A})\\
   &\Rightarrow \exists t \in [\tau_{i-n-1},\tau_{i-n}).\, (w
   (i-n,|w|) - t) \in L_{- \$} (\mathcal{A}) \cdot (\Sigma \times
   \mathbb{R}_{>0})^*\\
   &\Rightarrow (\Sigma \times \mathbb{R}_{>0})^n \cdot (w (i,|w|) -
   \tau_{i-1}) \cap L_{- \$} (\mathcal{A}) \cdot (\Sigma \times
   \mathbb{R}_{>0})^* \neq \emptyset\\
   &\Rightarrow (\Sigma \times \mathbb{R}_{>0})^n \cdot (w (i,j) -
   \tau_{i-1}) \cdot (\Sigma \times \mathbb{R}_{>0})^* \cap L_{- \$}
   (\mathcal{A}) \cdot (\Sigma \times \mathbb{R}_{>0})^* \neq \emptyset\\ 
   &\Rightarrow\forall (s,\rho,T) \in \Conf (i,j).\, (\Sigma
   \times \mathbb{R}_{>0})^n \cdot L_s \cdot (\Sigma \times \mathbb{R}_{>0})^*
   \cap L_{- \$} (\mathcal{A}) \cdot (\Sigma \times \mathbb{R}_{>0})^* \neq
   \emptyset\\ 
   &\Rightarrow\forall (s,\rho,T) \in \Conf (i,j).\, \bigl( \bigcup_{r \in
   L'_s} (\Sigma \times \mathbb{R}_{>0})^n \cdot \mathcal{W}(r) \cdot
   (\Sigma \times \mathbb{R}_{>0})^* \bigr) \cap
   \bigl(\bigcup_{r' \in L'} \mathcal{W} (r') \cdot
   (\Sigma \times \mathbb{R}_{>0})^*\bigr) \neq \emptyset\\   
   &\Leftrightarrow \forall (s,\rho,T) \in \Conf(i,j).\, \bigl(
   \bigcup_{r \in L'_s} (\Sigma \times \mathbb{R}_{>0})^n \cdot \mathcal{W}(r)
   \cdot (\Sigma \times \mathbb{R}_{>0})^*\bigr)  \cap
   \bigcup_{r' \in L'} \mathcal{W}(r') \neq \emptyset \\
   &\lor \bigl( \bigcup_{r \in L'_s} (\Sigma \times \mathbb{R}_{>0})^n
   \cdot \mathcal{W}(r) \bigr) \cap
   \bigl( \bigcup_{r' \in L'} \mathcal{W}(r') \cdot 
   (\Sigma \times \mathbb{R}_{>0})^* \bigr) \neq \emptyset\\
   &\Leftrightarrow \forall (s,\rho,T) \in \Conf(i,j).\,
   \bigl( \bigcup_{r \in L'_s}
   (\Sigma \times \mathbb{R}_{>0})^n \cdot \mathcal{W}(r)\bigr) \cap
   \bigl(
   \bigcup_{r' \in L'} \bigcup_{r'' \in \pref(r')} \mathcal{W}(r'')
   \bigr)
   \neq \emptyset \\
   &\lor 
   \bigl(
   \bigcup_{r \in L'_s} \pref((\Sigma \times \mathbb{R}_{>0})^n \cdot \mathcal{W}(r))
   \bigr)
   \cap
   \bigcup_{r' \in L'} \mathcal{W}(r')
   \neq \emptyset
  \end{align*}

 Thus,
  \begin{align*}
   &\Opt (i) = \min \{n \in \mathbb{Z}_{>0} \mid \exists t \in
   [\tau_{i-n-1},\tau_{i-n}), t' \in (t,\infty).\,(t,t') \in \mathcal{M}
   (w,\mathcal{A})\}\\
   &\ge\min\bigl\{ n \in \mathbb{Z}_{>0}\,\bigl|\bigr.\,
   \forall (s,\rho,T) \in \Conf(i,j).\\
   &\bigl( \bigcup_{r \in L'_s}
   (\Sigma \times \mathbb{R}_{>0})^n \cdot \mathcal{W}(r)\bigr) \cap
   \bigl(
   \bigcup_{r' \in L'} \bigcup_{r'' \in \pref(r')} \mathcal{W}(r'')
   \bigr)
   \neq \emptyset \\
   &\lor 
   \bigl(
   \bigcup_{r \in L'_s} \pref((\Sigma \times \mathbb{R}_{>0})^n \cdot \mathcal{W}(r))
   \bigr)
   \cap \bigcup_{r' \in L'} \mathcal{W}(r')
   \neq \emptyset\bigr\}\\
   &\text{\{Such $n$ is not more than $m$\}}\\
   &=\min\bigl\{ n \in \mathbb{Z}_{>0}\,\bigl|\bigr.\,
   \forall (s,\rho,T) \in \Conf(i,j).\\
   &\bigl( \bigcup_{r \in L'_s}
   (\Sigma \times \mathbb{R}_{>0})^n \cdot \mathcal{W}(r)\bigr) \cap
   \bigl(
   \bigcup_{r' \in L'} \bigcup_{r'' \in \pref(r')} \mathcal{W}(r'')
   \bigr)
   \neq \emptyset \\
   &\lor 
   \bigl(
   \bigcup_{r \in L'_s} (\Sigma \times \mathbb{R}_{>0})^n \cdot \pref(\mathcal{W}(r))
   \bigr)
   \cap
   \bigcup_{r' \in L'} \mathcal{W}(r')
   \neq \emptyset\bigr\}\\
   &=\min\Bigl\{ n \in \mathbb{Z}_{>0}\,\bigl|\bigr.\,
   \forall (s,\rho,T) \in \Conf(i,j).\\
   &\bigl( \bigcup_{r \in L'_s}
   (\Sigma \times \mathbb{R}_{>0})^n \cdot \mathcal{W}(r)\bigr) \cap
   \bigl(
   \bigcup_{r' \in L'} \bigcup_{r'' \in \pref(r')} \mathcal{W}(r'')
   \bigr)
   \neq \emptyset \\
   &\lor 
   \bigl(
   \bigcup_{r \in L'_s} \bigcup_{r'' \in \pref (r)}(\Sigma \times \mathbb{R}_{>0})^n \cdot \mathcal{W}(r'')
   \bigr)
   \cap
   \bigcup_{r' \in L'} \mathcal{W}(r')
   \neq \emptyset\Bigr\}\\  
   &\text{\{$\Conf (i,j)$ is finite.\}}\\
   &=\max_{(s,\rho,T) \in \Conf (i,j)}
   \min\Bigl\{ n \in \mathbb{Z}_{>0}\,\bigl|\bigr.\\
   &\bigl( \bigcup_{r \in L'_s}
   (\Sigma \times \mathbb{R}_{>0})^n \cdot \mathcal{W}(r)\bigr) \cap
   \bigl(
   \bigcup_{r' \in L'} \bigcup_{r'' \in \pref(r')} \mathcal{W}(r'')
   \bigr)
   \neq \emptyset \\
   &\lor 
   \bigl(
   \bigcup_{r \in L'_s} \bigcup_{r'' \in \pref (r)}(\Sigma \times \mathbb{R}_{>0})^n \cdot \mathcal{W}(r'')
   \bigr)
   \cap
   \bigcup_{r' \in L'} \mathcal{W}(r')
   \neq \emptyset\Bigr\}\\
   &\ge\max_{(s,\rho,T) \in \Conf (i,j)}
   \min\bigl\{ n \in \mathbb{Z}_{>0}\,\bigl|\bigr.\\
   &\bigl( \bigcup_{r \in L'_s}
   \mathcal{W}(r)\bigr) \cap
   \bigl(
   \bigcup_{r' \in L'}
   \bigcup_{r'' \in \pref(r'(n,|r'|))}
   \mathcal{W}(r'')
   \bigr)
   \neq \emptyset \\
   &\lor 
   \bigl(
   \bigcup_{r \in L'_s} \bigcup_{r'' \in \pref (r)} \mathcal{W}(r'')
   \bigr)
   \cap
   \bigl(
   \bigcup_{r' \in L'} \mathcal{W}(r'(n,|r'|))
   \bigr)
   \neq \emptyset\bigr\}\\
   &=\max_{(s,\rho,T) \in \Conf (i,j)} \min_{r \in L'_s}\min\bigl\{
   \min_{r' \in L'} \min \bigl\{ n \in \mathbb{Z}_{>0} \mid
   \mathcal{W}(r) \cap
   \bigl(\,\bigcup_{r''\in \pref \bigl(r'(n,|r'|)\bigr)}\mathcal{W}(r'') \,\bigr)
   \neq \emptyset\,\bigr\},\\
   &\min_{r' \in L'} \min \bigl\{ n \in \mathbb{Z}_{>0} \mid
   \bigl(\bigcup_{r''\in \pref(r)}\mathcal{W}(r'')\bigr)
   \cap \mathcal{W}\bigl(r'(n,|r'|)\bigr) \neq \emptyset\,\bigr\}\bigr\}\\
   &=\max_{(s,\rho,T) \in \Conf (i,j)} \min_{r \in
   L'_s}\min\{d_1 (r),d_2 (r)\}\\
   &=\Delta_2(\Conf (i,j))
  \end{align*}
 \end{proof}

\subsection{Proof of Prop.~\ref{prop:productAutomAndPartialRun}}
\begin{proof}
 Assume there is a timed word $(\overline{a},\overline{\tau})$ such that $(\overline{a},\overline{\tau})
 \in \mathcal{W} (r) \cap \mathcal{W} (r')$.
 Let $(\overline{s},\overline{\nu}) \in r$ be a partial run of
 $\mathcal{A}$ over $(\overline{a},\overline{\tau})$ and $(\overline{s}',\overline{\nu}')
 \in r'$ be a partial run of $\mathcal{A}'$ over $(\overline{a},\overline{\tau})$.
 Since $((\overline{s},\overline{s}')
 ,(\overline{\nu},\overline{\nu}')) \in (r,r')$ is a partial run over
 $(\overline{a},\overline{\tau})$ of $\mathcal{A} \times \mathcal{A}'$, $(\overline{a},\overline{\tau})$
 is a member of $\mathcal{W} ((r,r'))$.
 
 Assume there is a timed word $(\overline{a},\overline{\tau})$ such that $(\overline{a},\overline{\tau})
 \in \mathcal{W} ((r,r'))$.
 Let $((\overline{s},\overline{s}'),(\overline{\nu},\overline{\nu}'))
 \in (r,r')$ be a partial run over $(\overline{a},\overline{\tau})$ of
 $\mathcal{A} \times \mathcal{A}'$ where $\overline{\nu}$ is a clock
 interpretation over $C$ and $\overline{\nu}'$ is a clock interpretation
 over $C'$.
 Since $(\overline{s},\overline{\nu}) \in r$ is a partial run over
 $(\overline{a},\overline{\tau})$ of $\mathcal{A}$ and $(\overline{s}',\overline{\nu}')
 \in r'$ is a partial run over $(\overline{a},\overline{\tau})$ over $\mathcal{A}'$,
 $(\overline{a},\overline{\tau}) \in \mathcal{W} (r) \cap \mathcal{W} (r')$. 
\end{proof}

\section{Detailed Results of The Experiments}
\label{appendix:detailedResultsOfOurExperiments}

The detailed results of our experiments are in the following tables.

 \begin{table}[ht]
    \centering
    \caption{Execution time in Case~1.}
    \label{121717_21Jan16}
    \scalebox{0.80}{  
    \drawexectable{0-0622112729.dat}
    }
\end{table}
\begin{table}[ht]
 \centering
 \caption{Execution time in Case~2.}
 \label{141859_16Feb16}
 \scalebox{1.0}{  
 \drawexectable{2-0622112729.dat}
 }  
\end{table}

     \begin{table}[ht]
      \centering
      \caption{Execution time in Case~3.}
      \label{131638_23Jan16}
      \scalebox{0.90}{
      \begin{filecontents*}{4-0622123306.dat}
0.2 40000 8028 0.827519 0.539559 10.2427 54
0.2 80000 15958 1.67915 0.84064 11.0413 58
0.2 160000 31826 2.67663 1.35022 10.8903 105
0.2 320000 64026 5.84096 1.91475 9.05861 331
0.2 3200000 639762 49.2627 14.5883 20.4257 2410
0.2 6400000 1279143 98.5113 28.7516 34.5348 5012
0.2 12800000 2559898 198.452 58.1035 64.0617 10602
0.3 40000 11898 1.5338 0.813435 11.8069 124
0.3 80000 23739 3.32009 1.3443 12.1843 176
0.3 160000 48202 4.74014 1.70719 9.0686 370
0.3 320000 96195 7.74485 2.71505 8.72694 948
0.3 3200000 960292 75.4215 25.9684 31.9326 9267
0.3 6400000 1921414 151.249 50.6073 56.4234 18385
0.3 12800000 3836454 303.692 101.044 106.873 36788
0.5 40000 19991 3.12862 1.3795 10.921 532
0.5 80000 40107 4.88858 2.17899 10.6741 866
0.5 160000 79965 7.71346 3.4653 10.5314 1678
0.5 320000 159771 13.1883 5.8147 11.7207 3391
0.5 3200000 1599831 136.567 56.2486 62.0433 35275
0.5 6400000 3200371 268.471 112.531 118.349 72153
0.5 12800000 6400401 534.351 223.59 229.385 142868
0.6 40000 23731 2.82768 1.78046 11.6672 663
0.6 80000 47953 5.42489 2.83688 11.0836 1280
0.6 160000 96414 9.12626 4.31868 10.8699 2686
0.6 320000 192552 16.5027 7.81413 13.6436 5310
0.6 3200000 1920929 163.805 75.1476 81.0417 53044
0.6 6400000 3838606 325.402 149.194 155.007 106409
0.6 12800000 7684548 653.871 298.207 304.021 213628
0.8 40000 32310 3.58176 2.79615 12.844 1099
0.8 80000 64368 7.33818 3.85823 10.8763 2207
0.8 160000 127930 11.6965 6.28094 12.1505 4488
0.8 320000 255611 22.7692 12.2156 17.9983 8331
0.8 3200000 2560634 231.753 120.683 126.547 86511
0.8 6400000 5121063 448.659 237.661 243.535 172821
0.8 12800000 10243600 892.602 478.181 484.017 344973
1.0 40000 40410 4.77891 3.30025 12.2458 1112
1.0 80000 79878 7.57968 5.58467 13.6088 2137
1.0 160000 160213 13.2764 7.76643 13.8974 2683
1.0 320000 319948 22.8345 13.5578 19.3477 2793
1.2 40000 48271 5.34024 3.22914 9.92567 1224
1.2 80000 95962 9.83505 5.84512 12.3936 2313
1.2 160000 191733 14.5875 9.77073 15.7234 2947
1.2 320000 384298 26.799 18.1099 23.9451 3070
1.5 40000 59766 6.44054 5.17391 13.3599 1157
1.5 80000 119910 11.0988 8.36269 15.0771 2241
1.5 160000 239409 18.0662 13.8013 19.7273 2804
2.0 40000 79898 9.11011 7.25933 15.2129 846
2.0 80000 159319 13.4021 10.8773 16.914 1610
2.0 160000 321191 23.3822 19.0378 24.9406 1964
\end{filecontents*}
\pgfplotstabletypeset[
      multicolumn names, 
      display columns/0/.style={
      fixed,fixed zerofill,precision=2,
      column name=$\lambda$, 
      },
      display columns/1/.style={
      sci,
      sci zerofill,
      column name=$\tau_{|\overline{\tau}|}$},
      display columns/2/.style={
      fixed,fixed zerofill,precision=0,
      column name=$|w|$},
      display columns/3/.style={
      fixed,fixed zerofill,precision=2,
      column name=Naive Time\,(ms.)},
      display columns/4/.style={
      fixed,fixed zerofill,precision=2,
      column name=BM Time\,(ms.)},
      display columns/5/.style={
      sci,
      sci zerofill,
      column name=BM + Preproc. Time\,(ms.)},
      display columns/6/.style={
      column name=$|Z|$},
      every head row/.style={
      before row={\toprule}, 
      after row={\midrule} 
      },
      every last row/.style={after row=\bottomrule}, 
      ]{4-0622123306.dat} 
      }
     \end{table}

\begin{table}[ht]
   \centering
   \caption{Execution time in Case~4.}
   \label{162108_21Jan16}
   \scalebox{1.00}{
   \begin{filecontents*}{3-0622123253.dat}
0.025 40000 1934 1.90355 2137
0.025 80000 3987 4.52757 4417
0.025 160000 7978 6.38569 8848
0.05 40000 4007 2.6116 4754
0.05 80000 8139 2.65035 9583
0.05 160000 15934 7.99909 18875
0.075 40000 6110 2.53761 7699
0.075 80000 12158 5.19822 15195
0.075 160000 23986 8.82773 29744
0.1 40000 7947 5.9676 10410
0.1 80000 15867 6.54337 20822
0.1 160000 31935 16.7255 41414
\end{filecontents*}
\pgfplotstabletypeset[
   sci,
   sci zerofill,
   multicolumn names, 
   display columns/0/.style={
   column name=$\lambda$, 
   fixed,fixed zerofill,precision=2,
   },
   display columns/1/.style={
   column name=$\tau_{|\overline{\tau}|}$},
   display columns/2/.style={
   fixed,fixed zerofill,precision=0,
   column name=$|w|$},
   display columns/3/.style={
   fixed,fixed zerofill,precision=2,
   column name=Naive Time\,(ms.)},
   display columns/4/.style={
   fixed,fixed zerofill,precision=0,
   column name=$|Z|$},
   fixed,fixed zerofill,precision=0,
   every head row/.style={
   before row={\toprule}, 
   after row={\midrule} 
   },
   every last row/.style={after row=\bottomrule}, 
   ]{3-0622123253.dat} 
   }
\end{table}
  
\begin{table}[h]
   \centering
   \caption{Execution time in Case~5.}
   \label{201607_13Feb16}
   \scalebox{1.0}{
   \drawexectable{6-0622112729.dat}
   }
\end{table}

\auxproof{
\section{Check list}
\begin{itemize}
 \item Technical question: What does the additional expressivity
       of timed automata (over
       TREs) buy us?
       \begin{itemize}
	\item Fig.~12 of the thesis
	\item A possible scenario: two anonymous
	      processes sending
	      messages with the interval of
	      1 second (they execute the
	      same action because they are anonymous)
	\item We will not emphasize this point, unless we come up 
	      with a very nice example e.g.\ through discussions with
	      the Toyota people
	\item To-do: \emph{formal arguments}. Let 
	      \begin{align*}
	       L_{0}
	       \;=\;
	       &
	       \bigl\{\,
	       (aa\dotsc a, \tau_{1}\tau_{2}\dotsc \tau_{2m})
	       \,\bigr|\,
	       m\in\mathbb{N}, 
	       	       \tau_{1}\in[0,1],
	       \tau_{2}\in[\tau_{1},1],
	       \\
	       \bigr.\,
	       &\qquad
	       \tau_{3}=\tau_{1}+1,
	       \tau_{4}=\tau_{2}+1,
	       \dotsc,
       	       \tau_{2m-1}=\tau_{2m-3}+1,
       	       \tau_{2m}=\tau_{2m-2}+1
	       \,\bigr\}
	      \end{align*}
	      \begin{itemize}
	       \item Construct a timed automaton $\mathcal{A}_{0}$ such
		     that $L(\mathcal{A}_{0})=L_{0}$. 
	       \item Show that there does not exist any TRE $r$ such
		     that
		     $L(r)=L_{0}$. 
	      \end{itemize}
       \end{itemize}
 \item Main contribution: Boyer-Moore
       \begin{itemize}
	\item Twice as fast. How good is it? Three possible answers.
	      \begin{itemize}
	       \item Find industrial case study and identify the size of
		     typical problem instances. 
	       \item Scan RV papers and see how they justify their
		     results.
	       \item Ask industry people. 
	      \end{itemize}
	\item Reviewed articles:
	      \begin{itemize}
	       \item There are not too many papers in RV that emphasize
		     efficiency, especially for off-line algorithms
	       \item \cite{DBLP:conf/rv/ColomboP12} gives us a good
		     sense on the scale of real-world monitoring
		     problems, but not directly related to the current
		     work. 
	      \end{itemize}
       \end{itemize}
 \item Why timed automata in place of TREs? Two reasons:
       \begin{itemize}
	\item Using TAs allows us to use the Boyer-Moore type optimization
	\item Additionally: TAs are more expressive (although the value
	      of the additional expressivity is not very clear yet)
       \end{itemize}
 \item Comparison with works on timed pattern matching with TREs
       \begin{itemize}
	\item Three papers. Explain what each of them does. Compare with
	      our current work.
	      \begin{itemize}
	       \item FORMATS'14: The first algorithm for the timed
                     pattern matching.  It is offline and inductive.
	       \item CAV'15: Application of timed pattern matching for
                     measuring performance evaluation. The pattern TRE
                     represents the segments to measure. 
	       \item TACAS'16: An online algorithm for the timed pattern
                     matching based on the derivative similar to the
                     Brzozowski derivative.
	      \end{itemize}
       \end{itemize}
 \item Some additional results:
       \begin{itemize}
	\item Theorem: the result of timed pattern matching is described by a
       finite union of zones
	\item The online version of the naive algorithm
       \end{itemize}
 \item Sections
       \begin{itemize}
	\item Introduction
	      \begin{itemize}
	       \item Give the feeling of how the BM algorithm works with a
			    small example
	      \end{itemize}
	\item Preliminaries      
	      \begin{itemize}
	       \item Timed Pattern Matching
		     \begin{itemize}
		      \item Timed Words
			    \begin{itemize}
			     \item Comparison with signals
			    \end{itemize}
		      \item Timed Automata
			    \begin{itemize}
			     \item Briefly mention TREs, and that TAs are more expressive
			    \end{itemize}
		     \end{itemize}
	       \item The Boyer-Moore algorithm
		     \begin{itemize}
		      \item Generalized Boyer-Moore, not formally but
			    intuitively, preferably in parallel to the
			    example in the introduction
		     \end{itemize}
	      \end{itemize}       
	\item The Naive Algorithm
	      \begin{itemize}
	       \item Theorem: the answer is representable
	       \item Its online variant
	      \end{itemize}
	\item A Timed Boyer-Moore Algorithm
	\item Experiments
       \end{itemize}
 \item Algorithms
       \begin{itemize}
	\item The original Boyer-Moore algorithm (for string matching)
	\item The ``generalized'' Boyer-Moore type algorithm
	      in~\cite{Watson2003} for pattern matching (where a pattern
	      is given by an NFA)
	      \begin{itemize}
	       \item In short: the Boyer-Moore type algorithm
		     in~\cite{Watson2003} for pattern matching
	      \end{itemize}
	\item Our naive algorithm for timed pattern matching
	\item A online variant of our naive  algorithm for timed pattern
	      matching
	      \begin{itemize}
	       \item In short: our online (naive) algorithm for timed pattern
	      matching
	      \end{itemize}
	\item Our Boyer-Moore type algorithm for timed pattern matching
       \end{itemize}
\end{itemize}
}

\auxproof{
\section{Our Notations}
\begin{itemize}
 \item {\color{red}\xmark} $\mathbb{N}$ and $\mathbb{Z_{+}}$.
 \item {\color{dgreen}\cmark} $\mathbb{Z}$, $\mathbb{Z}_{>0}$, and $\mathbb{Z}_{\geq
       0}$.
 \item {\color{dgreen}\cmark}   $\overline{\sigma} =
       \sigma_1,\sigma_2,\cdots,\sigma_n$, and $\sigma \in \Sigma$.
 \item {\color{red}\xmark} $\overline{a} =
       a_1,a_2,\cdots,a_n$, and $a \in \Sigma$.
\end{itemize}
}

\end{document}